\documentclass[a4paper,twoside,11pt,english,intlimits]{article}

\usepackage[utf8]{inputenc}
\usepackage[T2A,T1]{fontenc}
\DeclareSymbolFont{cyrillic}{T2A}{cmr}{m}{n}
\DeclareMathSymbol{\Sha}{\mathalpha}{cyrillic}{216}

\usepackage[english]{babel}

\usepackage{amsmath,amsthm,amssymb,stmaryrd}
\usepackage{mathtools}
\usepackage[ttscale=.875]{libertine}
\usepackage[libertine]{newtxmath}

\usepackage{isabelle,isabellesym}
\isabellestyle{it}

\usepackage{setspace}

\usepackage{enumerate}
\usepackage{fancyhdr,titlesec,url}
\usepackage[all,knot,poly]{xy}
\usepackage{array}
\usepackage{hyperref}
\usepackage{graphicx}
\usepackage[numbers,sort&compress]{natbib}
\usepackage{multirow}
\usepackage{ifthen}
\usepackage{time}
\usepackage{algorithm2e}
\usepackage[capitalize,noabbrev]{cleveref}
\usepackage{etoolbox}

\usepackage{cancel}
\usepackage{tikz}
\usepackage{tikz-cd}
\usepackage{xcolor}
\usetikzlibrary{matrix,arrows,arrows.meta,decorations.markings,decorations.pathmorphing}

\usepackage{listings}
\usepackage{color}
\definecolor{keywordcolor}{rgb}{0.7, 0.1, 0.1}   
\definecolor{tacticcolor}{rgb}{0.0, 0.1, 0.6}    
\definecolor{commentcolor}{rgb}{0.4, 0.4, 0.4}   
\definecolor{symbolcolor}{rgb}{0.0, 0.1, 0.6}    
\definecolor{sortcolor}{rgb}{0.1, 0.5, 0.1}      
\definecolor{attributecolor}{rgb}{0.7, 0.1, 0.1} 

\newcommand{\periodafter}[1]{\ifstrempty{#1}{}{#1.}}
\titleformat{\section}[block]{\scshape\filcenter\LARGE\boldmath}{\thesection.}{.5em}{}
\titleformat{\subsection}[block]{\bfseries\filcenter\large\boldmath}{\thesubsection.}{.5em}{\medskip}
\titleformat{\subsubsection}[runin]{\bfseries\boldmath}{\thesubsubsection.}{.5em}{\periodafter}
\titlespacing{\subsubsection}{0pt}{\topsep}{.5em}

\newtheoremstyle{ntheorem}%
	{\topsep}{\topsep}{\itshape}{0pt}{\bfseries}{.}{.5em}%
	{\thmnumber{#2.\hspace{.5em}}\thmname{#1}\thmnote{ (#3)}}
	
\newtheoremstyle{ndefinition}%
	{\topsep}{\topsep}{\normalfont}{0pt}{\bfseries}{.}{.5em}%
	{\thmnumber{#2.\hspace{.5em}}\thmname{#1}\thmnote{ (#3)}}
	
\newtheoremstyle{nremark}%
	{\topsep}{\topsep}{\normalfont}{0pt}{\itshape}{.}{.5em}%
	{\thmnumber{}\thmname{#1}\thmnote{ (#3)}}

\theoremstyle{ntheorem}
  	\newtheorem{theorem}[subsubsection]{Theorem}
  	\newtheorem{proposition}[subsubsection]{Proposition}
	\newtheorem{lemma}[subsubsection]{Lemma}

\theoremstyle{ndefinition}

	\newtheorem{example}[subsubsection]{Example}
	\newtheorem{remark}[subsubsection]{Remark}
	
\makeatletter
\def\@equationname{equation}
\newenvironment{eqn}[1]{%
    \def\mymathenvironmenttouse{#1}%
    \ifx\mymathenvironmenttouse\@equationname%
        \refstepcounter{subsubsection}%
    \else
        \patchcmd{\@arrayparboxrestore}{equation}{subsubsection}{}{}
        \patchcmd{\print@eqnum}{equation}{subsubsection}{}{}%
        \patchcmd{\incr@eqnum}{equation}{subsubsection}{}{}%
    \fi
    \csname\mymathenvironmenttouse\endcsname%
}{%
    \ifx\mymathenvironmenttouse\@equationname%
        \tag{\thesubsubsection}%
    \fi
    \csname end\mymathenvironmenttouse\endcsname%
}
\makeatother

\fancypagestyle{plain}{
\fancyhf{}\fancyfoot[LC,RC]{}
\fancyfoot[LE,RO]{$\thepage$}

}

\UseTips
\SelectTips{eu}{11}

\newdir{ >}{{}*!/-10pt/@{>}}
\newdir{ -}{{}*!/-10pt/@{}}
\newdir{> }{{}*!/+10pt/@{>}}

\makeatletter

\xyletcsnamecsname@{dir4{}}{dir{}}
\xydefcsname@{dir4{-}}{\line@ \quadruple@\xydashh@}
\xydefcsname@{dir4{.}}{\point@ \quadruple@\xydashh@}
\xydefcsname@{dir4{~}}{\squiggle@ \quadruple@\xybsqlh@}
\xydefcsname@{dir4{>}}{\Tttip@}
\xydefcsname@{dir4{<}}{\reverseDirection@\Tttip@}

\xydef@\quadruple@#1{%
	\edef\Drop@@{%
		\dimen@=#1\relax
		\dimen@=.5\dimen@
		\A@=-\sinDirection\dimen@
		\B@=\cosDirection\dimen@
		\setboxz@h{%
			\setbox2=\hbox{\kern3\A@\raise3\B@\copy\z@}%
			\dp2=\z@ \ht2=\z@ \wd2=\z@ \box2
			\setbox2=\hbox{\kern\A@\raise\B@\copy\z@}%
			\dp2=\z@ \ht2=\z@ \wd2=\z@ \box2
			\setbox2=\hbox{\kern-\A@\raise-\B@\copy\z@}%
			\dp2=\z@ \ht2=\z@ \wd2=\z@ \box2
			\setbox2=\hbox{\kern-3\A@\raise-3\B@ \noexpand\boxz@}%
			\dp2=\z@ \ht2=\z@ \wd2=\z@ \box2
		}%
		\ht\z@=\z@ \dp\z@=\z@ \wd\z@=\z@ \noexpand\styledboxz@
	}%
}

\xydef@\Tttip@{\kern2pt \vrule height2pt depth2pt width\z@
	\Tttip@@ \kern2pt \egroup
	\U@c=0pt \D@c=0pt \L@c=0pt \R@c=0pt \Edge@c={\circleEdge}%
	\def\Leftness@{.5}\def\Upness@{.5}%
	\def\Drop@@{\styledboxz@}\def\Connect@@{\straight@{\dottedSpread@\jot}}}
	
\xydef@\Tttip@@{%
	\dimen@=.25\dimen@
 	\B@=\cosDirection\dimen@
	\setboxz@h\bgroup\reverseDirection@\line@ \wdz@=\z@ \ht\z@=\z@ \dp\z@=\z@
	{\vDirection@(1,-1)\xydashl@ \xyatipfont\char\DirectionChar}%
	{\vDirection@(1,+1)\xydashl@ \xybtipfont\char\DirectionChar}%
}

\xydef@\ar@form{
	\ifx \space@\next \expandafter\DN@\space{\xyFN@\ar@form}%
	\else\ifx ^\next \DN@ ^{\xyFN@\ar@style}\edef\arvariant@@{\string^}%
	\else\ifx _\next \DN@ _{\xyFN@\ar@style}\edef\arvariant@@{\string_}%
	\else\ifx 0\next \DN@ 0{\xyFN@\ar@style}\def\arvariant@@{0}%
	\else\ifx 1\next \DN@ 1{\xyFN@\ar@style}\def\arvariant@@{1}%
	\else\ifx 2\next \DN@ 2{\xyFN@\ar@style}\def\arvariant@@{2}%
	\else\ifx 3\next \DN@ 3{\xyFN@\ar@style}\def\arvariant@@{3}%
	\else\ifx 4\next \DN@ 4{\xyFN@\ar@style}\def\arvariant@@{4}%
	\else\ifx \bgroup\next \let\next@=\ar@style
	\else\ifx [\next \DN@[##1]{\ar@modifiers{[##1]}}
	\else\ifx *\next \DN@ *{\ar@modifiers}%
	\else\addLT@\ifx\next \let\next@=\ar@slide
	\else\ifx /\next \let\next@=\ar@curveslash
	\else\ifx (\next \let\next@=\ar@curveinout 
	\else\addRQ@\ifx\next \addRQ@\DN@{\ar@curve@}%
	\else\addLQ@\ifx\next \addLQ@\DN@{\xyFN@\ar@curve}%
	\else\addDASH@\ifx\next \addDASH@\DN@{\defarstem@-\xyFN@\ar@}%
	\else\addEQ@\ifx\next \addEQ@\DN@{\def\arvariant@@{2}\defarstem@-\xyFN@\ar@}%
	\else\addDOT@\ifx\next \addDOT@\DN@{\defarstem@.\xyFN@\ar@}%
	\else\ifx :\next \DN@:{\def\arvariant@@{2}\defarstem@.\xyFN@\ar@}%
	\else\ifx ~\next \DN@~{\defarstem@~\xyFN@\ar@}%
	\else\ifx !\next \DN@!{\dasharstem@\xyFN@\ar@}%
	\else\ifx ?\next \DN@?{\ar@upsidedown\xyFN@\ar@}%
	\else \let\next@=\ar@error
	\fi\fi\fi\fi\fi\fi\fi\fi\fi\fi\fi\fi\fi\fi\fi\fi\fi\fi\fi\fi\fi\fi\fi \next@}

\makeatother

\newcommand{\fl}{\rightarrow}

\newcommand{\qfl}{\xymatrix@1@C=10pt{\ar@4 [r] &}}

\renewcommand{\angle}[1]{\langle #1 \rangle}

\renewcommand{\tilde}[1]{\widetilde{#1}}

\DeclareMathOperator{\id}{Id}

\renewcommand{\phi}{\varphi}
\renewcommand{\epsilon}{\varepsilon}

\renewcommand{\S}{\mathbf{S}}

\def\catego#1{\mathsf{#1}}

\newcommand{\ifthen}[2]{\ifthenelse{#1}{#2}{}}

\DeclareMathOperator{\colim}{colim}

\DeclareMathOperator{\Nbb}{\mathbb{N}}

\DeclareMathOperator{\Ccal}{\mathcal{C}}
\DeclareMathOperator{\Dcal}{\mathcal{D}}

\DeclareMathOperator{\Scal}{\mathcal{S}}

\newcommand{\oto}[1]{\overset{#1}{\to}}
\newcommand{\ofrom}[1]{\overset{#1}{\leftarrow}}

\theoremstyle{ntheorem}

\tikzset{global scale/.style args={#1and#2}{scale=#1, every node/.append style={scale=#2}}}
\tikzcdset{global scale/.style args={#1and#2and#3and#4}{column sep={#1em,between origins}, row sep={#2em,between origins}, every label/.append style={scale=#4}, cells={nodes={scale=#3}}, nodes in empty cells}}
\tikzcdset{longer arrows/.style args={#1and#2}{every arrow/.append style={shorten <= -#1, shorten >= -#2}}}
\tikzset{triangle/.style = {path picture={\draw (path picture bounding box.north west) -- (path picture bounding box.north east) -- (path picture bounding box.south) -- (path picture bounding box.north west);}}}
\tikzset{Rightarrow/.style={double equal sign distance,>={Implies},->}, Rrightarrow/.style={-,preaction={draw,Rightarrow}}, Rrrightarrow/.style={preaction={draw,Rightarrow},-,double,double distance=0.2pt}}

\makeatletter
\providecommand{\leftsquigarrow}{%
  \mathrel{\mathpalette\reflect@squig\relax}%
}
\newcommand{\reflect@squig}[2]{%
  \reflectbox{$\m@th#1\rightsquigarrow$}%
}
\makeatother

\DeclareFontFamily{U}{mathx}{\hyphenchar\font45}
\DeclareFontShape{U}{mathx}{m}{n}{
      <5> <6> <7> <8> <9> <10>
      <10.95> <12> <14.4> <17.28> <20.74> <24.88>
      mathx10
      }{}
\DeclareSymbolFont{mathx}{U}{mathx}{m}{n}
\DeclareFontSubstitution{U}{mathx}{m}{n}
\DeclareMathAccent{\widecheck}{0}{mathx}{"71}

\renewcommand{\id}{id}

\newcommand{\FC}[1]{{#1}^{c}}
\newcommand{\FS}[1]{{#1}^{s}}
\newcommand{\FCG}[1]{{#1}^{c \Gamma}}
\newcommand{\FSG}[1]{{#1}^{s \gamma}}

\renewcommand{\leq}{\leqslant}
\renewcommand{\geq}{\geqslant}

\newcount\hh
\newcount\mm
\mm=\time
\hh=\time
\divide\hh by 60
\divide\mm by 60
\multiply\mm by 60
\mm=-\mm
\advance\mm by \time
\def\hhmm{\number\hh:\ifnum\mm<10{}0\fi\number\mm}

\definecolor{vert}{rgb}{0,0.45,0}
\definecolor{rouge}{rgb}{0.89,0.04,0.36}
\definecolor{MyGray}{gray}{0.6}
\definecolor{MyRed}{RGB}{212,42,42}

\def\catego#1{\mathsf{#1}}

\newcommand{\SinCat}[1]{\catego{SCub}_{#1}}
\newcommand{\SinCatG}[1]{\SinCat{#1}^{\gamma}}
\newcommand{\CubCat}[1]{\catego{Cub}_{#1}}
\newcommand{\CubCatG}[1]{\CubCat{#1}^{\Gamma}}

\tikzcdset{arrow style=tikz, diagrams={>=stealth}}

\newcommand{\pcomp}{\circ}

\newcommand{\auteur}[3]{
\noindent
\begin{minipage}[t]{.45\textwidth}
\begin{flushright}
\textsc{#1} \\
{\footnotesize\textsf{#2}}
\end{flushright} 
\end{minipage}
\qquad
\begin{minipage}[t]{.45\textwidth}
#3
\end{minipage}
}

\begin{document}
\thispagestyle{empty}

\begin{center}

\begin{doublespace}
\begin{huge}
{\scshape Single-set cubical categories}
\end{huge}

\vskip+1.5pt

\begin{huge}
{\scshape and their formalisation with a proof assistant}
\end{huge}

\vskip+1.5pt

\begin{huge}
{\scshape (extended version)}
\end{huge}

\vskip+2pt

\bigskip
\hrule height 1.5pt 
\bigskip

\vskip+5pt

\begin{Large}
{\scshape Philippe Malbos - Tanguy Massacrier - Georg Struth}
\end{Large}
\end{doublespace}

\vskip+20pt

\begin{small}\begin{minipage}{14cm}
\noindent\textbf{Abstract --}
We introduce a single-set axiomatisation of cubical
$\omega$-categories, including connections and inverses. We justify
these axioms by establishing a series of equivalences between the
category of single-set cubical $\omega$-categories, and their variants
with connections and inverses, and the corresponding cubical
$\omega$-categories. We also report on the formalisation of cubical
$\omega$-categories with the Isabelle/HOL proof assistant, which has
been instrumental in developing the single-set axiomatisation.

\medskip

\smallskip\noindent\textbf{Keywords --} Cubical $\omega$-categories,
formalised mathematics, Isabelle/HOL, higher-dimensional rewriting.

\smallskip\noindent\textbf{M.S.C. 2020 --} 18N30, 68V15, 03B35, 68Q42.
\end{minipage}
\end{small}
\end{center}

\begin{center}
\begin{small}\begin{minipage}{12cm}
\renewcommand{\contentsname}{}
\setcounter{tocdepth}{1}
\tableofcontents
\end{minipage}
\end{small}
\end{center}

\section{Introduction}

Cubical sets and categories are fundamental structures widely used in mathematics
and theoretical computer science. Several lines of research have
shaped their axioms.  Cubical sets
provide abstract descriptions of higher-dimensional cubes and their
faces. They were first introduced in mathematics for modelling
homotopy types~\cite{Serre1951,Kan1955}. Their algebraic and
categorical descriptions were subsequently obtained via topological
cubical complexes and similar
structures~\cite{Loday1982,BrownLoday1987}.  Cubical categories, which
equip cubical sets with compositions along faces of higher-dimensional
cubes, were introduced by Brown and Higgins for their generalisation
of van Kampen's theorem to higher
dimensions~\cite{BrownHiggins1981,BrownHiggins1981b}. These articles
also introduce a notion of connection on cubical sets, essentially an
operation of rotation of neighbouring faces. More recently,
Lucas~\cite{LucasPhD2017} has added a notion of
inversion for cubes that imposes a groupoid stucture on parts
of the cubical structure. See~\cite{GrandisMauri2003} for a discussion
of additional structure on cubical sets.

Formally, a cubical set is a family of sets
  $(K_n)_{n \in \Nbb}$ equipped with face maps
  $\partial_{n,i}^\alpha : K_n \to K_{n-1}$ and degeneracy maps
  $\epsilon_{n,i}:K_{n-1}\to K_n$, for $1\leq i\leq n$ and
  $\alpha\in\{+,-\}$. 
The former attach faces to higher dimensional cubes; the latter represent lower dimensional cubes as degenerate higher dimensional ones.  The
cubical structure is imposed by the cubical relations
\begin{gather*}
\partial_{n-1,i}^\alpha\partial_{n,j}^\beta = \partial_{n-1,j-1}^\beta\partial_{n,i}^\alpha\quad (i < j),
\qquad
\epsilon_{n+1,i}\epsilon_{n,j} = \epsilon_{n+1,j+1}\epsilon_{n,i}\quad (i\leq j),\\
\partial_{n,i}^\alpha\epsilon_{n,j} = \epsilon_{n-1,j-1}\partial_{n-1,i}^\alpha\quad (i<j),
\qquad
\partial_{n,i}^\alpha\epsilon_{n,j} = \id\quad (i=j),
\qquad
\partial_{n,i}^\alpha\epsilon_{n,j} = \epsilon_{n-1,j}\partial_{n-1,i-1}^\alpha\quad (i>j).
\end{gather*}

In a cubical category, compositions of cubes along
  their faces are defined for each direction~$i$ in a way compatible with
face and degeneracy maps. Adding connections and inverses to cubical
sets and categories imposes further axioms, as expected.

The category of cubical sets with connections and structure-preserving
maps between them forms a strict test category \emph{à la}
Grothendieck~\cite{Maltsiniotis2009}, which makes
it suitable for studying homotopy~\cite{Tonks1992,BrownHigginsSivera2011}.  
Compared to simplicial models, they facilitate the handling of
products. Further, Al-Agl, Brown and Steiner have shown that
categories of cubical categories with connections and
those of globular categories (strict $\omega$-categories),
another kind of higher categories, are equivalent
\cite{AlAglBrownSteiner2002}.

In computer science, some fundamental models of
homotopy type theory are based on cubical
sets~\cite{BezemCoquandHuber2014,CohenCHM17}; 
see~\cite{AngiuliBCHHL21} for an overview.  They support a
constructive approach to Kan fibrations in the simplicial set model of
homotopy type theory~\cite{KapulkinLumsdaine2018}, several properties
of which are undecidable~\cite{BezemCoquand2015}. This prevents a
computational interpretation of Voevodsky's univalence axiom, which is
possible in the cubical model~\cite{BezemCoquandHuber2019}.

A second application of cubical sets in computer science lies in
geometrical and topological models of
concurrency~\cite{FajstrupGoubaultHaucourtMimramRaussen2016}. A
prominent example are
higher-dimensional-automata~\cite{Pratt1991,vanGlabbeek2006}. Here,
$n$-cells represent transitions of a concurrent system where
  $n$ concurrent events are active, and the cubical cell
structure comes from the fact that each concurrent event in an
$n$-dimensional cell can either be active or inactive in each of its
$2n$ $(n-1)$-dimensional faces.  Higher-dimensional automata subsume
many other models of concurrency~\cite{vanGlabbeek2006}. They have
been studied from
homological~\cite{GoubaultJensen1992,Gaucher2002,Kahl18},
homotopical~\cite{Gaucher2000,FajstrupRaussenGoubault2006,
  Grandis2009}, language
theoretic~\cite{FahrenbergJohansenStruthZiemianski2023} and
algorithmic~\cite{FajstrupGR98} points of views.

Finally, at the interface of mathematics and computing, cubical categories have
recently been proposed as a tool for higher-dimensional
rewriting~\cite{Lucas17,Polybook2024}, a categorification of term
rewriting~\cite{Terese03} with applications in categorical algebra.
Diagrammatic statements and proofs of abstract rewriting results, such
as Newman's lemma or the Church-Rosser theorem, use
  indeed cubical shapes; confluence diagrams associated with critical pairs,
triples and $n$-tuples form squares, cubes and $n$-cubes,
respectively.  A categorical description leads to the notion of
polygraphic resolution, which allows the study of homotopical
properties of rewriting
systems~\cite{GuiraudMalbos18,GuiraudMalbos12advances,Polybook2024}.  Explicit
constructions of such resolutions lend themselves naturally to a
formalisation in cubical categories~\cite{Lucas17,LucasPhD2017}.

The ubiquity of cubical sets and cubical categories alone merits a
formalisation with a proof assistant to support reasoning with these
highly combinatorial structures in applications (the axioms for
cubical $\omega$-categories with connections and inverses in
Subsection~\ref{SS:CubicalCategories} below, for instance, cover about
two pages). Yet instead of merely typing an extant axiomatisation into
a prover and checking some well known properties by machine, we use the
Isabelle/HOL proof assistant~\cite{NipkowPaulsonWenzel2002} to develop
an alternative axiomatisation for cubical categories. It is based on
single-set categories~\cite{MacLane98}, where only arrows are modelled
explicitly, while objects remain implicit via their one-to-one correspondence
with identity arrows. Single-set approaches have a long history in
category theory~\cite{MacLane48}; they feature in well-known textbooks~\cite{Freyd1964,MacLane98,FreydScedrov90} and form the basis
of three encyclopaedic formalisations of category theory with
Isabelle~\cite{Stark16,Stark17,Stark20}. Formally, a single-set
category is a set $\Scal$ with source and target maps
$\delta^-,\delta^+ : \Scal \fl \Scal$ and a composition $\pcomp$, a
partial operation such that $x\pcomp y$ is defined if and only if
$\delta^+x=\delta^-y$ for $x,y\in \Scal$, which satisfy
\begin{gather*}
\delta^-(x\pcomp y) = \delta^-x, 
\qquad
\delta^+(x\pcomp y) =  \delta^+y,
\qquad
x\pcomp \delta^+x = x,
\qquad
\delta^-x\pcomp x = x,\\
(x\pcomp y) \pcomp z = x\pcomp (y\pcomp z),
\qquad
\delta^-=\delta^+\delta^-,
\qquad
\delta^+=\delta^-\delta^+.
\end{gather*}
Identity arrows arise as fixed
points of $\delta^-$ or equivalently those of $\delta^+$.  Single-set
categories are therefore algebraically simpler than their
classical siblings defined via objects and arrows. Functors
and natural transformations are simply
functions~\cite{Freyd1964}. Single-set higher categories may thus be more suitable for symbolic
reasoning and automated proof search than their classical
counterparts. Indeed, single-set globular categories
are used widely~\cite{BrownH81,MacLane98,
  AlAglBrownSteiner2002,Steiner04} and have been formalised with
Isabelle~\cite{CalkMalbosPousStruth23,CalkStruth24}. Yet
  single-set cubical categories remain to be defined.

This is not entirely straightforward.  We had to introduce symmetry
maps that relate the sets of fixed points modelling higher identities
in different directions as replacements of the traditional degeneracy
maps. Initially, this led to an unwieldy number of axioms, which would
have been tedious to use and would have inflated the
  categorical equivalence proof, which justifies them relative to
  their classical counterparts.
  
Isabelle has been instrumental in taming these axioms due to its
powerful support for proof automation and counterexample search, which
sets it apart from other proof assistants. Its proof automation comes
from internal simplification and proof procedures and external proof search tools -- so-called hammers
-- for first-order logic. Counterexample search uses SAT solvers and
decision procedures, for instance for linear arithmetic. This
combination supports not only a natural mathematical workflow with
proofs and refutations, it also allows checking axiom systems for
redundancy (via deduction) and irredundancy (via counterexamples)
rapidly and effectively. It has already proved its worth for
developing other algebraic
axiomatisations~\cite{DesharnaisS11,FurusawaS15,
  CalkMalbosPousStruth23}. Here, Isabelle has helped us to bring the
single-set axiomatisation for cubical categories to a manageable size
without compromising its structural coherence, to simplify candidate
axioms and to analyse candidate axioms that emerged during our
development rapidly. Starting from an around $40$ initial
  candidate axioms, we have used Isabelle in an iterative process,
  simplifying candidate axioms and removing redundant ones, then
  attempting an equivalence proof, and adding new axioms if that
  failed. Without our confidence in Isabelle's automated proof
  support, we might not have attempted this research.

The single-set axiomatisation for cubical categories thus forms the
main conceptual contribution in this article. Our main technical
contribution consists in the proofs of categorical equivalence
mentioned, and our main engineering contribution is the formalisation
of a mathematical component for cubical categories with Isabelle. In
combination, these results constitute a case study in innovative, not
merely reconstructive formalised mathematics.

The overall structure of this article is simple: our axioms for
single-set cubical categories are introduced in
Section~\ref{SS:SingleSetHigherCategories}, the proofs that the
resulting categories are essentially the same as their classical
counterparts are given in
Section~\ref{S:EquivalenceWithCubicalCategories}, our Isabelle
formalisation and the workflow leading to our axioms are
discussed in Section~\ref{S:FormalisationIsabelle}. Finally, we
summarise our results and present some avenues for future work
in Section~\ref{S:Conclusion}.

More specifically, we recall the variant of single-set
categories~\cite{CranchDS20,FahrenbergJSZ21,Struth23}, on which our
axioms for single-set cubical categories are based, in
Subsection~\ref{SS:singleSetCategories}. 
Subsection~\ref{SS:SingleSetCubicalCategories}
introduces single-set cubical $\omega$-categories, axiomatised as a
set~$\Scal$ equipped with families of maps indexed by directions
$i\in \Nbb_+$: face maps $\delta_i^-$ and $\delta_i^+$, composition
maps~$\circ_i$, symmetry maps $s_i$ and reverse symmetry maps
$\tilde{s}_i$. We show how single-set
cubical $n$-categories appear as truncations.  We define the category
$\SinCat{\omega}$ with single-set cubical $\omega$-categories as
objects and functions corresponding to functors of classical
$\omega$-categories as morphisms.  We also list some structural
properties of these categories.  In
Subsection~\ref{SS:SingleSetWithConnections} we add connections, in
Subsection~\ref{SS:SinSetWithInv} we further add inverses in each
dimension greater than $p$. This yields the categories
$\SinCatG{\omega}$ and $\SinCatG{(\omega,p)}$ as well as truncated
variants for each dimension $n$.  Inverses are relevant to
constructive proofs in homotopy type theory and higher-dimensional
rewriting.

In Subsection~\ref{SS:CubicalCategories} we recall the classical
axioms for cubical $\omega$- and $n$-categories, including connections
and inverses. This leads to the categories $\CubCat{\omega}$,
$\CubCatG{\omega}$ and $\CubCatG{(\omega,p)}$ with classical cubical
$\omega$-categories as objects and functors as morphisms.  We then
present proofs of the equivalences
$\SinCat{\omega} \simeq \CubCat{\omega}$ in
Theorem~\ref{T:EquivSinSetClass},
$\SinCatG{\omega} \simeq \CubCatG{\omega}$ in
Theorem~\ref{T:EquivSinSetClassConn} and
$\SinCatG{(\omega,p)} \simeq \CubCatG{(\omega,p)}$ in
Theorem~\ref{T:EquivSinSetClassConnInv} in
Subsections~\ref{SS:EquivalenceWithCubicalCategories},
\ref{SS:EquivalenceWithCubicalCategoriesWithConnections} and
\ref{SS:EquivalenceWithCubicalCategoriesWithInverses} respectively.
Straightforward modifications yield similar equivalences between
$n$-categories, which we do not list explicitly.

Subsection~\ref{SS:nutshell} contains a brief overview of
Isabelle, Subsection~\ref{SS:catoid-isa} recalls the formalisation
of single-set categories with Isabelle~\cite{Struth23}, which
underlies our formalisation of cubical $\omega$-categories with and
without connections in Subsection~\ref{SS:cubical-isabelle}. For
technical reasons, we do not formalise cubical
$(\omega,p)$-categories. But we show how a non-trivial proof about
$(\omega,0)$-categories (of Proposition~\ref{P:OmegaZeroInv}) can be
formalised with our axiomatisation at the same level of granularity.

While this article can be read as an exercise in formalised
mathematics, we did not aim to formalise all our results, as
  it would distract from our main goal: to showcase the unique
  benefits of Isabelle's proof automation in the analysis of higher
  categories.
Alternatively, disregarding Section~\ref{S:FormalisationIsabelle}, it
can be read as a mathematical paper with contributions beyond
Isabelle. While the calculational lemmas in
Section~\ref{SS:SingleSetHigherCategories} have been checked by
machine, and proofs therefore been omitted, the categorical
equivalences in Section~\ref{S:EquivalenceWithCubicalCategories} have
not been formalised, though that would have been possible at
  least in parts. Once again: our main use case for Isabelle in this work has been the
development of the axioms in $\SinCat{\omega}$, $\SinCatG{\omega}$ and
$\SinCatG{(\omega,p)}$. Further work with proof assistants on higher
categories, higher rewriting and higher automata is left for future
work.  Our Isabelle components for cubical categories, including a PDF
proof document, can be found in the Archive of Formal
Proofs~\cite{MassacrierStruth24}.

\section{Single-set cubical categories}
\label{SS:SingleSetHigherCategories}

In this section we introduce our axiomatisation of single-set cubical
categories. In Subsection~\ref{SS:singleSetCategories} we recall a
previous axiomatisation of single-set categories.  In Subsection
\ref{SS:SingleSetCubicalCategories} we introduce single-set cubical
$\omega$-categories and $n$-categories. Extensions of these categories
with connections and inverses are presented in
Subsections~\ref{SS:SingleSetWithConnections}
and~\ref{SS:SinSetWithInv}.

\subsection{Single-set categories}
\label{SS:singleSetCategories}

We start with recalling the definition and basic properties of
single-set categories. While any axiomatisation would work for our
purposes, we have chosen one in which the partiality of arrow
composition is captured by a multioperation that maps pairs of
elements to sets of elements, including the empty
set~\cite{CranchDS20,FahrenbergJSZ21}. It is already well-supported by
Isabelle components~\cite{Struth23} and has previously served as a
basis for formalising globular single-set
$\omega$-categories~\cite{CalkMalbosPousStruth23,CalkStruth24}.

\subsubsection{}
\label{SSS:SingleSetOneCategories}
A \emph{single-set category} $(\Scal,\delta^-,\delta^+,\odot)$
consists of the following data:
\begin{itemize}
\item a set $\Scal$ of \emph{cells},
\item \emph{face maps} $\delta^\alpha: \Scal\to\Scal$ for
  $\alpha \in \{-,+\}$, which are extended to
  $\mathcal{P}(\Scal) \to \mathcal{P}(\Scal)$ by taking images,
\item a \emph{composition map} $\odot: \Scal\times\Scal \to \mathcal{P}(\Scal)$, which is extended to $\mathcal{P}(\Scal) \times \mathcal{P}(\Scal) \to \mathcal{P}(\Scal)$ as
\begin{eqn}{equation*}
\label{E:FaceMapsCompoSets}
X \odot Y = \bigcup_{x \in X, y\in Y} x \odot y, \qquad \text{ for all } X,Y \subseteq \Scal.
\end{eqn}
\end{itemize}
It satisfies, for all $x,y,z\in\Scal$,
\begin{enumerate}[{\bf (i)}]
\item \label{I:AxiomAssoc} \emph{associativity}: $\{x\} \odot (y \odot z) =
  (x \odot y) \odot \{z\}$,
 \item \label{I:AxiomUnits} \emph{units}: $x \odot \delta^+ x = \{x\}$
   and $\delta^- x \odot x=\{x\}$,
\item \label{I:AxiomLocality} \emph{locality}: $x \odot y \neq \varnothing \Leftrightarrow \delta^+ x = \delta^- y$,
\item \label{I:AxiomFunctionality} \emph{functionality}: $\forall z,z' \in x
  \odot y,\  z=z'$.
\end{enumerate}

The cells of single-set categories correspond to arrows of classical
categories.  The face maps $\delta^-$ and $\delta^+$ send each cell in
$\Scal$ to its \emph{source cell} and \emph{target cell},
respectively, which are \emph{identity cells}. These are in bijective
correspondence with objects of classical categories.

Henceforth we tacitly assume that upper indices such as $\alpha$ in
all face maps $\delta^\alpha$ range over $\{-,+\}$. We also
  write $\delta^{-\alpha}$ to indicate that values of $\alpha$ are
  exchanged relative to an occurrence of $\delta^\alpha$.
Further, in order to avoid lengthy technical terms, we
  henceforth refer to single-set categories simply as categories
  wherever possible.

\begin{remark}
\label{R:RemarkCatoids}
Omitting the functionality axiom and the right-to-left direction of
the locality axiom from the definition of single-set categories yields
axioms for \emph{catoids}. Removing locality, functionality and the
unit axioms yields \emph{multisemigroups}. See~\cite{FahrenbergJSZ21}
for details.
\end{remark}

\subsubsection{Composition as a partial operation}
\label{SSS:PartialCompositionSingleSet}
Two cells $x$, $y$ of a category $\Scal$ are
\emph{composable} if $x \odot y \neq \varnothing$, in which case we
write $\Delta(x,y)$ for short.  Functionality makes $\odot$ a partial
operation $\pcomp : \Delta \hookrightarrow \Scal\times\Scal \to \Scal$, which
sends each $(x,y)\in \Scal\times\Scal$ to the unique $z \in x \odot y$
whenever $\Delta(x,y)$. As we can recover
\begin{eqn}{equation*}
x \odot y =
\begin{cases}
\{ x \pcomp y \} & \text{if } \Delta(x,y), \\
\varnothing & \text{otherwise}
\end{cases}
\end{eqn}
from $\delta^-$, $\delta^+$ and $\circ$, we henceforth write
$(\Scal,\delta^-,\delta^+,\circ)$ instead of
$(\Scal,\delta^-,\delta^+,\odot)$ and work with $\circ$ instead of $\odot$.

\subsubsection{}
A \emph{morphism} $f: \Scal \to \Scal'$ \emph{of categories} $\Scal$ and $\Scal'$ is a map satisfying, for all $x,y\in\Scal$,
\begin{eqn}{equation*}
\label{I:AxiomMorphFace}
f \delta^\alpha = {\delta'}^\alpha f
\qquad\text{ and }\qquad
\Delta(x,y) ~ \Rightarrow ~ f(x \pcomp y) = f(x) \pcomp' f(y).
\end{eqn}
Such morphisms correspond to functors between classical categories.

\begin{example}
With $\pcomp$, the associativity axiom of categories becomes
\begin{align*}
\Delta(x, y \pcomp z) \land \Delta(y,z) ~ & \Leftrightarrow ~ \Delta(x, y) \land \Delta(x \pcomp y, z),
\\
\Delta(x, y \pcomp z) \land \Delta(y,z) ~ & \Rightarrow ~ x \pcomp (y \pcomp z) = (x \pcomp y) \pcomp z.
\end{align*}
By the first law, the left-hand side of the associativity law is
defined if and only if its right-hand side is. By the second law, the
two sides of this law are equal if either side is defined.  Likewise,
the unit axioms simplify to $\Delta(x, \delta^+ x)$, $\Delta(\delta^- x, x)$, $x \pcomp \delta^+ x = x$ and
$\delta^- x \pcomp x = x$.
\end{example}

\begin{example}
\label{Ex:IdentityAxiom}
In preparation for the cubical categories below, suppose that the
cells of a category are formed by squares that
can be composed horizontally, for instance the commuting diagrams in
an arrow category. The right unit axiom $x \circ \delta^+ x = x$ above
can then be illustrated as
\begin{equation*}
  	\begin{tikzcd}[global scale = 1.3 and 1.3 and 1 and 1, longer arrows = 0.4em and 0.4em]
	\phantom{\circ}\ar[rr] \ar[dd]  && \phantom{\circ}\ar[dd, dotted]\\
	& x & \\
	\phantom{\circ}\ar[rr] && \phantom{\circ}
      \end{tikzcd}
\ \circ \
        	\begin{tikzcd}[global scale = 1.3 and 1.3 and 1 and 1, longer arrows = 0.4em and 0.4em]
	\phantom{\circ}\ar[rr, equal] \ar[dd, dotted]  && \phantom{\circ}\ar[dd, dotted]\\
	& \delta^+x & \\
	\phantom{\circ}\ar[rr, equal] && \phantom{\circ}
      \end{tikzcd}
      \ \ = \ \
  	\begin{tikzcd}[global scale = 1.3 and 1.3 and 1 and 1, longer arrows = 0.4em and 0.4em]
	\phantom{\circ}\ar[rr] \ar[dd]  && \phantom{\circ}\ar[dd, dotted]\\
	& x & \\
	\phantom{\circ}\ar[rr] && \phantom{\circ}
      \end{tikzcd}    
\end{equation*}
The upper and lower faces of $\delta^+ x$ are drawn as equality arrows
to indicate that the left and right faces of this cell, shown as
dotted arrows, are equal. We assign a more precise semantics to such
cubes below.
\end{example}

We frequently need the following laws in calculations. They have been
verified with Isabelle~\cite{Struth23}.

\begin{lemma}
\label{L:sscat-props}
Let $\Scal$ be a category. Then
\begin{enumerate}[{\bf (i)}]
\item \label{I:AxiomStability} $\delta^\alpha \delta^\beta = \delta^\beta$,
\item \label{I:AxiomFixPlusMinus} $\delta^- x = x ~\Leftrightarrow~ \delta^+ x = x$ for all $x \in \Scal$,
\item $\delta^- (x \pcomp y) = \delta^- x$ and $\delta^+ (x \pcomp y) = \delta^+ y$ for all $x,y \in \Scal$ such that $\Delta(x,y)$.
\end{enumerate}
\end{lemma}
\begin{example}
We can illustrate Lemma~\ref{L:sscat-props}(\ref{I:AxiomStability}) as
\begin{equation*}
  	\begin{tikzcd}[global scale = 1.3 and 1.3 and 1 and 1, longer arrows = 0.4em and 0.4em]
	\phantom{\circ}\ar[rr] \ar[dd]  && \phantom{\circ}\ar[dd]\\
	& {} & \\
	\phantom{\circ}\ar[rr] && \phantom{\circ}
      \end{tikzcd}
\ \stackrel{\delta^\beta}{\mapsto} \
\begin{tikzcd}[global scale = 1.3 and 1.3 and 1 and 1, longer arrows = 0.4em and 0.4em]
	\phantom{\circ}\ar[rr, equal] \ar[dd]  && \phantom{\circ}\ar[dd]\\
	& {} & \\
	\phantom{\circ}\ar[rr, equal] && \phantom{\circ}
      \end{tikzcd}
      \ \stackrel{\delta^\alpha}{\mapsto} \
          	\begin{tikzcd}[global scale = 1.3 and 1.3 and 1 and 1, longer arrows = 0.4em and 0.4em]
	\phantom{\circ}\ar[rr, equal] \ar[dd]  && \phantom{\circ}\ar[dd]\\
	& {} & \\
	\phantom{\circ}\ar[rr, equal] && \phantom{\circ}
      \end{tikzcd}
\end{equation*}
\end{example}

\subsubsection{Fixed points}
Lemma~\ref{L:sscat-props}(\ref{I:AxiomStability}) and
(\ref{I:AxiomFixPlusMinus}) imply that the sets of fixed points of
$\delta^-$ and $\delta^+$ in a category $\Scal$
are equal and also equal to the sets of all left and right identities
in $\Scal$. We write $\Scal^\delta$ for the resulting set of all
identities. It corresponds to the set of objects in (small) classical
categories. In our examples, identities are illustrated by degenerated
cubes, in which some opposite arrows are equality arrows.

\subsection{Single-set cubical categories}
\label{SS:SingleSetCubicalCategories}

Our single-set axiomatisation of cubical $\omega$-categories is based on a family of categories
$(\Scal,\delta_i^-,\delta_i^+,\circ_i)$ for each
$i \in \Nbb_+ = \Nbb \backslash \{0\}$. We thus equip our previous
notation with indices. 
In particular we write $\Scal^i$ for the set of fixed points of $\delta_i^\alpha$.

\subsubsection{}
\label{SSS:DefCubSingleSetWithoutConn}
A \emph{single-set cubical $\omega$-category} consists of a family of
single-set categories
$(\Scal,\delta_i^-,\delta_i^+,\circ_i)_{i\in \Nbb_+}$ with
\emph{symmetry maps} $s_i: \Scal\to\Scal$ and \emph{reverse symmetry
  maps} $\tilde{s}_i: \Scal\to\Scal$ for each $i \in \Nbb_+$. These
satisfy, for all $w,x,y,z\in\Scal$ and $i,j \in \Nbb_+$,
\begin{enumerate}[{\bf (i)}]
\item \label{I:AxiomCommutativity} $\delta_i^\alpha \delta_j^\beta = \delta_j^\beta \delta_i^\alpha$ if $i \neq j$,
\item \label{I:DeltaCompoCompat} $\delta_i^\alpha (x \circ_j y) = \delta_i^\alpha x \circ_j \delta_i^\alpha y$ if $i \neq j$ and $\Delta_j(x,y)$,
\item \label{I:ExchangeLaw} $(w \circ_i x) \circ_j (y \circ_i z)
  = (w \circ_j y) \circ_i (x \circ_j z)$ if $i \neq j$, $\Delta_i(w,x)$, $\Delta_i(y,z)$, $\Delta_j(w,y)$ and $\Delta_j(x,z)$,
\item \label{I:AxiomSymType} $s_i (\Scal^i) \subseteq \Scal^{i+1}$ and $\tilde{s}_i (\Scal^{i+1}) \subseteq \Scal^i$,
\item \label{I:AxiomSymInv} $\tilde{s}_i s_i x = x$ and $s_i \tilde{s}_i y = y$ if $x \in \Scal^i$ and $y \in \Scal^{i+1}$,
\item \label{I:AxiomFaceSym} $\delta_j^\alpha s_j x = s_j \delta_{j+1}^\alpha x$ and $\delta_i^\alpha s_j x = s_j \delta_i^\alpha x$ if $i \neq j,j+1$ and $x \in \Scal^j$,
\item \label{I:AxiomSymCompatCompo} $s_i (x \circ_{i+1} y) = s_i x \circ_i s_i y$ and $s_i (x \circ_j y) = s_i x \circ_j s_i y$ if $j \neq i,i+1$, $x,y \in \Scal^i$ and $\Delta_j(x,y)$,
\item \label{I:AxiomSymFix} $s_i x = x$ if $x \in \Scal^i \cap \Scal^{i+1}$,
\item \label{I:AxiomSymBraid} $s_i s_j x = s_j s_i x$ if $|i-j| \geq 2$ and $x \in \Scal^i \cap \Scal^j$,
\item \label{I:AxiomFiniteDimCells} $\exists k \in \Nbb \ \forall i \geq k+1, \ x \in \Scal^i$.
\end{enumerate}

A \emph{single-set cubical $n$-category}, for $n\in \Nbb$, is defined
by the same data, but the index $i\in \Nbb_+$ is restricted to
$1 \leq i\leq n$, and the $s_i$ and $\tilde{s}_i$ to
$1\leq i \leq n-1$. Likewise, all $\omega$-axioms are restricted to
these ranges, and Axiom~(\ref{I:AxiomFiniteDimCells}) is omitted, as
it is now entailed.  All results for
$\omega$-categories in this article restrict to $n$-categories. 
As the $\omega$-categories or
  $n$-categories considered in this article are usually cubical, we
  drop this adjective whenever possible. Hence we often refer to
  single-set cubical $\omega$-categories simply as $\omega$-categories
  and likewise for $n$-categories.

As previously, we call the elements of $\omega$-categories \emph{cells}. 
The face maps $\delta_i^-$ and
$\delta_i^+$ attach \emph{lower} and \emph{upper faces in direction
  $i$} to them. The symmetry maps $s_i$ and reverse symmetry
  maps $\tilde{s}_i$ rotate identities for $\circ_i$ to identities for
  $\circ_{i+1}$ and vice versa.

\subsubsection{Explanation of axioms}

The cells of categories model higher-dimensional
cubes, possibly with degenerate faces, which can be composed by gluing
them together along their faces.  While the examples in
Subsection~\ref{SS:singleSetCategories} provide some intuition for
squares in low dimensions, we can illustrate higher dimensional cells
and their compositions only through projections.

A cell $x$ and its faces in the directions $i$ and $j$ can be illustrated as
\begin{eqn}{equation*}
\label{E:FaceMapsDirections}
\begin{tikzcd}[global scale = 2 and 2 and 1 and 1]
\ar[r, shorten <= -.45em, shorten >= -.3em] \ar[d, shorten <= -.6em, shorten >= -.3em] & j \\
i & 
\end{tikzcd}
\hspace{5em}
\begin{tikzcd}[global scale = 3 and 2.2 and 1 and 1, longer arrows = 0.4em and 0.4em]
\delta_i^- \delta_j^- x \ar[rr, "\delta_i^- x"] \ar[dd, "\delta_j^- x"'] && \delta_i^- \delta_j^+ x \ar[dd, "\delta_j^+ x"] \\
 & x & \\
\delta_i^+ \delta_j^- x \ar[rr, "\delta_i^+ x"'] && \delta_i^+ \delta_j^+ x
\end{tikzcd}
\hspace{4em}
\begin{tikzcd}[global scale = 1.9 and 1.5 and 1 and 1, longer arrows = 0.4em and 0.4em]
\phantom{\circ} \ar[rr, equal] \ar[dd, equal] && \phantom{\circ} \ar[rr] \ar[dd, equal] && \phantom{\circ} \ar[rr, equal] \ar[dd, equal] && \phantom{\circ} \ar[dd, equal] \\
 & \delta_i^- \delta_j^- x && \delta_i^- x && \delta_i^- \delta_j^+ x & \\
\phantom{\circ} \ar[rr, equal] \ar[dd] && \phantom{\circ} \ar[rr] \ar[dd] && \phantom{\circ} \ar[rr, equal] \ar[dd] && \phantom{\circ} \ar[dd] \\
 & \delta_j^- x && x && \delta_j^+ x & \\
\phantom{\circ} \ar[rr, equal] \ar[dd, equal] && \phantom{\circ} \ar[rr] \ar[dd, equal] && \phantom{\circ} \ar[rr, equal] \ar[dd, equal] && \phantom{\circ} \ar[dd, equal] \\
 & \delta_i^+ \delta_j^- x && \delta_i^+ x && \delta_i^+ \delta_j^+ x & \\
\phantom{\circ} \ar[rr, equal] && \phantom{\circ} \ar[rr] && \phantom{\circ} \ar[rr, equal] && \phantom{\circ}
\end{tikzcd}
\end{eqn}
The arrows on the left indicate the directions $i$ and $j$.
The diagram on the right shows the faces of $x$ as degenerate
cells. They are identities of the composition in the same direction.
Many of the axioms $\omega$-categories can be illustrated by such diagrams.
\begin{itemize}
\item Axiom~{(\ref{I:AxiomCommutativity})} determines the cubical cell
  shape. It can be depicted as
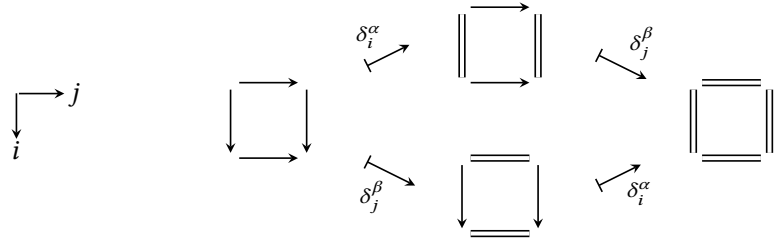
\begin{equation*}
\begin{tikzcd}[global scale = 2 and 2 and 1 and 1]
	\phantom{\circ}\ar[r, shorten <= -.61em, shorten >= -.3em]
        \ar[d, shorten <= -.65em, shorten >= -.3em]  & j\\
	i & 
\end{tikzcd}
      \qquad\qquad
\begin{tikzcd}[global scale = 1.3 and 1.3 and 1 and 1, longer arrows = 0.4em and 0.4em]
	\phantom{\circ}\ar[rr] \ar[dd]  && \phantom{\circ}\ar[dd]\\
	& {} & \\
	\phantom{\circ}\ar[rr] &&
        \phantom{\circ}
\end{tikzcd}
\begin{tikzcd}[global scale = 1.3 and 1.3 and 1 and 1]
        &&\\
        \phantom{\circ}\ar[urr, |->, "\delta_i^\alpha"] &&\phantom{\circ}\\
        &&\\
        \phantom{\circ}\ar[drr, |->, "\delta_j^\beta"'] &&\phantom{\circ}\\
        &&
\end{tikzcd}
\begin{tikzcd}[global scale = 1.3 and 1.3 and 1 and 1, longer arrows = 0.4em and 0.4em]
  	\phantom{\circ}\ar[rr]
        \ar[dd, equal]  &&
        \phantom{\circ}\ar[dd, equal]\\
	& {} & \\
	\phantom{\circ}\ar[rr] && \phantom{\circ}\\
        &&\\
	\phantom{\circ}\ar[rr, equal] \ar[dd]  && \phantom{\circ}\ar[dd]\\
	& {} & \\
	\phantom{\circ}\ar[rr, equal] && \phantom{\circ}
\end{tikzcd}
\begin{tikzcd}[global scale = 1.3 and 1.3 and 1 and 1]
        \phantom{\circ}\ar[drr, |->, "\delta_j^\beta"] &&\phantom{\circ}\\
        &&\\
        &&\\
        &&\\
        \phantom{\circ}\ar[urr, |->, "\delta_i^\alpha"'] &&\phantom{\circ}
\end{tikzcd}
\begin{tikzcd}[global scale = 1.3 and 1.3 and 1 and 1, longer arrows = 0.4em and 0.4em]
	\phantom{\circ}\ar[rr, equal] \ar[dd, equal]  &&
        \phantom{\circ}\ar[dd, equal]\\
	& {} & \\
	\phantom{\circ}\ar[rr, equal] && \phantom{\circ}
\end{tikzcd}
\end{equation*}
\item Axioms~{(\ref{I:AxiomCommutativity})}
  and~(\ref{I:DeltaCompoCompat}) make face maps morphisms and hence
  functors with respect to the underlying categories in any
  other direction.
\item The \emph{interchange law} in Axiom~(\ref{I:ExchangeLaw}) makes
  composition in any direction bifunctorial with respect to the
  compositions in any other direction. It can be depicted as
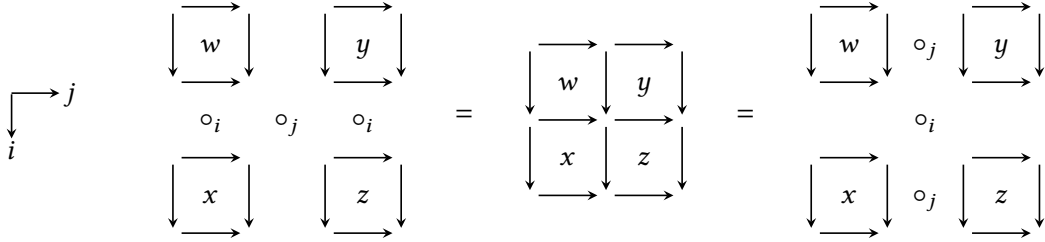
\[
\begin{array}{cccc}
	\begin{tikzcd}[global scale = 2 and 2 and 1 and 1]
	\ar[r, shorten <= -.45em, shorten >= -.3em] \ar[d, shorten <= -.6em, shorten >= -.3em] & j \\
	i & 
	\end{tikzcd}
&
\qquad
	\begin{tikzcd}[global scale = 1.3 and 1.3 and 1 and 1, longer arrows = 0.4em and 0.4em]
	\phantom{\circ} \ar[rr] \ar[dd] && \phantom{\circ} \ar[dd] && \phantom{\circ} \ar[rr] \ar[dd] && \phantom{\circ} \ar[dd] \\
	 & w &&&& y & \\
	\phantom{\circ} \ar[rr] && \phantom{\circ} && \phantom{\circ} \ar[rr] && \phantom{\circ} \\
	 & \circ_i && \circ_j && \circ_i & \\
	\phantom{\circ} \ar[rr] \ar[dd] && \phantom{\circ} \ar[dd] && \phantom{\circ} \ar[rr] \ar[dd] && \phantom{\circ} \ar[dd] \\
	 & x &&&& z & \\
	\phantom{\circ} \ar[rr] && \phantom{\circ} && \phantom{\circ} \ar[rr] && \phantom{\circ}
	\end{tikzcd}
&
\quad = \quad
	\begin{tikzcd}[global scale = 1.3 and 1.3 and 1 and 1, longer arrows = 0.4em and 0.4em]
	\phantom{\circ} \ar[rr] \ar[dd] && \phantom{\circ} \ar[rr] \ar[dd] && \phantom{\circ} \ar[dd] \\
	 & w && y & \\
	\phantom{\circ} \ar[rr] \ar[dd] && \phantom{\circ} \ar[rr] \ar[dd] && \phantom{\circ} \ar[dd] \\
	 & x && z & \\
	\phantom{\circ} \ar[rr] && \phantom{\circ} \ar[rr] && \phantom{\circ}
	\end{tikzcd}
&
\quad = \quad
	\begin{tikzcd}[global scale = 1.3 and 1.3 and 1 and 1, longer arrows = 0.4em and 0.4em]
	\phantom{\circ} \ar[rr] \ar[dd] && \phantom{\circ} \ar[dd] && \phantom{\circ} \ar[rr] \ar[dd] && \phantom{\circ} \ar[dd] \\
	 & w && \circ_j && y & \\
	\phantom{\circ} \ar[rr] && \phantom{\circ} && \phantom{\circ} \ar[rr] && \phantom{\circ} \\
	 &&& \circ_i &&& \\
	\phantom{\circ} \ar[rr] \ar[dd] && \phantom{\circ} \ar[dd] && \phantom{\circ} \ar[rr] \ar[dd] && \phantom{\circ} \ar[dd] \\
	 & x && \circ_j && z & \\
	\phantom{\circ} \ar[rr] && \phantom{\circ} && \phantom{\circ} \ar[rr] && \phantom{\circ}
	\end{tikzcd}
\end{array}
\]

\item The \emph{typing axioms} in~(\ref{I:AxiomSymType}) restrict
  $s_i$ to $\Scal^i \to \Scal^{i+1}$ and $\tilde{s}_i$ to
  $\Scal^{i+1} \to \Scal^i$. Though the $s_i$ and $\tilde{s}_i$ are
  defined as total maps on $\Scal$, only their typed versions matter.
 
\item Axiom~(\ref{I:AxiomSymInv}) states that each
  $s_i: \Scal^i \to \Scal^{i+1}$ and
  $\tilde{s}_i:\Scal^{i+1} \to \Scal^i$ forms a bijective pair.
  
\item Axiom~(\ref{I:AxiomFaceSym}) captures the action of a symmetry
  map $s_j$ on a cell $x \in \Scal^j$: it rotates it together with
  its faces into a cell in $\Scal^{j+1}$:
\begin{eqn}{equation*}
\begin{array}{cccc}
	\begin{tikzcd}[global scale = 2.5 and 2 and 1 and 1]
	\ar[r, shorten <= -.45em, shorten >= -.3em] \ar[d, shorten <= -.6em, shorten >= -.3em] & j+1 \\
	j & 
	\end{tikzcd}
&
\qquad
	\begin{tikzcd}[global scale = 1.3 and 1.3 and 1 and 1, longer arrows = 0.4em and 0.4em]
	\phantom{\circ} \ar[rr] \ar[dd, equal] && \phantom{\circ} \ar[dd, equal] \\
	& x & \\
	\phantom{\circ} \ar[rr] && \phantom{\circ}
	\end{tikzcd}
&
\
\overset{s_j}{\mapsto} \
&
	\begin{tikzcd}[global scale = 1.3 and 1.3 and 1 and 1, longer arrows = 0.4em and 0.4em]
	\phantom{\circ} \ar[rr, equal] \ar[dd] && \phantom{\circ} \ar[dd] \\
	& s_j x & \\
	\phantom{\circ} \ar[rr, equal] && \phantom{\circ}
	\end{tikzcd}
\end{array}
\end{eqn}
\item Axioms~{(\ref{I:AxiomFaceSym})}
  and~(\ref{I:AxiomSymCompatCompo}) state that the symmetry maps $s_i$
  are morphisms, hence functors, with respect to the categories in any direction $j \neq i,i+1$.
\item The \emph{dimensionality} axiom~(\ref{I:AxiomSymFix}) imposes
  that $s_i$ is an identity map on $\Scal^i\cap\Scal^{i+1}$.
\item Axiom~(\ref{I:AxiomSymBraid}) is a \emph{braiding} axiom for
  symmetries. An illustration would require higher-dimensional cubes.
\item Axiom~(\ref{I:AxiomFiniteDimCells}) imposes that every cell has
  finite dimension, as defined below. It is important for the
    proof of equivalence between single-set cubical categories and
    their classical counterparts as well as for the coherence of the
    entire approach.
\end{itemize}
A more structural explanation of symmetries and their reverses is
given in the following subsections.

\subsubsection{Lattice of fixed points}
\label{SSS:Lattice}
Let $\Scal$ be an $\omega$-category.  
It is easy to see that the set of all
$\Scal^I = \bigcap_{i \in I} \Scal^i$, for $I\subseteq \Nbb_+$, forms
a lattice with respect to set inclusion. In fact,
$\Scal^I \subseteq \Scal^J$ if and only if $I \supseteq J$. In this
lattice, the $S^I$ with co-finite $I$ model sets of cells of finite
dimension.

Formally, a cell $x \in \Scal$ \emph{has dimension at most}
$k\in \Nbb_+$ if $x \in \Scal^I$ for some $I \subseteq \Nbb_+$ with
$|\Nbb_+\mathop{\setminus} I| = k$; it \emph{has dimension} $k$ if $k$
is the least positive integer for which it has dimension at most $k$.

If the complements of the co-finite sets $I,J \subseteq \Nbb_+$ have
the same (finite) cardinality, then the restrictions of symmetries,
their reverses and their compositions induce bijections
$\Scal^I \simeq \Scal^J$. These identify cells of the same dimension;
they respect face maps and compositions.

For finite $I=\{i_1,\dots,i_k\}$, we write $\Scal^{i_1,\dots,i_k}$. In
this case,
Axiom~\eqref{SSS:DefCubSingleSetWithoutConn}~\eqref{I:AxiomCommutativity}
implies that $x \in \Scal^{i_1,\dots,i_k}$ if and only if
$\delta_{i_1}^{\alpha_1} \dots \delta_{i_k}^{\alpha_k} x = x$ for all
$\alpha_1,\dots,\alpha_k$.

Finally, we write $\Scal^{>n}$ for $\Scal^I$ when
$I=\{i\in \Nbb\mid i > n\}$. All cells in $\Scal^{>n}$ have dimension
at most $n$.  Axiom~\eqref{I:AxiomFiniteDimCells} implies that
$\Scal = \bigcup_{n\geq 0} \Scal^{>n}$. Hence every cell $x \in \Scal$
has indeed some finite dimension, which is crucial for the
  constructions of the equivalence proof relating single-set cubical
  categories with their classical counterparts in
  Section~\ref{S:EquivalenceWithCubicalCategories}.
  In particular, Axiom~\eqref{I:AxiomFiniteDimCells} thus imposes 
  that the gradation of $\Scal$ in terms of $\Scal^{>n}$ corresponds 
  to the gradation of the cubical sets $K_n$ shown in the introduction.

\subsubsection{Remark}
The definition of symmetry in~\eqref{SSS:DefCubSingleSetWithoutConn}
differs from the notions of symmetric cubical monoid and category,
introduced by Grandis and Mauri~\cite{GrandisMauri2003} and
Grandis~\cite{Grandis2007}, respectively.

\begin{example}
  The lattice of fixed points for the $3$-category
  $\Scal$ is given by the Hasse diagram
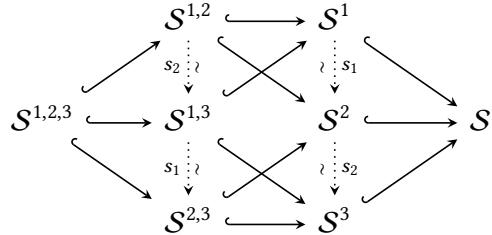
\[
\begin{tikzcd}[global scale = 5 and 3.5 and 1 and 1]
 & \Scal^{1,2} \ar[r, hookrightarrow] \ar[dr, hookrightarrow] \ar[d, dotted, "s_2"', "{\rotatebox[origin=c]{90}{$\sim$}}"]
 & \Scal^1 \ar[dr, hookrightarrow] \ar[d, dotted, "s_1", "{\rotatebox[origin=c]{90}{$\sim$}}"']
 &
\\
\Scal^{1,2,3} \ar[ur, hookrightarrow] \ar[r, hookrightarrow] \ar[dr, hookrightarrow]
 & \Scal^{1,3} \ar[ur, hookrightarrow] \ar[dr, hookrightarrow] \ar[d, dotted, "s_1"', "{\rotatebox[origin=c]{90}{$\sim$}}"]
 & \Scal^2 \ar[r, hookrightarrow] \ar[d, dotted, "s_2", "{\rotatebox[origin=c]{90}{$\sim$}}"']
 & \Scal
\\
 & \Scal^{2,3} \ar[ur, hookrightarrow] \ar[r, hookrightarrow]
 & \Scal^3 \ar[ur, hookrightarrow]
 &
\end{tikzcd}
\]
All cells in $\Scal^{1,2,3}$ have dimension $0$, those in
$\Scal^{1,2}$, $\Scal^{1,3}$ and $\Scal^{2,3}$ have dimension at most
$1$, those in $\Scal^1$, $\Scal^2$ and $\Scal^3$ dimension at
most $2$, and those in $\Scal$ dimension at most $3$.
\end{example}

\subsubsection{Categories of cubical categories}
A \emph{morphism} $f: \Scal \to \Scal'$ \emph{of $\omega$-categories} $\Scal$ and $\Scal'$ is a morphism of the
underlying categories, for each $i \in \Nbb_+$, which
preserves symmetries restricted to their types. Hence, for $i \geq 1$
and $x \in \Scal^i$, $f s_i x = s'_i f x$.  This defines the category
$\SinCat{\omega}$ of single-set cubical $\omega$-categories.

Owing to the definition in~\eqref{I:AxiomMorphFace}, a morphism $f: \Scal \to \Scal'$ of
$\omega$-categories restricts to a map
$\Scal^I \to {\Scal'}^I$ for each $I \subseteq \Nbb_+$. These
restrictions also commute with the symmetries.

A \emph{morphism of $n$-categories} is defined as
above, but with symmetries ranging over $1 \leq i \leq n-1$.  This
defines the category $\SinCat{n}$ of single-set cubical
$n$-categories.

\begin{lemma}
\label{L:LemmaInvSymPropMorph}
Morphisms in $\SinCat{\omega}$ preserve typed reverse symmetry maps:
  $f \tilde{s}_i x = \tilde{s}'_i f x$, for all $f: \Scal \to \Scal'$,
  $i \in \Nbb_+$ and $x \in \Scal^{i+1}$.
\end{lemma}

\subsubsection{Truncations}
There is a truncation functor $U^n_m: \SinCat{n} \to \SinCat{m}$, for
all $1 \leq m \leq n$, which forgets the $k$-dimensional structure
for each $k > m$. It sends each $n$-category
$\Scal$ to the $m$-category on the set
$\Scal^{m+1,\dots,n}$ and restricts the ranges of face maps and
compositions to $1 \leq i \leq m$, as well as that of the symmetry
maps to $1\leq i \leq m-1$.

There is also a truncation functor
$U_m: \SinCat{\omega} \to \SinCat{m}$, which sends each $\omega$-category $\Scal$ to the $m$-category with set $\Scal^{>m}$ and restricts face and symmetry
maps as well as compositions as above.

\subsubsection{Properties of symmetries}

The following facts provide further structural properties of
symmetries and reverse symmetries. They have been proved with Isabelle.

\begin{lemma}
\label{L:LemmaInvSymProp0}
Let $\Scal$ be  an $\omega$-category.  For all
$i,j\in \Nbb_+$,
\begin{enumerate}[{\bf (i)}]
\item $\delta_{j+1}^\alpha s_j x = s_j \delta_j^\alpha x$ if $x \in \Scal^j$,
\item \label{I:LemmaSymCompo} $s_i (x \circ_i y) = s_i x \circ_{i+1} s_i y$ if $x,y \in \Scal^i$ and $\Delta_j(x,y)$,
\item Yang-Baxter: $s_i s_{i+1} s_i x = s_{i+1} s_i s_{i+1} x$ if $x \in \Scal^{i,i+1}$.
\end{enumerate}
\end{lemma}

An automatic proof of Lemma~\ref{L:LemmaInvSymProp0}\eqref{I:LemmaSymCompo}, called
\isa{sym-func1} in our Isabelle component, is shown in
Section~\ref{SS:cubical-isabelle}.

Using Axiom~\ref{SSS:DefCubSingleSetWithoutConn}\eqref{I:AxiomSymInv},
we can further derive properties for $\tilde{s}$ that are dual to the
symmetry axioms.

\begin{lemma}
\label{L:LemmaInvSymProp}
Let $\Scal$ be an $\omega$-category.
For all $i,j\in \Nbb_+$, 
\begin{enumerate}[{\bf (i)}]
\item \label{I:LemmaFaceInvSym} if $x \in \Scal^{j+1}$, then
\begin{eqn}{equation*}
\delta_i^\alpha \tilde{s}_j x =
\begin{cases}
\tilde{s}_j \delta_{j+1}^\alpha x & \text{if } i=j, \\
\tilde{s}_j \delta_j^\alpha x & \text{if } i=j+1, \\
\tilde{s}_j \delta_i^\alpha x & \text{otherwise}.
\end{cases}
\end{eqn}
\item \label{I:LemmaInvSymCompo} if $x,y \in \Scal^{i+1}$ and $\Delta_j(x,y)$, then
\begin{eqn}{equation*}
\tilde{s}_i (x \circ_j y) =
\begin{cases}
\tilde{s}_i x \circ_{i+1} \tilde{s}_i y & \text{if } j=i, \\
\tilde{s}_i x \circ_i \tilde{s}_i y & \text{if } j=i+1, \\
\tilde{s}_i x \circ_j \tilde{s}_i y & \text{otherwise},
\end{cases}
\end{eqn}
\item $\tilde{s}_i x = x$ if $x \in \Scal^{i,i+1}$,
\item if $|i-j| \geq 2$, $x \in \Scal^{i,j+1}$, $y \in \Scal^{i+1,j}$ and $z \in \Scal^{i+1,j+1}$, then
\begin{align*}
s_i \tilde{s}_j x = \tilde{s}_j s_i x, \qquad
\tilde{s}_i s_j y = s_j \tilde{s}_i y, \qquad
\tilde{s}_i \tilde{s}_j z = \tilde{s}_j \tilde{s}_i z,
\end{align*}
\item $\tilde{s}_i \tilde{s}_{i+1} \tilde{s}_i x = \tilde{s}_{i+1} \tilde{s}_i \tilde{s}_{i+1} x$ if $x \in \Scal^{i+1,i+2}$.
\end{enumerate}
\end{lemma}
Lemma~\ref{L:LemmaInvSymProp}(\ref{I:LemmaFaceInvSym}) can be
depicted, for $x \in \Scal^{j+1}$, as
\begin{eqn}{equation*}
\begin{array}{cccc}
	\begin{tikzcd}[global scale = 2.5 and 2 and 1 and 1]
	\ar[r, shorten <= -.45em, shorten >= -.3em] \ar[d, shorten <= -.6em, shorten >= -.3em] & j+1 \\
	j & 
	\end{tikzcd}
&
\qquad
	\begin{tikzcd}[global scale = 1.3 and 1.3 and 1 and 1, longer arrows = 0.4em and 0.4em]
	\phantom{\circ} \ar[rr, equal] \ar[dd] && \phantom{\circ} \ar[dd] \\
	& x & \\
	\phantom{\circ} \ar[rr, equal] && \phantom{\circ}
	\end{tikzcd}
&
\
\overset{\tilde{s_j}}{\mapsto} \
&
	\begin{tikzcd}[global scale = 1.3 and 1.3 and 1 and 1, longer arrows = 0.4em and 0.4em]
	\phantom{\circ} \ar[rr] \ar[dd, equal] && \phantom{\circ} \ar[dd, equal] \\
	& \tilde{s}_j x & \\
	\phantom{\circ} \ar[rr] && \phantom{\circ}
	\end{tikzcd}
\end{array}
\end{eqn}

An interactive proof of the third case in this part of
Lemma~\ref{L:LemmaInvSymProp}, called \isa{inv-sym-face} in our
Isabelle component, is also shown in
Subsection~\ref{SS:cubical-isabelle}.

\subsection{Connections}
\label{SS:SingleSetWithConnections}

We now add connections to cubical categories, translating
the approach of Al-Agl, Brown and Steiner~\cite{AlAglBrownSteiner2002}.

\subsubsection{}
\label{SSS:DefCubSingleSetWithConn}
\emph{Connection maps} for a cubical $\omega$-category $\Scal$ are 
  maps $\gamma_i^\alpha : \Scal\to\Scal$, for $i\in
\Nbb_+$ and $\alpha \in \{-,+\}$, satisfying, for all $i,j\in \Nbb_+$,
\begin{enumerate}[{\bf (i)}]
\item \label{I:AxiomConnFaces}
$\delta_j^\alpha \gamma_j^\alpha x = x$,
$\delta_{j+1}^\alpha \gamma_j^\alpha x = s_j x$ and 
$\delta_i^\alpha \gamma_j^\beta x = \gamma_j^\beta \delta_i^\alpha x$ if $i \neq j,j+1$ and $x \in \Scal^j$, 
\item \label{I:AxiomConnCompoCorner} if $j \neq i, i+1$ and $x,y \in \Scal^i$, then
\begin{align*}
\Delta_{i+1}(x,y) &~\Rightarrow~ \gamma_i^+ (x \circ_{i+1} y) = (\gamma_i^+ x \circ_{i+1} s_i x) \circ_i (x \circ_{i+1} \gamma_i^+ y), \\
\Delta_{i+1}(x,y) &~\Rightarrow~ \gamma_i^- (x \circ_{i+1} y) = (\gamma_i^- x \circ_{i+1} y) \circ_i (s_i y \circ_{i+1} \gamma_i^- y), \\
\Delta_j(x,y) &~\Rightarrow~ \gamma_i^\alpha (x \circ_j y) = \gamma_i^\alpha x \circ_j \gamma_i^\alpha y,
\end{align*}
\item \label{I:AxiomConnStab} $\gamma_i^\alpha x = x$ if $x \in \Scal^{i,i+1}$, 
\item \label{I:AxiomConnZigzag} $\gamma_i^+ x \circ_{i+1} \gamma_i^- x = x$ and $\gamma_i^+ x \circ_i \gamma_i^- x = s_i x$ if $x \in \Scal^i$,
\item \label{I:AxiomConnBraid} $\gamma_i^\alpha \gamma_j^\beta x = \gamma_j^\beta
  \gamma_i^\alpha x$ if $|i-j| \geq 2$ and $x \in \Scal^{i,j}$,
\item \label{I:AxiomConnShift} $s_{i+1} s_i \gamma_{i+1}^\alpha x = \gamma_i^\alpha s_{i+1} x$ if $x \in \Scal^{i,i+1}$.
\end{enumerate}

\emph{Connection maps} for a cubical $n$-category $\Scal$ are maps 
$\gamma_i^\alpha: \Scal \to \Scal$ in the range $1 \leq i \leq
n-1$ satisfying the above axioms within appropriate ranges.

As for symmetries, we henceforth assume that upper indices of
connection maps range over $\{-,+\}$.

Connection maps $\gamma_i^\alpha$ restrict to the type
$\Scal^i \to \Scal$, and we restrict our attention to these
types.

\subsubsection{Explanation of axioms}

Henceforth we illustrate connection maps as
\begin{eqn}{equation*}
\begin{tikzcd}[global scale = 2.5 and 2 and 1 and 1]
\ar[r, shorten <= -.45em, shorten >= -.3em] \ar[d, shorten <= -.6em, shorten >= -.3em] & i+1 \\
i & 
\end{tikzcd}
\qquad
\begin{tikzcd}[global scale = 1.3 and 1.3 and 1 and 1, longer arrows = 0.4em and 0.4em]
\phantom{\circ} \ar[rr, equal] \ar[dd, equal] && \phantom{\circ} \ar[dd] \\
 & \gamma_i^+ x & \\
\phantom{\circ} \ar[rr, "x"'] && \phantom{\circ}
\end{tikzcd}
\ = \
\vcenter{\hbox{
\begin{tikzpicture}[global scale = 1 and 1]
\node [] () at (0.5,-0.25) {$x$};
\node [] () at (0.5,1.25) {};
\draw [-] (0,0) -- (1,0);
\draw [-] (0,1) -- (1,1);
\draw [-] (0,0) -- (0,1);
\draw [-] (1,0) -- (1,1);
\draw [-] (0.5,0.1) -- (0.5,0.5) -- (0.9,0.5);
\end{tikzpicture}
}}
\qquad
\begin{tikzcd}[global scale = 1.3 and 1.3 and 1 and 1, longer arrows = 0.4em and 0.4em]
\phantom{\circ} \ar[rr, "x"] \ar[dd] && \phantom{\circ} \ar[dd, equal] \\
 & \gamma_i^- x & \\
\phantom{\circ} \ar[rr, equal] && \phantom{\circ}
\end{tikzcd}
\ = \
\vcenter{\hbox{
\begin{tikzpicture}[global scale = 1 and 1]
\node [] () at (0.5,1.25) {$x$};
\node [] () at (0.5,-0.25) {};
\draw [-] (0,0) -- (1,0);
\draw [-] (0,1) -- (1,1);
\draw [-] (0,0) -- (0,1);
\draw [-] (1,0) -- (1,1);
\draw [-] (0.1,0.5) -- (0.5,0.5) -- (0.5,0.9);
\end{tikzpicture}
}}
\end{eqn}
The two diagrams on the left of the equations follow the style of
previous sections, those on the right are standard in the
literature~\cite{BrownSpencer1976,BrownHiggins1981,AlAglBrownSteiner2002}.
Fixed points in direction $i$ are shown as
\begin{eqn}{equation*}
\begin{tikzcd}[global scale = 2.5 and 2 and 1 and 1]
\ar[r, shorten <= -.45em, shorten >= -.3em] \ar[d, shorten <= -.6em, shorten >= -.3em] & i+1 \\
i & 
\end{tikzcd}
\qquad
\vcenter{\hbox{
\begin{tikzpicture}[global scale = 1 and 1]
\node [] () at (0.5,-0.25) {$x$};
\draw [-] (0,0) -- (1,0);
\draw [-] (0,1) -- (1,1);
\draw [-] (0,0) -- (0,1);
\draw [-] (1,0) -- (1,1);
\draw [-] (0.5,0.1) -- (0.5,0.9);
\end{tikzpicture}
}}
\end{eqn}

Using these diagrams, we can explain some of the axioms with
connections.

\begin{itemize}
\item Axiom~(\ref{I:AxiomConnFaces}) determines the cubical shape of $\gamma_i^\alpha x$:
\begin{eqn}{equation*}
\begin{tikzcd}[global scale = 2.5 and 2 and 1 and 1]
\ar[r, shorten <= -.45em, shorten >= -.3em] \ar[d, shorten <= -.6em, shorten >= -.3em] & i+1 \\
i & 
\end{tikzcd}
\qquad
\gamma_i^+ x
\ = \
\begin{tikzcd}[global scale = 1.3 and 1.3 and 1 and 1, longer arrows = 0.4em and 0.4em]
\phantom{\circ} \ar[rr, equal] \ar[dd, equal] && \phantom{\circ} \ar[dd, "s_i x"] \\
 & \gamma_i^+ x & \\
\phantom{\circ} \ar[rr, "x"'] && \phantom{\circ}
\end{tikzcd}
\qquad
\gamma_i^- x
\ = \
\begin{tikzcd}[global scale = 1.3 and 1.3 and 1 and 1, longer arrows = 0.4em and 0.4em]
\phantom{\circ} \ar[rr, "x"] \ar[dd, "s_i x"'] && \phantom{\circ} \ar[dd, equal] \\
 & \gamma_i^- x & \\
\phantom{\circ} \ar[rr, equal] && \phantom{\circ}
\end{tikzcd}
\end{eqn}
\item The first two equations of the \emph{corner
    axiom}~(\ref{I:AxiomConnCompoCorner}) can be shown as
\begin{eqn}{equation*}
\begin{tikzcd}[global scale = 2.5 and 2 and 1 and 1]
\ar[r, shorten <= -.45em, shorten >= -.3em] \ar[d, shorten <= -.6em, shorten >= -.3em] & i+1 \\
i & 
\end{tikzcd}
\qquad
\vcenter{\hbox{
\begin{tikzpicture}[global scale = 1 and 1]
\node [] () at (0.5,-0.25) {$x \circ_{i+1} y$};
\node [] () at (0.5,1.25) {};
\draw [-] (0,0) -- (1,0);
\draw [-] (0,1) -- (1,1);
\draw [-] (0,0) -- (0,1);
\draw [-] (1,0) -- (1,1);
\draw [-] (0.5,0.1) -- (0.5,0.5) -- (0.9,0.5);
\end{tikzpicture}
}}
\ = \
\vcenter{\hbox{
\begin{tikzpicture}[global scale = 1 and 1]
\node [] () at (0.5,-0.25) {$x$};
\node [] () at (1.5,-0.25) {$y$};
\node [] () at (1,2.25) {};
\draw [-] (0,0) -- (2,0);
\draw [-] (0,0) -- (0,2);
\draw [-] (0,1) -- (2,1);
\draw [-] (0,2) -- (2,2);
\draw [-] (1,0) -- (1,2);
\draw [-] (2,0) -- (2,2);
\draw [-] (0.5,0.1) -- (0.5,0.9);
\draw [-] (1.1,1.5) -- (1.9,1.5);
\draw [-] (0.5,1.1) -- (0.5,1.5) -- (0.9,1.5);
\draw [-] (1.5,0.1) -- (1.5,0.5) -- (1.9,0.5);
\end{tikzpicture}
}}
\end{eqn}
\item Axioms~(\ref{I:AxiomConnFaces})
  and~(\ref{I:AxiomConnCompoCorner}) impose that connection maps are
  morphisms and hence functors with respect to the underlying
 categories in any direction $j \neq i, i+1$.
\item By Axiom~(\ref{I:AxiomConnStab}), $\gamma_i^\alpha$ is an
  identity map on $\Scal^{i,i+1}$.
\item The \emph{zigzag axiom},~(\ref{I:AxiomConnZigzag}) is depicted
  as
\begin{eqn}{equation*}
\begin{tikzcd}[global scale = 2.5 and 2 and 1 and 1]
\ar[r, shorten <= -.45em, shorten >= -.3em] \ar[d, shorten <= -.6em, shorten >= -.3em] & i+1 \\
i & 
\end{tikzcd}
\qquad
\vcenter{\hbox{
\begin{tikzpicture}[global scale = 1 and 1]
\node [] () at (0.5,-0.25) {$x$};
\node [] () at (0.5,1.25) {};
\draw [-] (0,0) -- (2,0);
\draw [-] (0,1) -- (2,1);
\draw [-] (0,0) -- (0,1);
\draw [-] (1,0) -- (1,1);
\draw [-] (2,0) -- (2,1);
\draw [-] (0.5,0.1) -- (0.5,0.5) -- (0.9,0.5);
\draw [-] (1.1,0.5) -- (1.5,0.5) -- (1.5,0.9);
\end{tikzpicture}
}}
\ = \
\vcenter{\hbox{
\begin{tikzpicture}[global scale = 1 and 1]
\node [] () at (0.5,-0.25) {$x$};
\node [] () at (0.5,1.25) {};
\draw [-] (0,0) -- (1,0);
\draw [-] (0,1) -- (1,1);
\draw [-] (0,0) -- (0,1);
\draw [-] (1,0) -- (1,1);
\draw [-] (0.5,0.1) -- (0.5,0.9);
\end{tikzpicture}
}}
\end{eqn}
\item The \emph{braiding axiom} (\ref{I:AxiomConnBraid}) is hard to
  illustrate: four dimensions would be required for drawing it.
\item Finally, Axiom~\eqref{I:AxiomConnShift} relates connections in
  different directions:
\begin{eqn}{equation*}
\begin{tikzcd}[global scale = 1.3 and 1.2 and 1 and 1]
& i+2 & \\
\ar[ur, shorten <= -.25em, shorten >= -.2em] \ar[rr, shorten <= -.45em, shorten >= -.3em] \ar[dd, shorten <= -.6em, shorten >= -.3em] && i+1 \\\\
i && 
\end{tikzcd}
\qquad
\gamma_{i+1}^+ x \  = \ 
\vcenter{\hbox{
\begin{tikzpicture}[global scale = 1 and 1]
\draw [-] (0,0) -- (1,0);
\draw [-] (0,0) -- (0,1);
\draw [-] (0,1) -- (1,1);
\draw [-] (1,0) -- (1,1);
\draw [-] (0.5,1.5) -- (1.5,1.5);
\draw [-] (1.5,0.5) -- (1.5,1.5);
\draw [-] (0,1) -- (0.5,1.5);
\draw [-] (1,0) -- (1.5,0.5);
\draw [-] (1,1) -- (1.5,1.5);
\draw [-] (0.95,0.95) -- (0.75,0.75) -- (1.15,0.75);
\end{tikzpicture}
}}
\ \overset{s_i}{\mapsto} \
\vcenter{\hbox{
\begin{tikzpicture}[global scale = 1 and 1]
\draw [-] (0,0) -- (1,0);
\draw [-] (0,0) -- (0,1);
\draw [-] (0,1) -- (1,1);
\draw [-] (1,0) -- (1,1);
\draw [-] (0.5,1.5) -- (1.5,1.5);
\draw [-] (1.5,0.5) -- (1.5,1.5);
\draw [-] (0,1) -- (0.5,1.5);
\draw [-] (1,0) -- (1.5,0.5);
\draw [-] (1,1) -- (1.5,1.5);
\draw [-] (0.95,0.95) -- (0.75,0.75) -- (0.75,0.35);
\end{tikzpicture}
}}
\ \overset{s_{i+1}}{\mapsto} \
\vcenter{\hbox{
\begin{tikzpicture}[global scale = 1 and 1]
\draw [-] (0,0) -- (1,0);
\draw [-] (0,0) -- (0,1);
\draw [-] (0,1) -- (1,1);
\draw [-] (1,0) -- (1,1);
\draw [-] (0.5,1.5) -- (1.5,1.5);
\draw [-] (1.5,0.5) -- (1.5,1.5);
\draw [-] (0,1) -- (0.5,1.5);
\draw [-] (1,0) -- (1.5,0.5);
\draw [-] (1,1) -- (1.5,1.5);
\draw [-] (1.15,0.75) -- (0.75,0.75) -- (0.75,0.35);
\end{tikzpicture}
}}
\  = \  \gamma_i^+ s_{i+1} x
\end{eqn}
\end{itemize}

\subsubsection{Category of $\omega$-categories with connections}
A \emph{morphism} $f: \Scal \to \Scal'$ \emph{of $\omega$-categories with connections} 
is a morphism of $\omega$-categories that preserves
connections: $f \gamma_i^\alpha x = {\gamma'}_i^\alpha f x$, for all
$i\in \Nbb_+$ and~$x \in \Scal^i$. This defines the category
$\SinCatG{\omega}$ of single-set cubical $\omega$-categories with
connections.
In a same way, we define the category $\SinCatG{n}$ of $n$-categories with connections.

For $1 \leq m \leq n$, there is a truncation functor
$U^n_m: \SinCatG{n} \to \SinCatG{m}$ that forgets the $k$-dimensional
structure for $k > m$. It sends any $n$-category with connections $\Scal$ to the 
$m$-category with connections with set $\Scal^{m+1,\dots,n}$, keeping
only the face maps and compositions up to $m$ and the symmetries and
connections up to $m-1$.

There is also a truncation functor
$U_m: \SinCatG{\omega} \to \SinCatG{m}$ that forgets the
$k$-dimensional structure for $k > m$. It sends any 
$\omega$-category with connections $\Scal$ to the 
$m$-category with connections with set~$\Scal^{>m}$, keeping only the
face maps and compositions indexed up to $m$ and the symmetries and
connections indexed up to $m-1$.

\medskip

The following properties of connections have been proved using
Isabelle.
\begin{lemma}
\label{L:LemmaInvConnProp0}
Let $\Scal$ be an $\omega$-category with
connections. For all $i,j\in \Nbb_+$,
\begin{enumerate}[{\bf (i)}]
\item $\delta_j^\alpha \gamma_j^{-\alpha} x = \delta_{j+1}^\alpha x$ and $\delta_{j+1}^\alpha \gamma_j^{-\alpha} x = \delta_{j+1}^\alpha x$ if $x \in \Scal^i$,
\item if $j \neq i, i+1$, $x,y \in \Scal^i$ and $\Delta_{i+1}(x,y)$, then
\begin{align*}
\gamma_i^+ (x \circ_{i+1} y) = (\gamma_i^+ x \circ_i x)
  \circ_{i+1} (s_i x \circ_i \gamma_i^+ y), \quad\text{ and }\quad
\gamma_i^- (x \circ_{i+1} y) = (\gamma_i^- x \circ_i s_i y) \circ_{i+1} (y \circ_i \gamma_i^- y),
\end{align*}
\item \label{I:LemmaConnSym} $\gamma_i^\alpha s_j x = s_j \gamma_i^\alpha x$ and $\gamma_i^\alpha \tilde{s}_j y = \tilde{s}_j \gamma_i^\alpha y$ if $|i-j| \geq 2$, $x \in \Scal^{i,j}$ and $y \in \Scal^{i,j+1}$,
\item $\tilde{s}_i \tilde{s}_{i+1} \gamma_i^\alpha x = \gamma_{i+1}^\alpha \tilde{s}_{i+1} x$ if $x \in \Scal^{i,i+2}$.
\end{enumerate}
\end{lemma}

\subsection{Inverses}
\label{SS:SinSetWithInv}

Applications in higher rewriting and homotopy type theory require
cubical $(\omega,p)$-categories, where cells of dimension strictly
greater than $p$ are invertible. More specifically, $p=0$ is needed
for homotopy type theory~ \cite{BezemCoquandHuber2014} and the
categorification of abstract rewriting, while the categorification of
string rewriting requires $p=1$, and that of term and diagram
rewriting $p=2$~\cite{Polybook2024}. Here, we translate an approach
introduced by Lucas~\cite{LucasPhD2017} to single-sets, following his notation closely.

\subsubsection{Single-set cubical categories with inverses}
\label{SSS:DefInverse}
A cell $x$ of an $\omega$-category with connections
$\Scal$ is \emph{$r_i$-invertible}, for $i\in \Nbb_+$, if there exists
a cell $y\in \Scal$ such that
\[
\Delta_i(x,y),
\qquad
x \circ_i y = \delta_i^- x,
\qquad
\Delta_i(y,x),
\qquad
y \circ_i x = \delta_i^+ x.
\]
A cell $x \in \Scal^{>n}$ \emph{has an $r_i$-invertible
  $n-1$-shell}, for $k,i\in \Nbb_+$, if $\delta_j^\alpha x$ is
$r_i$-invertible for each $1 \leq j \leq n$ with $j \neq i$.

A \emph{single-set cubical $(\omega,p)$-category (with connections)}
is a single-set cubical $(\omega,p)$-category (with connections)
$\Scal$ such that every cell in $\Scal^{>n}$ with an $r_i$-invertible
$n-1$-shell is $r_i$-invertible for all $n\geq p+1$ and
$1 \leq i \leq n$.

We have shown with Isabelle that the $r_i$-inverse of any $x$ is
uniquely defined. We therefore write $r_i x$ for the $r_i$-inverse of
$x$.

We restrict these definitions to $n$-categories
with connections as usual, removing the indices outside the range
$1\leq i\leq n$. In particular, a \emph{single-set cubical
  $(n,p)$-category} is a single-set cubical
$n$-category~$\Scal$ such that every cell in
$\Scal^{k+1,\dots,n}$ with an $r_i$-invertible $k-1$-shell is
$r_i$-invertible for all $p+1 \leq k \leq n$ and~$1 \leq i \leq k$,

\subsubsection{Category of $(\omega,p)$-categories}
  A \emph{morphism of $(\omega,p)$-categories} is
  simply a morphism in $\SinCatG{\omega}$.  Inverses are preserved
  because $\Delta_i(x, r_i x)$ implies $\Delta_i(f(x), f(r_i x))$ and
  $f(x) \circ_i f(r_i x) = \delta_i^- f(x)$, and likewise for the
  opposite order of composition. Thus $f(x)$ is $r_i$-invertible with
  inverse $f(r_i x)$.  The situation for $(n,p)$-categories is
  similar. This defines the categories $\SinCatG{(\omega,p)}$ and
  $\SinCatG{(n,p)}$.

\subsubsection{Explanation of inverse maps}
A $(1,0)$-category $\Scal$ is a category
with a map $r_1:\Scal\to\Scal$ such that $\Delta_1(x, r_1 x)$,
$\Delta_1(r_1 x, x)$, $x \circ_1 r_1 x = \delta_1^- x$ and
$r_1 x \circ_1 x = \delta_1^+ x$, for every $x \in \Scal$. This
defines a groupoid. The cell $x$ is sent by $r_1$ to the backward
arrow in the following diagram
\[
r_1 \left(\delta_1^- x \oto{x} \delta_1^+ x\right) ~=~ \delta_1^+ x \ofrom{} \delta_1^- x.
\]

A $(2,0)$-category $\Ccal$ is a
$2$-category with maps $r_1,r_2:\Scal\to\Scal$ such that
$\Delta_i(x,r_i x)$, $\Delta_i(r_i x,x)$,
$x \circ_i r_i x = \delta_i^- x$,
$r_i x \circ_i x = \delta_i^+ x$, for $i \in \{1, 2\}$ and
$x \in \Scal$. The inverses of $x$ are depicted as
\[
r_1
\left(\begin{tikzcd}[global scale = 1.3 and 1.3 and 1 and 1.1, longer arrows = 0.4em and 0.4em]
\phantom{\circ} \ar[rr, "\delta_1^- x"] \ar[dd, "\delta_2^- x"'] && \phantom{\circ} \ar[dd, "\delta_2^+ x"] \\
 & x & \\
\phantom{\circ} \ar[rr, "\delta_1^+ x"'] && \phantom{\circ}
\end{tikzcd}\right) \
=
\begin{tikzcd}[global scale = 1.3 and 1.3 and 1 and 1.1, longer arrows = 0.4em and 0.4em]
\phantom{\circ} \ar[rr, "\delta_1^+ x"] && \phantom{\circ} \\
 & r_1 x & \\
\phantom{\circ} \ar[uu, "r_1 \delta_2^- x"] \ar[rr, "\delta_1^- x"'] && \phantom{\circ} \ar[uu, "r_1 \delta_2^+ x"']
\end{tikzcd},
\qquad
r_2
\left(\begin{tikzcd}[global scale = 1.3 and 1.3 and 1 and 1.1, longer arrows = 0.4em and 0.4em]
\phantom{\circ} \ar[rr, "\delta_1^- x"] \ar[dd, "\delta_2^- x"'] && \phantom{\circ} \ar[dd, "\delta_2^+ x"] \\
 & x & \\
\phantom{\circ} \ar[rr, "\delta_1^+ x"'] && \phantom{\circ}
\end{tikzcd}\right) \
=
\begin{tikzcd}[global scale = 1.3 and 1.3 and 1 and 1.1, longer arrows = 0.4em and 0.4em]
\phantom{\circ} \ar[dd, "\delta_2^+ x"'] && \phantom{\circ} \ar[ll, "r_2 \delta_1^- x"'] \ar[dd, "\delta_2^- x"] \\
 & r_2 x & \\
\phantom{\circ} && \phantom{\circ} \ar[ll, "r_2 \delta_1^+ x"]
\end{tikzcd}.
\]

\subsubsection{Truncation}
The truncation functors $U^n_m: \SinCatG{n} \to \SinCatG{m}$ and $U_m:
\SinCatG{\omega} \to \SinCatG{m}$ induce functors $U^n_m:
\SinCatG{(n,p)} \to \SinCatG{(m,p)}$ and $U_m: \SinCatG{(\omega,p)}
\to \SinCatG{(m,p)}$ for $p \leq m \leq n$.

\begin{remark}
  For $i \in \Nbb_+$, every cell in $\Scal^i$ is its own
  $r_i$-inverse.
\end{remark}

\begin{lemma}
  Let $\Scal$ be an $(\omega,p)$-category with
  connections. For every $i,j\in \Nbb_+$ and every $r_i$-invertible
  $x,y \in \Scal$,
\begin{enumerate}[{\bf (i)}]
\item \label{I:PropRiCompatFace} $\delta_i^\alpha r_i x = \delta_i^{-\alpha} x$ and $\delta_j^\alpha r_i x = r_i \delta_j^\alpha x$ if $j \neq i$,
\item if $j \neq i$, then $r_i (x \circ_i y) = r_i y \circ_i r_i x$ if $\Delta_i(x,y)$, and $r_i (x \circ_j y) = r_i x \circ_j r_i y$ if $\Delta_j(x,y)$,
\item $r_i s_{i-1} x = s_{i-1} x$ and $r_i s_j y = s_j r_i y$ if $j \neq i-1$, $x \in \Scal^{i-1}$ and $y \in \Scal^j$,
\item $r_i \tilde{s}_i x = \tilde{s}_i x$ and $r_i \tilde{s}_j y = \tilde{s}_j r_i y$ if $j \neq i$, $x \in \Scal^{i+1}$ and $y \in \Scal^{j+1}$,
\item $r_i \gamma_j^\alpha x = \gamma_j^\alpha r_i x$ if $i \neq j,j+1$ and $x \in \Scal^j$.
\end{enumerate}
\end{lemma}
\begin{proof}
  We prove the first item of \eqref{I:PropRiCompatFace} as an
  example. By $r_i$-invertibility of $x$, $\Delta_i(x,r_i x)$ and
  $\Delta_i(r_i x,x)$. Hence $\delta_i^+ x = \delta_i^- r_i x$ and
  $\delta_i^+ r_i x = \delta_i^- x$. Suppose $j \neq i$. Then
  $z = \delta_j^\alpha r_i x$ satisfies
\[
\Delta_i(\delta_j^\alpha x,z),
\qquad
\delta_j^\alpha x \circ_i z = \delta_i^- \delta_j^\alpha x,
\qquad
\Delta_i(z,\delta_j^\alpha x),
\qquad
z \circ_i \delta_j^\alpha x = \delta_i^+ \delta_j^\alpha x.
\]
Thus $z = r_i \delta_j^\alpha x$, because $r_i$-inverses are unique.
The remaining proofs are similar.
\end{proof}

\begin{lemma}
\label{L:SingleSetInverseFaceStab}
Let $\Scal$ be an $(\omega,p)$-category. If
$x \in \Scal^{>n}$ is $r_i$-invertible, then $r_i x \in \Scal^{>n}$.
\end{lemma}
\begin{proof}
  Suppose $x$ is $r_i$-invertible. For any $m \geq n+1$, apply
  $\delta_m^\alpha$ to $x \circ_i r_i x = \delta_i^- x$ and
  $x \circ_i r_i x = \delta_i^- x$. This yields
  $\delta_m^\alpha r_i x \circ_i x = \delta_i^+ x$. Uniqueness of
  $r_i x$ then implies that $\delta_m^\alpha r_i x = r_i x$.
\end{proof}

\begin{proposition}
  \label{P:OmegaZeroInv}
Every cell in an $(\omega,0)$-category  is
$r_i$-invertible for each $i\in \Nbb_+$.
\end{proposition}
\begin{proof}
By Axiom~\eqref{I:AxiomFiniteDimCells},
$\Scal = \bigcup_{n\geq 0} \Scal^{>n}$. We show by induction on the
dimension $n$ of cells that every cell $x \in \Scal^{>n}$ is
$r_i$-invertible for all $i\in \Nbb_+$. 

For $n=0$, suppose $x \in \Scal^{>0}$. Then $x$ is its own
$r_i$-inverse for all $i\in \Nbb_+$.
  
Suppose the property holds for $n-1$. Let $i \in \Nbb_+$ and $x \in \Scal^{>n}$. If $i \geq n+1$, then $x$ is its own $r_i$-inverse, so suppose $i \leq n$. Let $1 \leq j \leq n$ with $i \neq j$. We have $\delta_j^\alpha x \in \Scal^{j,n+1,n+2,\dots}$ and thus $s_{n-1} \dots s_j \delta_j^\alpha x \in \Scal^{>n}$ using Axioms~(\ref{I:AxiomSymType}) and~(\ref{I:AxiomFaceSym}) of Definition~\ref{SSS:DefCubSingleSetWithoutConn}. There are two cases depending on the value of $i$.
\begin{itemize}
\item If $j < i \leq n$, then $s_{n-1} \dots s_j \delta_j^\alpha x$
  has an $r_{i-1}$-inverse $y$ by the induction hypothesis. Writing $z =
  \tilde{s}_j \dots \tilde{s}_{n-1} y$ and using property~\eqref{I:LemmaInvSymCompo} of \cref{L:LemmaInvSymProp},
\begin{align*}
z \circ_i \delta_j^\alpha x
&= \tilde{s}_j \dots \tilde{s}_{n-1} y \circ_i \tilde{s}_j \dots \tilde{s}_{n-1} s_{n-1} \dots s_j \delta_j^\alpha x \\
&= \tilde{s}_j \dots \tilde{s}_{n-1} (y \circ_{i-1} s_{n-1} \dots s_j \delta_j^\alpha x) \\
&= \tilde{s}_j \dots \tilde{s}_{n-1} \delta_{i-1}^+ s_{n-1} \dots s_j \delta_j^\alpha x \\
&= \delta_i^+ \delta_j^\alpha x.
\end{align*}
Similarly $\delta_j^\alpha x \circ_i z = \delta_i^- \delta_j^\alpha x$, $\Delta_i(z, \delta_j^\alpha x)$ and $\Delta_i(\delta_j^\alpha x, z)$. So $z$ is the $r_i$-inverse of $\delta_j^\alpha x$.
\item If $i < j$ then $s_{n-1} \dots s_j \delta_j^\alpha x$ has an
  $r_i$-inverse $y$ by the induction hypothesis. Using the same
  abbreviation and lemma as before, 
\begin{align*}
z \circ_i \delta_j^\alpha x
&= \tilde{s}_j \dots \tilde{s}_{n-1} y \circ_i \tilde{s}_j \dots \tilde{s}_{n-1} s_{n-1} \dots s_j \delta_j^\alpha x \\
&= \tilde{s}_j \dots \tilde{s}_{n-1} (y \circ_i s_{n-1} \dots s_j \delta_j^\alpha x) \\
&= \tilde{s}_j \dots \tilde{s}_{n-1} \delta_i^+ s_{n-1} \dots s_j \delta_j^\alpha x \\
&= \delta_i^+ \delta_j^\alpha x,
\end{align*}
using property~\eqref{I:LemmaInvSymCompo} of
\cref{L:LemmaInvSymProp}. Likewise,
$\delta_j^\alpha x \circ_i z = \delta_i^- \delta_j^\alpha x$,
$\Delta_i(z, \delta_j^\alpha x)$ and $\Delta_i(\delta_j^\alpha x, z)$,
and it follows again that $z$ is the $r_i$-inverse of
$\delta_j^\alpha x$.
\end{itemize}
This shows that $x$ has an $r_i$-invertible $n-1$-shell and is therefore $r_i$-invertible.
\end{proof}

A formalisation of this proof with Isabelle is shown in
Subsection~\ref{SS:isaatwork}.

\section{Equivalence with classical cubical categories}
\label{S:EquivalenceWithCubicalCategories}

We now present our main theorems: a series of equivalences which
justify our single-set axioms relative to the classical ones.
We begin in Subsection~\ref{SS:CubicalCategories} by recalling the classical 
axioms for cubical categories. For cubical categories with connection, we use the axioms of
Al-Agl, Brown and Steiner~\cite{AlAglBrownSteiner2002}; for inverses,
we follow Lucas~\cite{LucasPhD2017}. These categories, with
appropriate functors between them, form the categories
$\CubCat{\omega}$ without connections, $\CubCatG{\omega}$ with
connections and $\CubCatG{(\omega,p)}$ with inverses.  In
Subsection~\ref{SS:EquivalenceWithCubicalCategories}, we prove an
equivalence of categories $\SinCat{\omega} \simeq \CubCat{\omega}$, in
Subsection~\ref{SS:EquivalenceWithCubicalCategoriesWithConnections} we
extend it to $\SinCatG{\omega}\simeq \CubCatG{\omega}$ and in
Subsection~\ref{SS:EquivalenceWithCubicalCategoriesWithInverses}
further to $\SinCatG{(\omega,p)} \simeq \CubCatG{(\omega,p)}$.

Straightforward modifications of these proofs lead to equivalences
between the corresponding $n$-categories. We do not list
them explicitly.

\subsection{Cubical categories}
\label{SS:CubicalCategories}

First we recall the classical definitions of cubical categories as
cubical sets with cell compositions.

\subsubsection{}
\label{SSS:DefClassCubCat}
A \emph{(cubical) $\omega$-category} is a family
$\Ccal=(\Ccal_n)_{n \in \Nbb}$ of sets of \emph{$n$-cells}
with \emph{face maps}
$\partial_{n,i}^\alpha : \Ccal_n \to \Ccal_{n-1}$, \emph{degeneracy
  maps} $\epsilon_{n,i}:\Ccal_{n-1}\to\Ccal_n$ and \emph{compositions}
$\star_{n,i}:\Ccal_n\times_{n,i}\Ccal_n \fl \Ccal_n$, for
$1\leq i\leq n$, where $\Ccal_n\times_{n,i}\Ccal_n$ denotes the
pullback of the cospan
$\partial_{n,i}^+:\Ccal_n\to \Ccal_{n-1}\leftarrow \Ccal_n:
\partial_{n,i}^-$. These satisfy, for all $1\leq i,j \leq n$,
\begin{enumerate}[{\bf (i)}]
\label{E:AxiomClassCubCat}
\item $a\star_{n,i}(b\star_{n,i}c) = (a\star_{n,i}b)\star_{n,i}c$ if either side is defined,
\item $a\star_{n,i}\epsilon_{n,i}\partial_{n,i}^+a = \epsilon_{n,i}\partial_{n,i}^-a\star_{n,i}a = a$,
\item $\partial_{n-1,i}^\alpha\partial_{n,j}^\beta = \partial_{n-1,j-1}^\beta\partial_{n,i}^\alpha$ if $i < j$,
\item \label{I:AxiomFaceCompoClass} if $a,b$ are $\star_{n,j}$-composable then
\begin{eqn}{equation*}
\partial_{n,i}^\alpha(a\star_{n,j}b) =
\begin{cases}
\partial_{n,i}^\alpha a\star_{n,j-1}\partial_{n,i}^\alpha b & \quad\mathrm{if~}i<j, \\
\partial_{n,i}^-a & \quad\mathrm{if~}i=j\mathrm{~and~}\alpha=-, \\
\partial_{n,i}^+b & \quad\mathrm{if~}i=j\mathrm{~and~}\alpha=+, \\
\partial_{n,i}^\alpha a\star_{n,j}\partial_{n,i}^\alpha b & \quad\mathrm{if~}i>j,
\end{cases}
\end{eqn}
\item if $a,b$ are $\star_{n,i}$-composable, $c,d$ are $\star_{n,i}$-composable, $a,c$ are $\star_{n,j}$-composable and $b,d$ are $\star_{n,j}$-composable, then $(a\star_{n,i}b)\star_{n,j}(c\star_{n,i}d) = (a\star_{n,j}c)\star_{n,i}(b\star_{n,j}d)$,
\item
\begin{eqn}{equation*}
\partial_{n,i}^\alpha\epsilon_{n,j} =
\begin{cases}
\epsilon_{n-1,j-1}\partial_{n-1,i}^\alpha & \quad\text{if~}i<j, \\
\id_{C_{n-1}} & \quad\text{if~}i=j, \\
\epsilon_{n-1,j}\partial_{n-1,i-1}^\alpha & \quad\text{if~}i>j,
\end{cases}
\end{eqn}
\item if $a,b$ are $\star_{n,j}$-composable then
\begin{eqn}{equation*}
\epsilon_{n+1,i}(a\star_{n,j}b) =
\begin{cases}
\epsilon_{n+1,i}a\star_{n+1,j+1}\epsilon_{n+1,i}b & \quad\mathrm{if~}i\leq j, \\
\epsilon_{n+1,i}a\star_{n+1,j}\epsilon_{n+1,i}b & \quad\mathrm{if~}i>j,
\end{cases}
\end{eqn}
\item $\epsilon_{n+1,i}\epsilon_{n,j} = \epsilon_{n+1,j+1}\epsilon_{n,i}$ if $i\leq j$.
\end{enumerate}

A \emph{(cubical) $n$-category} is a family
$\Ccal=(\Ccal_k)_{k \leq n}$ of sets of \emph{$k$-cells} with
\emph{face maps} $\partial_{k,i}^\alpha : \Ccal_k \to \Ccal_{k-1}$,
\emph{degeneracy maps} $\epsilon_{k,i}:\Ccal_{k-1}\to\Ccal_k$ and
\emph{composition maps}
$\star_{k,i}:\Ccal_k\times_{k,i}\Ccal_k \fl \Ccal_k$, for
$1\leq i\leq k \leq n$. These satisfy the above axioms, removing those
involving face maps, compositions and degeneracies outside the
appropriate ranges.

Each degeneracy map $\epsilon_{n,i}$ yields identities for
$\star_{n,i}$.

\subsubsection{}
\label{SSS:DefClassCubCatConn}
A \emph{(cubical) $\omega$-category with connections} $\Ccal$ is an
$\omega$-category with \emph{connection maps} \linebreak 
$\Gamma_{n,i}^\alpha:\Ccal_{n-1} \to \Ccal_n$ for $1\leq i<n$,  such that
\begin{enumerate}[{\bf (i)}]
\item \label{I:AxiomFaceConnClass}
\begin{eqn}{equation*}
\partial_{n,i}^\alpha\Gamma_{n,j}^\beta =
\begin{cases}
\Gamma_{n-1,j-1}^\beta\partial_{n-1,i}^\alpha & \quad\mathrm{if~}i<j, \\
\id_{C_{n-1}} & \quad\mathrm{if~}i=j,j+1\mathrm{~and~}\alpha=\beta, \\
\epsilon_{n-1,j}\partial_{n-1,j}^\alpha & \quad\mathrm{if~}i=j,j+1\mathrm{~and~}\alpha=-\beta, \\
\Gamma_{n-1,j}^\beta\partial_{n-1,i-1}^\alpha & \quad\mathrm{if~}i>j+1,
\end{cases}
\end{eqn}
\item if $a,b$ are $\star_{n,j}$-composable then
\begin{eqn}{equation*}
\Gamma_{n+1,i}^\alpha(a\star_{n,j}b)=
\begin{cases}
\Gamma_{n+1,i}^\alpha a\star_{n+1,j+1}\Gamma_{n+1,i}^\alpha b & \quad\mathrm{if~}i<j, \\
(\Gamma_{n+1,i}^-a\star_{n+1,i}\epsilon_{n+1,i+1}b)\star_{n+1,i+1}(\epsilon_{n+1,i}b\star_{n+1,i}\Gamma_{n+1,i}^-b) & \quad\mathrm{if~}i=j\mathrm{~and~}\alpha=-, \\
(\Gamma_{n+1,i}^+a\star_{n+1,i}\epsilon_{n+1,i}a)\star_{n+1,i+1}(\epsilon_{n+1,i+1}a\star_{n+1,i}\Gamma_{n+1,i}^+b) & \quad\mathrm{if~}i=j\mathrm{~and~}\alpha=+, \\
\Gamma_{n+1,i}^\alpha a\star_{n+1,j}\Gamma_{n+1,i}^\alpha b & \quad\mathrm{if~}i>j,
\end{cases}
\end{eqn}
\item $\Gamma_{n,i}^+a\star_{n,i}\Gamma_{n,i}^-a=\epsilon_{n,i+1}a$ and $\Gamma_{n,i}^+a\star_{n,i+1}\Gamma_{n,i}^-a=\epsilon_{n,i}a$,
\item
\begin{eqn}{equation*}
\Gamma_{n+1,i}^\alpha\epsilon_{n,j}=
\begin{cases}
\epsilon_{n+1,j+1}\Gamma_{n,i}^\alpha & \quad\mathrm{if~}i<j, \\
\epsilon_{n+1,i}\epsilon_{n,i} & \quad\mathrm{if~}i=j, \\
\epsilon_{n+1,j}\Gamma_{n,i-1}^\alpha & \quad\mathrm{if~}i>j,
\end{cases}
\end{eqn}
\item
\begin{eqn}{equation*}
\Gamma_{n+1,i}^\alpha\Gamma_{n,j}^\beta=
\begin{cases}
\Gamma_{n+1,j+1}^\beta\Gamma_{n,i}^\alpha & \quad\mathrm{if~}i<j, \\
\Gamma_{n+1,i+1}^\alpha\Gamma_{n,i}^\alpha & \quad\mathrm{if~}i=j\mathrm{~and~}\alpha=\beta.
\end{cases}
\end{eqn}
\end{enumerate}

A \emph{(cubical) $n$-category with connections} is an
$n$-category with \emph{connection maps}
$\Gamma_{k,i}^\alpha:\Ccal_{k-1}\to\Ccal_k$ for~$1\leq i<k \leq n$,
satisfying axioms with the appropriate index restrictions.

There is some index shifting in the above axioms. For instance in the axiom describing the faces of compositions~\ref{SSS:DefClassCubCat}\eqref{I:AxiomFaceCompoClass} and the one describing the faces of connections~\ref{SSS:DefClassCubCatConn}\eqref{I:AxiomFaceConnClass}, there is a $j-1$ index appearing in the case where $j>i$, which does not appear in the other cases. Furthermore, there are two indices, for the dimension $k$ and the direction $i$. Their single-set counterparts~\ref{SSS:DefCubSingleSetWithoutConn}\eqref{I:DeltaCompoCompat} and \ref{SSS:DefCubSingleSetWithConn}\eqref{I:AxiomConnFaces} do not have such explicit index shifting and the dimension index is missing.

\subsubsection{Cubical categories with inverses}
\label{SS:DefGrpdCubCat}
Let $1 \leq i \leq n$. An $n$-cell $a$ of an $\omega$-category with connections $\Ccal$ is
\emph{$R_{n,i}$-invertible} if there is a $n$-cell $b$ such that
\[
a \star_{n,i} b = \epsilon_{n,i} \partial_{n,i}^- a
\quad \text{and}\quad
b \star_{n,i} a = \epsilon_{n,i} \partial_{n,i}^+ a. 
\]
Each $R_{n,i}$-inverse of an $n$-cell $a$ is unique and denoted
$R_{n,i} a$. A $n$-cell $a$ \emph{has an $R_{n-1,i}$-invertible shell}
if the cells $\partial_{n,j}^\alpha a$ are $R_{n-1,i-1}$-invertible
for all $1 \leq j < i$, and the cells $\partial_{n,j}^\alpha a$ are
$R_{n-1,i}$-invertible for all~$i < j \leq n$.

A \emph{(cubical) $(\omega,p)$-category (with connections)} $\Ccal$,
for $p\in \Nbb$, is an $\omega$-category with connections in
which, for all $n > p$ and $1 \leq i \leq n$, every $n$-cell with an
$R_{n-1,i}$-invertible shell is $R_{n,i}$-invertible.

The above definitions of invertibility and
invertible shell extend to $n$-categories with
connections, by removing the $R_{k,i}$ with indices $(k,i)$ whenever $k > n$.  In
particular, for $0\leq p\leq n$, a \emph{(cubical) $(n,p)$-category
  (with connections)} $\Ccal$ is an $n$-category with
connections in which, for all $p+1 \leq k \leq n$ and~$1 \leq i \leq k$, every $k$-cell with an $R_{k-1,i}$-invertible shell
is $R_{k,i}$-invertible.

\subsubsection{Categories of cubical categories}
A \emph{functor} $F:\Ccal\to\Dcal$ \emph{of $\omega$-categories} is a family of maps
$(F_n:\Ccal_n\to\Dcal_n)_{n \in \Nbb}$ that preserve all face,
degeneracy and composition maps:
\begin{align*}
  F_{n-1} \partial_{n,i}^\alpha
  = \partial_{n,i}^\alpha F_n,\qquad
  F_n (a \star_{n,i} b)
  = F_n a \star_{n,i} F_n b,\qquad
  F_n \epsilon_{n,i}
  = \epsilon_{n,i} F_{n-1},
\end{align*}
for all $1\leq i\leq n$ and $\star_{n,i}$-composable
$a,b \in \Ccal_n$.  Cubical $\omega$-categories and their functors
form the category $\CubCat{\omega}$.

Further, a \emph{functor of $\omega$-categories with
  connections} is a functor between the underlying 
  $\omega$-categories that preserves the connection maps:
$F_n \Gamma_{n,i}^\alpha = \Gamma_{n,i}^\alpha F_{n-1}$, for all
$1\leq i < n$.  Cubical $\omega$-categories with connections and their
functors form the category $\CubCatG{\omega}$.

Finally, a \emph{functor of $(\omega,p)$-categories}
$F:\Ccal\to\Dcal$ is a functor between the underlying 
$\omega$-categories with connections. As in the single-set case,
inverses are preserved by such functors.  Cubical
$(\omega,p)$-categories and their functors form the category
$\CubCatG{(\omega,p)}$.

The categories $\CubCat{n}$, $\CubCatG{n}$ and $\CubCatG{(n,p)}$ of
cubical $n$-categories and their morphisms are defined by truncation
as usual.

\subsection{Equivalence for cubical $\omega$-categories}
\label{SS:EquivalenceWithCubicalCategories}

We are now prepared for the following result.
\begin{theorem}
\label{T:EquivSinSetClass}
There is an equivalence of categories
\[
\begin{tikzcd}[global scale = 10 and 4 and 1 and 1]
\FC{(-)}  :  \SinCat{\omega} \ar[r, shift left] & \CubCat{\omega}  :  \FS{(-)} \ar[l, shift left].
\end{tikzcd}
\]
\end{theorem}

We develop the proof of $\SinCat{\omega}\simeq \CubCat{\omega}$ in the
remainder of this subsection.  The functors $\FC{(-)}$ and $ \FS{(-)}$
are defined in \eqref{SSS:FC} and \eqref{SSS:FS}, the natural
isomorphisms in \eqref{SSS:NaturalIsomorphism1} and
\eqref{SSS:NaturalIsomorphism2} below. 
We henceforth refer to classical and single-set $\omega$-categories to distinguish between
  objects in $\CubCat{\omega}$ and~$\SinCat{\omega}$.

\subsubsection{The functor $\FC{(-)}$}
\label{SSS:FC}
For each category $(\Scal,\delta,\circ,s)$ in $\SinCat{\omega}$, we
define the category
$\FC{(\Scal,\delta,\circ,s)} = (\FC{\Scal},\partial,\epsilon,\star)$
in $\CubCat{\omega}$ with
\begin{enumerate}[{\bf (i)}]
\item sets of $n$-cells  $\FC{\Scal}_n = \Scal^{>n}$ for all $n \in \Nbb$,
\item  $\partial_{n,i}^\alpha: \FC{\Scal}_n \to
  \FC{\Scal}_{n-1}$ such that $\partial_{n,i}^\alpha = s_{n-1} \dots
  s_i \delta_i^\alpha$ for all $1 \leq i \leq n$,
\item  $\epsilon_{n,i}: \FC{\Scal}_{n-1} \to\FC{\Scal}_n$
  such that $\epsilon_{n,i} = \tilde{s}_i \dots \tilde{s}_{n-1}$
  for all $1 \leq i \leq n$,
\item compositions $\star_{n,i}$ such that
  $x \star_{n,i} y = x \circ_i y$ if $\Delta_i (x, y)$ and 
  undefined otherwise, for all $1 \leq i \leq n$ and
  $x,y \in \FC{\Scal}_n$.
\end{enumerate}
With each morphism $f : \Scal \to \Scal'$ in $\SinCat{\omega}$,
$\FC{(-)}$ associates the functor $\FC{f}: \FC{\Scal} \to \FC{\Scal'}$
on $n$-cells in $\CubCat{\omega}$ as the restriction of $f$ to a map
from $\FC{\Scal}_n=\Scal^{>n}$ to $\FC{\Scal'}_n=\Scal'^{>n}$, for
each $n \in \Nbb$.

\begin{lemma}
\label{L:FC-well-defined}
The functor $\FC{(-)}$ is well-defined.
\end{lemma}
\begin{proof}
  We need to show that $\FC{(-)}: \SinCat{\omega} \to \CubCat{\omega}$
  sends each category in $\SinCat{\omega}$ to a category in
  $\CubCat{\omega}$ and each morphism in the former category to a
  functor in the latter, and that it is itself a functor.

  First we show that $(\FC{\Scal},\partial,\epsilon,\star)$ is indeed
  a classical $\omega$-category, that is, we check the axioms in
  \refeq{SSS:DefClassCubCat}.  Here we only consider some
  representative examples. Proofs for the remaining axioms shown in
  Appendix~\ref{A:FC-well-defined}.  Suppose $1 \leq i,j \leq n$ and
  $a,b,c,d \in (\FC{\Scal})_n$.
\begin{enumerate}[{\bf (i)}]
\item If $a \star_{n,i} (b \star_{n,i} c)$ is defined, then
  $a \star_{n,i} (b \star_{n,i} c) = a \circ_i (b \circ_i c) = (a
  \circ_i b) \circ_i c = (a \star_{n,i} b) \star_{n,i} c$ because
  $\Delta_i (a, b \star_{n,i} c)$ and $\Delta_i (b, c)$, using
  Axioms~\ref{SSS:SingleSetOneCategories}\eqref{I:AxiomAssoc},
  \eqref{I:AxiomLocality}, \eqref{I:AxiomFunctionality} among
  others. The same holds if $(a \star_{n,i} b) \star_{n,i} c$ is
  defined.
\item The compositions
  $a \star_{n,i} \epsilon_{n,i} \partial_{n,i}^+ a$ and
  $\epsilon_{n,i} \partial_{n,i}^- a \star_{n,i} a$ are defined because
  $\Delta_i(a, \epsilon_{n,i} \partial_{n,i}^+ a)$ and
  $\Delta_i (\epsilon_{n,i} \partial_{n,i}^- a, a)$, using
  Axioms~\ref{SSS:SingleSetOneCategories}\eqref{I:AxiomUnits},
  \eqref{I:AxiomLocality}, \eqref{I:AxiomFunctionality} among
  others. It follows that
  $a \star_{n,i} \epsilon_{n,i} \partial_{n,i}^+ a = a \circ_i
  \delta_i^+ a = a$ and
  $\epsilon_{n,i} \partial_{n,i}^- a \star_{n,i} a = \delta_i^- a
  \circ_i a = a$ because
  $\epsilon_{n,i} \partial_{n,i}^\alpha a = \delta_i^\alpha a$.
\item If $i<j$, then
  Axiom~\ref{SSS:DefCubSingleSetWithoutConn}\eqref{I:AxiomCommutativity}
  and other facts imply that
\[
\begin{array}{rcl}
\partial_{n-1,j-1}^\beta \partial_{n,i}^\alpha & = & s_{n-2} \dots s_{j-1} \delta_{j-1}^\beta s_{n-1} \dots s_i \delta_i^\alpha \\
 & = & s_{n-2} \dots s_{j-1} s_{n-1} \dots s_j \delta_{j-1}^\beta s_{j-1} \dots s_i \delta_i^\alpha \\
 & = & s_{n-2} \dots s_{j-1} s_{n-1} \dots s_j \delta_j^\beta s_{j-2} \dots s_i \delta_i^\alpha \\
 & = & s_{n-2} \dots s_i \delta_i^\alpha s_{n-1} \dots s_j \delta_j^\beta \\
 & = & \partial_{n-1,i}^\alpha \partial_{n,j}^\beta.
\end{array}
\]
\end{enumerate}

Next we show that $\FC{f}$ is a morphism in $\CubCat{\omega}$. Suppose
$1 \leq i \leq n$ and $a,b \in \FC{\Scal}_n$. Then
\begin{enumerate}[{\bf (i)}]
\item $\FC{f}_{n-1} \partial_{n,i}^\alpha = f s_{n-1} \dots s_i \delta_i^\alpha = s_{n-1} \dots s_i \delta_i^\alpha f = \partial_{n,i}^\alpha \FC{f}_n$,
\item if $\partial_{n,i}^+ a = \partial_{n,i}^- b$ then $\Delta_i(a,b)$ and $\FC{f}_n (a \star_{n,i} b) = f (a \circ_i b) = f (a) \circ_i f (b) = \FC{f}_n a \star_{n,i} \FC{f}_n b$,
\item $\FC{f}_n \epsilon_{n,i} = f \tilde{s}_i \dots \tilde{s}_{n-1} = \tilde{s}_i \dots \tilde{s}_{n-1} f = \epsilon_{n,i} \FC{f}_{n-1}$.
\end{enumerate}
Further, it is clear from the definition that $\FC{f}$ preserves
sources, targets and compositions of morphisms in $\CubCat{\omega}$.
\end{proof}

All calculations in this proof are performed within
$\SinCat{\omega}$. Their formalisation with our Isabelle
  component seems therefore routine.

\subsubsection{The functor $\FS{(-)}$}
\label{SSS:FS}
Next, we define the category
$\FS{(\Ccal,\partial,\epsilon,\star)} = (\Scal,\delta,\circ,s)$ in
$\SinCat{\omega}$ for
each cubical $(\Ccal,\partial,\epsilon,\star)$ in $\CubCat{\omega}$ 
\begin{enumerate}[{\bf (i)}]
\item The underlying set is the following colimit in the category $\catego{Set}$:
\begin{equation*}
\Scal = \colim(\Ccal_0 \oto{\epsilon_{1,1}} \Ccal_1 \oto{\epsilon_{2,2}} \Ccal_2 \oto{\epsilon_{3,3}} \dots) = \bigsqcup_{n \in \Nbb} \Ccal_n / \sim,
\end{equation*}
where $a \in \Ccal_m$ and $b \in \Ccal_n$ with $m \leq n$ are
equivalent if and only if
$b = \epsilon_{n,n} \dots \epsilon_{m+1,m+1} a$ by injectivity of the
degeneracy maps. We write $\phi_n: \Ccal_n \to \Scal$ for the maps
forming a cocone to the colimit. They sends each $n$-cell $a$ to its
equivalence class in the above quotient.
\item The $\delta_i^\alpha: \Scal \to \Scal$, for $i \in \Nbb_+$, are
  the unique morphisms in $\catego{Set}$ induced by the cocone
  $(\phi_n \epsilon_{n,i} \partial_{n,i}^\alpha)_{n \geq i}$. They
  send the equivalence class of $a \in \Ccal_n$, $n \geq i$, to the
  set
  $\delta_i^\alpha [a]_\sim = [\epsilon_{n,i} \partial_{n,i}^\alpha
  a]_\sim$.
\item The $s_i: \Scal \to \Scal$, for $i \in\Nbb_+$, are the unique
  morphisms in $\catego{Set}$ induced by the cocone
  $(\phi_n \epsilon_{n,i+1} \partial_{n,i}^-)_{n \geq i+1}$. They send
  the equivalence class of $a \in \Ccal_n$ to
  $s_i [a]_\sim = [\epsilon_{n,i+1} \partial_{n,i}^- a]_\sim$.
\item The $\tilde{s}_i: \Scal \to \Scal$, for $i \in \Nbb_+$, are the
  unique morphisms in $\catego{Set}$ induced by the cocone
  $(\phi_n \epsilon_{n,i} \partial_{n,i+1}^-)_{n \geq i+1}$. They send
  the equivalence class of $a \in \Ccal_n$ to
  $\tilde{s}_i [a]_\sim = [\epsilon_{n,i} \partial_{n,i+1}^- a]_\sim$.
\item The $\circ_i: \Scal \times_{\Delta_i} \Scal \to \Scal$, for
  $i \in \Nbb_+$, send the equivalence classes of $a,b \in \Ccal_n$,
  to $[a]_\sim \circ_i [b]_\sim = [a \star_{n,i} b]_\sim$.
\end{enumerate}

For each functor $g : \Ccal \to \Ccal'$ in
$\CubCat{\omega}$ we define the morphism
$\FS{g}: \FS{\Ccal} \to \FS{\Ccal'}$ in $\SinCat{\omega}$ as the
unique morphism in $\catego{Set}$ induced by the cocone
$(\phi'_n \circ g_n : \Ccal_n \to \FS{\Ccal'})_{n \in \Nbb}$, where
the $\phi'_n$ are the inclusion maps $\Ccal'_n \to \FS{\Ccal'}$. It
sends the equivalence class of $a \in \Ccal_n$ to
$\FS{g}([a]_\sim) = [g(a)]_\sim$ in $\Ccal'_n$.

In the following, we do not distinguish between equivalence classes
and their representatives.

\begin{lemma}
\label{L:FS-well-defined}
The functor $\FS{(-)}$ is well-defined.
\end{lemma}
\begin{proof}
  The proof is similar to that of
  Lemma~\ref{L:FC-well-defined}. First we check that
  $(\Scal,\delta,\circ,s)$ is a category in $\SinCat{\omega}$,  
  verifying the axioms in \refeq{SSS:DefCubSingleSetWithoutConn}. Once
  again we only present some cases and refer to
  Appendix~\ref{A:FS-well-defined} for the remaining ones. 
  Suppose $i,j \in \Nbb_+$ and $w,x,y,z \in \Scal$ with
  representatives in~$\Ccal_n$ for $n \geq i,j$.
\begin{enumerate}[{\bf (i)}]
\item
  $\Delta_i(x, y) \Leftrightarrow \partial_{n,i}^+ x =
  \partial_{n,i}^- y$, hence $\Delta_i(x, y \circ_i z)$ and
  $\Delta_i(y, z)$ if and only $\Delta_i(x \circ_i y, z)$ and
  $\Delta_i(x, y)$. It follows that
  $x \circ_i (y \circ_i z) = x \star_{n,i} (y \star_{n,i} z) = (x
  \star_{n,i} y) \star_{n,i} z = (x \circ_i y) \circ_i z$.
\item $x \circ_i \delta_i^+ x = x \star_{n,i} \epsilon_{n,i} \partial_{n,i}^+ x = x$ and $\delta_i^- x \circ_i x = \epsilon_{n,i} \partial_{n,i}^- x \star_{n,i} x = x$.
\item
	\begin{itemize}
	\item If $i<j$ then
	\[
	\begin{array}{rcl}
	\delta_i^\alpha \delta_j^\beta x & = & \epsilon_{n,i} \partial_{n,i}^\alpha \epsilon_{n,j} \partial_{n,j}^\beta x \\
	 & = & \epsilon_{n,i} \epsilon_{n-1,j-1} \partial_{n-1,i}^\alpha \partial_{n,j}^\beta x \\
	 & = & \epsilon_{n,j} \epsilon_{n-1,i} \partial_{n-1,j-1}^\beta \partial_{n,i}^\alpha x \\
	 & = & \epsilon_{n,j} \partial_{n,j}^\beta \epsilon_{n,i} \partial_{n,i}^\alpha x \\
	 & = & \delta_j^\beta \delta_i^\alpha x,
	\end{array}
	\]
	\item if $i>j$ then
	\[
	\begin{array}{rcl}
	\delta_i^\alpha \delta_j^\beta x & = & \epsilon_{n,i} \partial_{n,i}^\alpha \epsilon_{n,j} \partial_{n,j}^\beta x \\
	 & = & \epsilon_{n,i} \epsilon_{n-1,j} \partial_{n-1,i-1}^\alpha \partial_{n,j}^\beta x \\
	 & = & \epsilon_{n,j} \epsilon_{n-1,i-1} \partial_{n-1,j}^\beta \partial_{n,i}^\alpha x \\
	 & = & \epsilon_{n,j} \partial_{n,j}^\beta \epsilon_{n,i} \partial_{n,i}^\alpha x \\
	 & = & \delta_j^\beta \delta_i^\alpha x.
	\end{array}
	\]
	\end{itemize}
\end{enumerate}

Next we show that $\FS{g}$ is a morphism in $\SinCat{\omega}$. For all
$1 \leq i \leq n$ and $x,y \in \Scal$ with representatives $a,b$ in
$\Ccal_n$,
\begin{enumerate}[{\bf (i)}]
\item $\FS{g} \delta_i^\alpha x = [g \epsilon_{n,i} \partial_{n,i}^\alpha x]_\sim = [\epsilon'_{n,i} {\partial'}_{n,i}^\alpha g x]_\sim = {\delta'}_i^\alpha \FS{g} x$,
\item if $\Delta_i(x,y)$, then $\partial_{n,i}^+ a = \partial_{n,i}^- b$, so $\FS{g} (x \circ_i y) = [g (a \star_{n,i} b)]_\sim = [g(a) \star_{n,i} g(b)]_\sim = \FS{g} (x) \circ_i \FS{g} (y)$,
\item if $n \geq i+1$, then $\FS{g} s_i x = [g \epsilon_{n,i+1} \partial_{n,i}^- x]_\sim = [\epsilon'_{n,i+1} {\partial'}_{n,i}^- g x]_\sim = s'_i \FS{g} x$.
\end{enumerate}
Finally, $\FS{g}$ is indeed a morphism: it is clear from its
definition that it preserves sources and targets, while reservation of
compositions of morphisms in $\SinCat{\omega}$ follows from properties
of colimits.
\end{proof}

This time, the calculations in the proof above are performed in
$\CubCat{\omega}$. Their formalisation with Isabelle is
  beyond the scope of this work.

\subsubsection{Natural isomorphism $\FS{(-)} \circ \FC{(-)} \Rightarrow \id$}
\label{SSS:NaturalIsomorphism1}
Let $(\Scal,\delta,\circ,s)$ be a category of $\SinCat{\omega}$ and
$(\Ccal,\partial,\epsilon,\star) = \FC{(\Scal,\delta,\circ,s)}$ as
previously.  The category
$(\Scal',\delta',\circ',s') := \FS{(\FC{(\Scal,\delta,\circ,s)})}$ in
$\SinCat{\omega}$ is computed as follows. Let $i \in \Nbb_+$.
\begin{enumerate}[{\bf (i)}]
\item
  $\Scal' = \colim (\Ccal_0 \oto{\epsilon_{1,1}} \Ccal_1
  \oto{\epsilon_{2,2}} \dots) = \colim (\Scal^{>0} \oto{\id}
  \Scal^{>1} \oto{\id}
  \dots)$. Axiom~\ref{SSS:DefCubSingleSetWithoutConn}\eqref{I:AxiomFiniteDimCells}
  then implies that
\[
\Scal' \; = \;  \bigsqcup_{n \in \Nbb} \Scal^{>n} / \sim \; = \; \bigcup_{n \in \Nbb} \Scal^{>n} = \Scal.
\]

\item The face maps are ${\delta'}_i^\alpha = \delta_i^\alpha$,
  because they send each $x \in \Scal^{>n}$ with $n \geq i$ to
  ${\delta'}_i^\alpha x = \epsilon_{n,i} \partial_{n,i}^\alpha x =
  \delta_i^\alpha x$. Thus in particular ${\Scal'}^i = \Scal^i$.
\item The symmetries and reverse symmetries are
  $s'_i = s_i \delta_i^-$ and
  $\tilde{s}'_i = \tilde{s}_i \delta_{i+1}^-$, because they send
  $x \in \Scal^{>n}$ with $n \geq i+1$ to
  $s'_i x = \epsilon_{n,i+1} \partial_{n,i}^- x = s_i \delta_i^- x$
  and
  $\tilde{s}'_i x = \epsilon_{n,i} \partial_{n,i+1}^- x = \tilde{s}_i
  \delta_{i+1}^-$. Therefore, $s'_i = s_i$ on $\Scal^i$ and
  $\tilde{s}'_i = \tilde{s}_i$ on $\Scal^{i+1}$.
\item The compositions are $\circ'_i = \circ_i$, because for all
  $x,y \in \Scal^{>n}$ with $n \geq i$ such that $\Delta'_i(x,y)$, we
  have $\Delta_i(x,y)$ and the compositions $\circ'_i$ send $x,y$ to
  $x \circ'_i y = x \star_{n,i} y = x \circ_i y$.
\end{enumerate}
Further, any morphism $f: \Scal \to \mathcal{T}$ in $\SinCat{\omega}$
satisfies $\FS{(\FC{f})} = f$.

\begin{lemma}
\label{L:NatIsoFSFCId}
The maps $\id_{\Scal}: \Scal' \to \Scal$ induce a natural isomorphism
$\mu: \FS{(-)} \circ \FC{(-)} \Rightarrow \id$.
\end{lemma}
\begin{proof}
  First, $\id {\delta'}_i^\alpha = \delta_i^\alpha \id$ and
  $\Delta_i(x,y)$ implies
  $\id(x \circ'_i y) = \id(x) \circ_i \id(y)$. Second,
  $\id(s'_i x) = s_i \id(x)$ for all $x \in {\Scal'}^i$. So, as
  $\Scal'=\Scal$, $\mu_{\Scal} = \id : \Scal' \to \Scal$ is a morphism
  in $\SinCat{\omega}$.

  Naturality of the family $\mu$ and the invertibility of each
  component $\mu_{\Scal}$ are clear.
\end{proof}

\subsubsection{Natural isomorphism $\id \Rightarrow \FC{(-)} \circ \FS{(-)}$}
\label{SSS:NaturalIsomorphism2}
Let $(\Ccal,\partial,\epsilon,\star)$ be a category in
$\CubCat{\omega}$ and let
$(\Scal,\delta,\circ,s) = \FS{(\Ccal,\partial,\epsilon,\star)}$ as
before. The category
$(\Ccal',\partial',\epsilon',\star') =
\FC{(\FS{(\Ccal,\partial,\epsilon,\star)})}$ in $\CubCat{\omega}$ is
computed as follows. Let $1 \leq i \leq n$.
\begin{enumerate}[{\bf (i)}]
\item The sets of $n$-cells are
\begin{align*}
\Ccal'_n &= \Scal^{>n} = \left( \bigsqcup_{m \in \Nbb} \Ccal_m / \sim \right)^{>n}.
\end{align*}
\item The face maps ${\partial'}_{n,i}^\alpha$ send each
  $a \in \Ccal'_n$, represented by some $a_0 \in \Ccal_m$ with
  $m \geq n$, to
\begin{align*}
{\partial'}_{n,i}^\alpha a
= s_{n-1} \dots s_i \delta_i^\alpha a
= [\epsilon_{m,n} \partial_{m,n-1}^- \dots \epsilon_{m,i+1} \partial_{m,i}^- \epsilon_{m,i} \partial_{m,i}^\alpha a_0]_\sim
= [\epsilon_{m,n} \partial_{m,i}^\alpha a_0]_\sim.
\end{align*}
\item The degeneracies $\epsilon'_{n,i}$ send each $a \in \Ccal'_n$,
  represented by some $a_0 \in \Ccal_m$ with $m \geq n$, to
\begin{align*}
\epsilon'_{n,i} a
= \tilde{s}_i \dots \tilde{s}_{n-1} a
= [\epsilon_{m,i} \partial_{m,i+1}^- \dots \epsilon_{m,n-1} \partial_{m,n}^- a_0]_\sim
= [\epsilon_{m,i} \partial_{m,n}^- a_0]_\sim.
\end{align*}
\item The compositions $a \star'_{n,i} b$, where $a,b \in \Ccal'_n$
  are represented by some $a_0,b_0 \in \Ccal_m$ with $m \geq n$,
  satisfy $a \star'_{n,i} b = [a_0 \star_{m,i} b_0]_\sim$ whenever
  $\Delta_i(a,b)$, that is, when $a_0 \star_{n,i} b_0$ is defined,
  and are undefined otherwise. 
\end{enumerate}
Further, for any morphism $g: \Ccal \to \Dcal$ in $\CubCat{\omega}$,
$\FC{(\FS{g})}$ sends $[a]_\sim$ in $\FC{(\FS{\Ccal})}_n$ to
$[g(a)]_\sim$ in $\FC{(\FS{\Dcal})}_n$.

Suppose $\phi_n: \Ccal_n \to \Scal$ is the morphism to the colimit,
which sends any $n$-cell $a$ to its equivalence class in
$\bigsqcup_{m \in \Nbb} \Ccal_m / \sim$. Its image is included in
$\Ccal'_n$ because
$a \sim \epsilon_{i,i} \dots \epsilon_{n+1,n+1} a = a'$ for each
$a \in \Ccal_n$ and $i \geq n+1$. Thus
$\delta_i^- \phi_n a = [\epsilon_{i,i} \partial_{i,i}^- a']_\sim =
\phi_n a$. Let $(\eta_{\Ccal})_n: \Ccal_n \to \Ccal'_n$ be the induced
map and $\eta_{\Ccal}: \Ccal \to \Ccal'$ the family
$\left( (\eta_{\Ccal})_n \right)_{n \in \Nbb}$. Let further
$(\overline{\eta}_{\Ccal})_n: \Ccal'_n \to \Ccal_n$ be the map that
sends $b$, represented by some $a \in \Ccal_m$ with $m \geq n$, to
$\partial_{n+1,n+1}^- \dots \partial_{m,m}^- a$. It is well-defined
because the image of $b$ does not depend on the choice of the
representative: indeed $a \sim a'$ and $a' \in \Ccal_l$ with
$l \geq m$ imply
\begin{align*}
\partial_{n+1,n+1}^- \dots \partial_{l,l}^- a'
= \partial_{n+1,n+1}^- \dots \partial_{l,l}^- \epsilon_{l,l} \dots \epsilon_{m+1,m+1} a
= \partial_{n+1,n+1}^- \dots \partial_{m,m}^- a.
\end{align*}
Finally, we write $\overline{\eta}_{\Ccal}: \Ccal' \to \Ccal$ for the
family $\left( (\overline{\eta}_{\Ccal})_n \right)_{n \in \Nbb}$.

\begin{lemma}
\label{L:NatIsoIdFCFS}
The maps $\eta_{\Ccal}: \Ccal \to \Ccal'$ induce a natural isomorphism $\eta: \id \Rightarrow \FC{(-)} \circ \FS{(-)}$.
\end{lemma}
\begin{proof}
  We need to show that the maps
  $\eta_{\Ccal}: \Ccal \leftrightarrows \Ccal'
  :\overline{\eta}_{\Ccal}$ are morphisms in $\CubCat{\omega}$, which
  are natural and inverses of each other. Let $a,b \in \Ccal_n$ and
  $c \in \Ccal_{n-1}$. Then
\begin{align*}
(\eta_{\Ccal})_{n-1} \partial_{n,i}^\alpha a
&= [\partial_{n,i}^\alpha a]_\sim
= [\epsilon_{n,n} \partial_{n,i}^\alpha a]_\sim
= {\partial'}_{n,i}^\alpha [a]_\sim
= {\partial'}_{n,i}^\alpha (\eta_{\Ccal})_n a, \\
(\eta_{\Ccal})_n (a \star_{n,i} b)
&= [a \star_{n,i} b]_\sim
= [a]_\sim \star'_{n,i} [b]_\sim
= (\eta_{\Ccal})_n a \star'_{n,i} (\eta_{\Ccal})_n b, \\
(\eta_{\Ccal})_n \epsilon_{n,i} c
&= [\epsilon_{n,i} c]_\sim
= [\epsilon_{n,i} \partial_{n,n}^- \epsilon_{n,n} c]_\sim
= \epsilon'_{n,i} [\epsilon_{n,n} c]_\sim
= \epsilon'_{n,i} [c]_\sim
= \epsilon'_{n,i} (\eta_{\Ccal})_{n-1} c.
\end{align*}
The $\eta_{\Ccal}$ are therefore morphisms in $\CubCat{\omega}$. The
proof that the $\overline{\eta}_{\Ccal}$ are morphisms in
$\CubCat{\omega}$ is similar.

The maps $\eta_{\Ccal}$ and $\overline{\eta}_{\Ccal}$ are
inverses. Indeed,
$\eta_{\Ccal} (\overline{\eta}_{\Ccal} (b)) = [\partial_{n+1,n+1}^-
\dots \partial_{m,m}^- a]_\sim = b$ holds for every~$b \in \Ccal'_n$
represented by some $a \in \Ccal_m$ with $m \geq n$, and likewise for
the other composition.

Finally, the family $\eta$ is natural because
$\eta_{\Ccal} g (a) = [g (a)]_\sim = \FC{(\FS{g})} [a]_\sim = g
(\overline{\eta}_{\Ccal} a)$ for every functor $g: \Ccal \to \Dcal$.
\end{proof}

\begin{remark}
  All axioms of single-set $\omega$-categories have been used in one
  direction of the proof of Theorem~\ref{T:EquivSinSetClass}, and they
  have been derived in the other one. Single-set 
  $\omega$-categories and their classical counterparts are therefore
  essentially the same.
  
  Interestingly, in this proof, Axiom~\ref{SSS:DefCubSingleSetWithoutConn}\eqref{I:AxiomFiniteDimCells}
  has only been used to establish the natural ismorphism $\mu$, more
  precisely for showing that $\Scal=\Scal'$ in the colimit
  construction. It is unnecessary for a proof of equivalence between
  $\SinCat{n}$ and $\CubCat{n}$, where the colimit construction
  simplifies; see also \ref{SS:SingleSetCubicalCategories}. The proofs
  of all other properties work uniformly for $n$ and $\omega$.
\end{remark}

\subsection{Equivalence for connections}
\label{SS:EquivalenceWithCubicalCategoriesWithConnections}

In this subsection we extend the equivalence
$\SinCat{\omega}\simeq \CubCat{\omega}$ to connections.
\begin{theorem}
\label{T:EquivSinSetClassConn}
There is an equivalence of categories
\[
\begin{tikzcd}[global scale = 10 and 4 and 1 and 1]
\FCG{(-)}  :  \SinCatG{\omega} \ar[r, shift left] & \CubCatG{\omega}  :  \FSG{(-)} \ar[l, shift left].
\end{tikzcd}
\]
\end{theorem}
In the proof, we go through the same steps as before.

\subsubsection{The functor $\FCG{(-)}$}
\label{SSS:FCG}
For every category $(\Scal,\delta,\circ,\gamma)$ in $\SinCatG{\omega}$,
the category $(\Ccal,\partial,\epsilon,\star,\Gamma) :=
\FCG{(\Scal,\delta,\circ,\gamma)}$ in $\CubCatG{\omega}$ is defined as follows:
\begin{enumerate}[{\bf (i)}]
\item the underlying $\omega$-category in $\CubCat{\omega}$ is $(\Ccal,\partial,\epsilon,\star) = \FC{(\Scal,\delta,\circ,s)}$,
\item the connections are the restrictions
  $\Gamma_{n,i}^\alpha: \Ccal_{n-1} \to \Ccal_n$ of
  $\Gamma_{n,i}^\alpha = \gamma_i^\alpha \tilde{s}_i \dots
  \tilde{s}_{n-1}$ for $1 \leq i < n$.
\end{enumerate}
For each morphism $f: \Scal \to \Scal'$ in $\SinCatG{\omega}$ we
define the morphisms $\FCG{f}: \FCG{\Scal} \to \FCG{\Scal'}$ as
$\FC{f}$ on $n$-cells in $\CubCatG{\omega}$ for each $n \in \Nbb$.

\begin{lemma}
\label{L:FCG-well-defined}
The functor $\FCG{(-)}$ is well-defined.
\end{lemma}
\begin{proof}
  As before, we first check that
  $\FCG{(\Scal, \partial, \epsilon, \star, \Gamma)}$ defines a
 category in $\CubCatG{\omega}$. We only need to consider the
  connection axioms. We show selected axioms only and refer to
  Appendix~\ref{A:FCG-well-defined} for the others. Let
  $1 \leq i, j < n$ and $a, b, c, d \in \FCG{\Scal}_n$.
\begin{enumerate}[{\bf (i)}]
\item
	\begin{itemize}
	\item If $i < j$, then
          Axiom~\ref{SSS:DefCubSingleSetWithConn}\eqref{I:AxiomConnFaces}
          and others imply that
	\begin{align*}
	\partial_{n,i}^\alpha \Gamma_{n,j}^\beta
	&= s_{n-1} \dots s_j s_{j-1} \gamma_j^\beta s_{j-2} \dots s_i \tilde{s}_j \dots \tilde{s}_{n-1} \delta_i^\alpha \\
	&= s_{n-1} \dots s_{j+1} \gamma_{j-1}^\beta s_j s_{j-2} \dots s_i \tilde{s}_j \dots \tilde{s}_{n-1} \delta_i^\alpha \\
	&= \gamma_{j-1}^\beta s_{j-2} \dots s_i \delta_i^\alpha \\
	&= \Gamma_{n-1,j-1}^\beta \partial_{n-1,i}^\alpha,
	\end{align*}
	\item $\partial_{n,i}^\alpha \Gamma_{n,i}^\alpha = s_{n-1} \dots s_i \delta_i^\alpha \gamma_i^\alpha \tilde{s}_i \dots \tilde{s}_{n-1} = \id$.
	\end{itemize}
      \item If $a, b$ are $\star_{n,j}$-composable, then
        Axiom~\ref{SSS:DefCubSingleSetWithConn}\eqref{I:AxiomConnCompoCorner}
        and others imply that,
	\begin{itemize}
	\item if $i < j$ then
	\begin{align*}
	\Gamma_{n+1,i}^\alpha (a \star_{n,j} b)
	&= \gamma_i^\alpha \tilde{s}_i \dots \tilde{s}_j (\tilde{s}_{j+1} \dots \tilde{s}_n a \circ_j \tilde{s}_{j+1} \dots \tilde{s}_n b) \\
	&= \gamma_i^\alpha (\tilde{s}_i \dots \tilde{s}_n a \circ_{j+1} \tilde{s}_i \dots \tilde{s}_n b) \\
	&= \Gamma_{n+1,i}^\alpha a \star_{n+1,j+1} \Gamma_{n+1,i}^\alpha b,
	\end{align*}
	\item if $i = j$ then
	\begin{align*}
	\Gamma_{n+1,i}^- (a \star_{n,i} b)
	&= \gamma_i^- (\tilde{s}_i \dots \tilde{s}_n a \circ_{i+1} \tilde{s}_i \dots \tilde{s}_n b) \\
	&= (\gamma_i^- \tilde{s}_i \dots \tilde{s}_n a \circ_i s_i \tilde{s}_i \dots \tilde{s}_n b) \circ_{i+1} (\tilde{s}_i \dots \tilde{s}_n b \circ_i \gamma_i^- \tilde{s}_i \dots \tilde{s}_n b) \\
	&= (\Gamma_{n+1,i}^- a \star_{n+1,i} \epsilon_{n+1,i+1} b) \star_{n+1,i+1} (\epsilon_{n+1,i} b \star_{n+1,i} \Gamma_{n+1,i}^- b).
	\end{align*}
	\end{itemize}
\setcounter{enumi}{3}
\item If $i < j$ then $\Gamma_{n+1,i}^\alpha \epsilon_{n,j}
	= \gamma_i^\alpha \tilde{s}_i \dots \tilde{s}_{j-1} \tilde{s}_{j+1} \dots \tilde{s}_n \tilde{s}_j \dots \tilde{s}_{n-1} 
	= \tilde{s}_{j+1} \dots \tilde{s}_n \gamma_i^\alpha \tilde{s}_i \dots \tilde{s}_{n-1} 
= \epsilon_{n+1,j+1} \Gamma_{n,i}^\alpha$.

\end{enumerate}

It remains to be shown that $\FCG{f}$ is a morphism in
$\CubCatG{\omega}$. Indeed, for each $1 \leq i \leq n$ and
$a,b \in \FCG{\Scal}_n$,
\begin{align*}
\FCG{f} \Gamma_{n,i}^\alpha a
\ = \ f \gamma_i^\alpha \tilde{s}_i \dots \tilde{s}_{n-1} a
\ = \ \gamma_i^\alpha \tilde{s}_i \dots \tilde{s}_{n-1} f a
\ = \ \Gamma_{n,i}^\alpha \FCG{f} a.
\end{align*}
The claim then follows from Lemma~\ref{L:FC-well-defined}.
\end{proof}

\subsubsection{The functor $\FSG{(-)}$}
\label{SSS:FSG}
For $(\Ccal,\partial,\epsilon,\star,\Gamma)$ in
$\CubCatG{\omega}$, the category
$\FSG{(\Ccal,\partial,\epsilon,\star,\Gamma)} :=
(\Scal,\delta,\odot,s,\gamma)$ in $\SinCatG{\omega}$ is defined as
follows:
\begin{enumerate}[{\bf (i)}]
\item The underlying single-set cubical $\omega$-category in $\SinCat{\omega}$ is $\FS{(\Ccal,\partial,\epsilon,\star)} = (\Scal,\delta,\odot,s)$.
\item The $\gamma_i^\alpha: \Scal \to \Scal$, for
  $i \geq 1$, are the unique morphisms in $\catego{Set}$ induced by
  the cocone
  $(\phi_n \Gamma_{n,i}^\alpha \partial_{n,i}^\alpha)_{n > i}$. They
  send the equivalence class of each $a \in \Ccal_n$ to
  $\gamma_i^\alpha [a]_\sim = [\Gamma_{n,i}^\alpha
  \partial_{n,i}^\alpha a]_\sim$.
\end{enumerate}
For each morphism $g: \Ccal \to \Ccal'$ in $\CubCat{\omega}$, the
morphism $\FSG{g}: \FSG{\Ccal} \to \FSG{\Ccal'}$ in $\SinCatG{\omega}$
is defined as $\FS{g}$.

\begin{lemma}
\label{L:FSG-well-defined}
The functor $\FSG{(-)}$ is well-defined.
\end{lemma}
\begin{proof}
  As usual, we start with checking that
  $(\Scal,\delta,\odot,s,\gamma)$ is a category in $\SinCatG{\omega}$
  It remains to consider the single-set axioms for connections. As
  usual, we only show some cases and refer to
  Appendix~\ref{A:FSG-well-defined} for the remaining ones. Let
  $i,j \in \Nbb_+$ and $x,y \in \Scal$ with representatives $a,b$ in
  $\Ccal_n$ with $n > i,j$.
\begin{enumerate}[{\bf (i)}]
\item If $i \neq j,j+1$ and $x \in \Scal^j$, then $\partial_{n,j}^- x= \partial_{n,j}^+ x$ so
	$\delta_j^\alpha \gamma_j^\alpha x
	= \epsilon_{n,j} \partial_{n,j}^\alpha \Gamma_{n,j}^\alpha \partial_{n,j}^\alpha x
	= \delta_j^\alpha x
	= x$.
\item If $j \neq i,i+1$, $x,y \in \Scal^i$ and $\Delta_{i+1}(x,y)$ then
	\begin{align*}
	\gamma_i^+ (x \circ_{i+1} y)
	&= \Gamma_{n,i}^+ (\partial_{n,i}^+ x \star_{n,i} \partial_{n,i}^+ y) \\
	&= (\Gamma_{n,i}^+ \partial_{n,i}^+ x \star_{n+1,i+1} \epsilon_{n,i+1} \partial_{n,i}^+ x) \star_{n+1,i} (\epsilon_{n,i} \partial_{n,i}^+ x \star_{n+1,i+1} \Gamma_{n,i}^+ \partial_{n,i}^+ y) \\
	&= (\gamma_i^+ x \circ_{i+1} s_i x) \circ_i (x \circ_{i+1} \gamma_i^+ y).
	\end{align*}
\setcounter{enumi}{4}
\item If $x \in \Scal^{i,j}$ and $i < j-1$ then
	\[
	\gamma_i^\alpha \gamma_j^\beta x
	= \Gamma_{n,i}^\alpha \Gamma_{n-1,j-1}^\beta \partial_{n-1,i}^\alpha \partial_{n,j}^\beta x
	= \Gamma_{n,j}^\beta \Gamma_{n-1,i}^\alpha \partial_{n-1,j-1}^\beta \partial_{n,i}^\alpha x
	= \gamma_j^\beta \gamma_i^\alpha x.
	\]
\item If $x \in \Scal^{i,i+1}$ then $\partial_{n,i}^- a = \partial_{n,i}^+ a$ and $\partial_{n,i+1}^- a = \partial_{n,i+1}^+ a$ so
\begin{align*}
s_{i+1} s_i \gamma_{i+1}^\alpha x
= \epsilon_{n,i+2} \Gamma_{n-1,i}^\alpha \partial_{n-1,i}^- \partial_{n,i+1}^\alpha x
= \Gamma_{n,i}^\alpha \epsilon_{n-1,i+1} \partial_{n-1,i}^\alpha \partial_{n,i+1}^- x
= \gamma_i^\alpha s_{i+1} x.
\end{align*}
\end{enumerate}
It remains to show that $\FSG{g}$ is a morphism in
$\SinCatG{\omega}$. For $i \geq 1$ and $x \in \Scal^i$,
\begin{align*}
\FSG{g} \gamma_i^\alpha x
= [g \Gamma_{n,i}^\alpha \partial_{n,i}^\alpha x]_\sim
= [\Gamma_{n,i}^\alpha \partial_{n,i}^\alpha g x]_\sim
= \gamma_i^\alpha \FSG{g} x.
\end{align*}
The claim then follows from Lemma~\ref{L:FS-well-defined}. 
\end{proof}

\subsubsection{Natural isomorphism $\FSG{(-)} \circ \FCG{(-)} \Rightarrow \id$}
Let $(\Scal,\delta,\circ,s,\gamma)$ be a category in
$\SinCatG{\omega}$ and let
$(\Ccal,\partial,\epsilon,\star,\Gamma) =
\FCG{(\Scal,\delta,\circ,s,\gamma)}$. The category
$(\Scal',\delta',\circ',s',\gamma') :=
\FSG{(\FCG{(\Scal,\delta,\circ,s,\gamma)})}$ in $\SinCatG{\omega}$ is
determined as follows. For $i \in \Nbb_+$,
\begin{enumerate}[{\bf (i)}]
\item the underlying $\omega$-category in $\SinCat{\omega}$ is
  $(\Scal',\delta',\circ',s') = \FS{(\FC{(\Scal,\delta,\circ,s)})}$
\item the connections are
  ${\gamma'}_i^\alpha = \gamma_i^\alpha \delta_i^\alpha$; they send
  $x \in \Scal^{>n}$ with $n>i$ to
  ${\gamma'}_i^\alpha x = \Gamma_{n,i}^\alpha \partial_{n,i}^\alpha x
  = \gamma_i^\alpha \delta_i^\alpha x$. Hence in particular
  ${\gamma'}_i^\alpha = \gamma_i^\alpha$ on $\Scal^i$.
\end{enumerate}
Also, for each morphism $f: \Scal \to \mathcal{T}$ in
$\SinCatG{\omega}$, $\FSG{(\FCG{f})}= \FS{(\FC{f})}$, so that
$\FSG{(\FCG{f})} = f$.

\begin{lemma}
\label{L:NatIsoFSGFCGId}
The maps $\mu_{\Scal}: \Scal' \to \Scal$ induce a natural isomorphism $\mu: \FSG{(-)} \circ \FCG{(-)} \Rightarrow \id$.
\end{lemma}
\begin{proof}
  Relative to Lemma~\ref{L:NatIsoFSFCId} it remains to show that
  $\mu_{\Scal}$ is a morphism in $\SinCatG{\omega}$. Indeed,
  $\id {\gamma'}_i^\alpha x = \gamma_i^\alpha \id x$ for all
  $x \in {\Scal'}^i$.
\end{proof}

\subsubsection{Natural isomorphism $\id \Rightarrow \FCG{(-)} \circ \FSG{(-)}$}
Let $(\Ccal,\partial,\epsilon,\star,\Gamma)$ be a category in
$\CubCatG{\omega}$ and let
$(\Scal,\delta,\circ,s,\gamma) :=
\FSG{(\Ccal,\partial,\epsilon,\star,\Gamma)}$ as previously.  In this
case, the category
$(\Ccal',\partial',\epsilon',\star',\Gamma') :=
\FCG{(\FSG{(\Ccal,\partial,\epsilon,\star,\Gamma)})}$ in
$\CubCatG{\omega}$ is determined as follows. For $1 \leq i < n$,
\begin{enumerate}[{\bf (i)}]
\item the underlying $\omega$-category in $\CubCat{\omega}$ is
  $(\Ccal',\partial',\epsilon',\star') =
  \FC{(\FS{(\Ccal,\partial,\epsilon,\star)})}$.
\item the connections are ${\Gamma'}_{n,i}^\alpha$; they send
  $a \in \Ccal'_n$ represented by some $a_0 \in \Ccal_m$ with
  $m \geq n$ to
\[
{\Gamma'}_{n,i}^\alpha a
= \gamma_i^\alpha \tilde{s}_i \dots \tilde{s}_{n-1} a
= [\Gamma_{m,i}^\alpha \partial_{m,i}^\alpha \epsilon_{m,i} \partial_{m,i+1}^- \dots \epsilon_{m,n-1} \partial_{m,n}^- a_0]
= [\Gamma_{m,i}^\alpha \partial_{m,n}^- a_0].
\]
\end{enumerate}
Moreover, for each morphism $g: \Ccal \to \Dcal$ in $\CubCat{\omega}$,
$\FCG{(\FSG{g})}= \FC{(\FS{g})}$ sends $[a]_\sim$ in
$\FCG{(\FSG{\Ccal})}_n$ to $[g(a)]$ in $\FCG{(\FSG{\Dcal})}_n$.

\begin{lemma}
\label{L:NatIsoIdFCGFSG}
The maps $\eta_{\Ccal}: \Ccal \to \Ccal'$ induce a natural isomorphism $\eta: \id \Rightarrow \FCG{(-)} \circ \FSG{(-)}$.
\end{lemma}
\begin{proof}
  We must show, relative to Lemma~\ref{L:NatIsoIdFCFS}, that
  $\eta_{\Ccal}$ and $\overline{\eta}_{\Ccal}$ are morphisms in
  $\CubCatG{\omega}$. Indeed, for all $1 \leq i < n$,
  $a \in \Ccal_{n-1}$ and $b \in {\Ccal'}_{n-1}$ represented by some
  $b_0 \in \Ccal_{m-1}$,
\begin{align*}
{\Gamma'}_{n,i}^\alpha \eta_{\Ccal} a
&= {\Gamma'}_{n,i}^\alpha [\epsilon_{n,n} a]
= [\Gamma_{n,i}^\alpha a]
= \eta_{\Ccal} \Gamma_{n,i}^\alpha a,\\
\overline{\eta}_{\Ccal} {\Gamma'}_{n,i}^\alpha b
&= \partial_{n+1,n+1}^- \dots \partial_{m,m}^- \Gamma_{m,i}^\alpha \partial_{m,n}^- b_0
= \Gamma_{n,i}^\alpha \partial_{n,n}^- \dots \partial_{m,m}^- b_0
= \Gamma_{n,i}^\alpha \overline{\eta}_{\Ccal} b.\qedhere
\end{align*}
\end{proof}

\subsection{Equivalence for inverses}
\label{SS:EquivalenceWithCubicalCategoriesWithInverses}

Finally, we extend the equivalence $\SinCatG{\omega} \simeq
\CubCatG{\omega}$ to the case $(\omega,p)$. 
\begin{theorem}
\label{T:EquivSinSetClassConnInv}
The categories $\SinCatG{(\omega,p)}$ and $\CubCatG{(\omega,p)}$ are equivalent.
\end{theorem}
\begin{proof}
  More specifically, we show that the functors
  $\FCG{(-)}: \SinCatG{\omega} \simeq \CubCatG{\omega} :\FSG{(-)}$
  from Theorem \ref{T:EquivSinSetClassConn} extend to
  $\SinCatG{(\omega,p)} \simeq \CubCatG{(\omega,p)}$.  First, suppose
  $\Scal$ is a category in $\SinCatG{(\omega,p)}$. Let
  $\Ccal = \FCG{\Scal}$ and, for all $n > p$ and $1 \leq i \leq n$,
  pick an $n$-cell $c$ in $\Ccal_n$ with an $R_{n-1,i}$-invertible
  shell.
\begin{itemize}
\item Then $\partial_{n,j}^\alpha c$ has an $R_{n-1,i-1}$-inverse, $d$
  say, for each $1 \leq j < i$, by the hypothesis. Therefore
\begin{align*}
\partial_{n,j}^\alpha c \star_{n-1,i-1} d &= \epsilon_{n-1,i-1} \partial_{n-1,i-1}^- \partial_{n,j}^\alpha c, \\
s_{n-1} \dots s_j \delta_j^\alpha c \circ_{i-1} d &= \delta_{i-1}^- s_{n-1} \dots s_j \delta_j^\alpha c, \\
\delta_j^\alpha c \circ_i e &= \delta_i^- \delta_j^\alpha c,
\end{align*}
where we write $e = \tilde{s}_j \dots \tilde{s}_{n-1} d$ for
short. Similarly we can show that
$e \circ_i \delta_j^\alpha c = \delta_i^+ \delta_j^\alpha c$,
$\Delta_i(\delta_j^\alpha c, e)$ and $\Delta_i(e, \delta_j^\alpha
c)$. Hence $e$ is the $r_i$-inverse of $\delta_j^\alpha c$.
\item Alternatively, the hypothesis implies
  that $\partial_{n,j}^\alpha c$ has an $R_{n-1,i}$-inverse, $d$ say, for each $i < j \leq n$.
  So
\begin{align*}
\partial_{n,j}^\alpha c \star_{n-1,i} d &= \epsilon_{n-1,i} \partial_{n-1,i}^- \partial_{n,j}^\alpha c, \\
s_{n-1} \dots s_j \delta_j^\alpha c \circ_i d &= \delta_i^- s_{n-1} \dots s_j \delta_j^\alpha c, \\
\delta_j^\alpha c \circ_i e &= \delta_i^- \delta_j^\alpha c,
\end{align*}
where we abbreviate $e = \tilde{s}_j \dots \tilde{s}_{n-1} d$. Again
we can prove that
$e \circ_i \delta_j^\alpha c = \delta_i^+ \delta_j^\alpha c$,
$\Delta_i(\delta_j^\alpha c, e)$ and $\Delta_i(e, \delta_j^\alpha c)$,
and $e$ is the $r_i$-inverse of $\delta_j^\alpha c$.
\end{itemize}
This shows that $c$ has an $r_i$-invertible $(n-1)$-shell and hence
the $r_i$-inverse $r_i c \in \Scal^{>n}$ by
Lemma~\ref{L:SingleSetInverseFaceStab}. It satisfies
$\Delta_i(c, r_i c)$, $\Delta_i(r_i c, c)$,
$c \circ_i r_i c = \delta_i^- c$ and
$r_i c \circ_i c = \delta_i^+ c$. It thus follows that
$r_i c \in \Ccal_n$, and that
\begin{align*}
c \star_{n,i} r_i c
&= c \circ_i r_i c
= \delta_i^- c
= \tilde{s}_i \dots \tilde{s}_{n-1} s_{n-1} \dots s_i \delta_i^- c
= \epsilon_{n,i} \partial_{n,i}^- c, \\
r_i c \star_{n,i} c
&= r_i c \circ_i c
= \delta_i^+ c
= \tilde{s}_i \dots \tilde{s}_{n-1} s_{n-1} \dots s_i \delta_i^+ c
= \epsilon_{n,i} \partial_{n,i}^+ c.
\end{align*}
Therefore, $c$ is $R_{n,i}$-invertible in $\Ccal$ and $\Ccal$ is a
cubical $(\omega,p)$-category.

Second, suppose $\Ccal$ is a classical $(\omega,p)$-category. Let
$\Scal = \FSG{\Ccal}$ and, for all $n > p$ and $i \geq 1$, pick a cell
$s$ in $\Scal^{>n}$ with an $r_i$-invertible $(n-1)$-shell. For each
$j \geq 1$ with $i \neq j$, $\delta_j^\alpha s$ then has an
$r_i$-inverse $t = r_i \delta_j^\alpha s$.  Pick representatives $c$
and $d$ in $\Ccal_n$ of $s$ and $t$, respectively, which is possible
because $s \in \Scal^{>n}$.  Then
\begin{align*}
  \delta_j^\alpha s \circ_i t = \delta_i^- \delta_j^\alpha s =
  \delta_i^- \delta_j^\alpha s \quad\text{ and }\quad
  \epsilon_{n,j} \partial_{n,j}^\alpha c \star_{n,i} d = \epsilon_{n,j} \partial_{n,j}^\alpha \epsilon_{n,i} \partial_{n,i}^- c.
\end{align*}
\begin{itemize}
\item If $i < j$, then
  $ \partial_{n,j}^\alpha c \star_{n-1,i} e = \partial_{n,j}^\alpha
  \epsilon_{n,i} \partial_{n,i}^- c = \epsilon_{n-1,i}
  \partial_{n-1,i}^- \partial_{n,j}^\alpha c$, where we have
  abbreviated $e = \partial_{n,j}^\alpha d$. Similarly we can show
  that
  $e \star_{n-1,i} \partial_{n,j}^\alpha c = \epsilon_{n-1,i}
  \partial_{n-1,i}^+ \partial_{n,j}^\alpha c$. Thus $e$ is the
  $R_{n-1,i}$-inverse of $\partial_{n,j}^\alpha c$.
\item Alternatively, if $i > j$, then
  $ \partial_{n,j}^\alpha c \star_{n-1,i-1} e = \partial_{n,j}^\alpha
  \epsilon_{n,i} \partial_{n,i}^- c = \epsilon_{n-1,i-1}
  \partial_{n-1,i-1}^- \partial_{n,j}^\alpha c$, where we abbreviate
  $e = \partial_{n,j}^\alpha d$. Similarly we can show that
  $e \star_{n-1,i-1} \partial_{n,j}^\alpha c = \epsilon_{n-1,i-1}
  \partial_{n-1,i-1}^+ \partial_{n,j}^\alpha c$. Hence once again $e$
  is the $R_{n-1,i-1}$-inverse of $\partial_{n,j}^\alpha c$.
\end{itemize}
It follows from the definition that $c$ has an $R_{n-1,i}$-invertible
shell and hence an $R_{n,i}$-inverse $R_{n,i} c \in \Ccal_n$.  It
satisfies
$c \star_{n,i} R_{n,i} c = \epsilon_{n,i} \partial_{n,i}^- c$ and
$R_{n,i} c \star_{n,i} c = \epsilon_{n,i} \partial_{n,i}^+ c$. Suppose
$u = [R_{n,i} c]_\sim \in \Scal$.  The previous equations then imply
that $s \circ_i u = \delta_i^- s$ and $u \circ_i s = \delta_i^+
s$. Hence $u$ is the $r_i$-inverse of $s$. This shows that $\Scal$ is
a single-set $(\omega,p)$-category.
\end{proof}

\section{Formalisation with Isabelle/HOL}
\label{S:FormalisationIsabelle}

The Isabelle/HOL proof assistant~\cite{NipkowPaulsonWenzel2002} has
been indispensable for developing our axiomatisation of single-set
cubical $\omega$-categories in
Section~\ref{SS:SingleSetHigherCategories}. 
In this section we
  describe our Isabelle components, report on our particular usage of
  its support for proof automation in taming these axioms and give an
  example of a non-trivial proof (of Proposition~\ref{P:OmegaZeroInv})
  using our formalisation.

\subsection{Isabelle/HOL in a nutshell}
\label{SS:nutshell}

Isabelle/HOL is based on a simply-typed
classical higher-order logic, which in practice often gives the
impression of working in typed set theory. Among similar proof
assistants such as Coq \cite{BertotCasteran2004} or Lean \cite{Lean},
it stands out due to its support for proof automation. On one hand,
Isabelle employs internal simplification and proof procedures as well
as external proof search tools for first-order logic, which can be
invoked using the \emph{Sledgehammer} tactic.  On the other, it
integrates SAT solvers for counterexample search using the
\emph{Nitpick} tactic. 

Yet these strengths come at a price: Isabelle's type classes, one of
its two main mechanisms for modelling and working with algebraic
hierarchies, allow only one single type parameter, which essentially
imposes a single-set approach to categories when using type
  classes. Consequently, the numbers~$n$ or~$\omega$ do not feature
as type parameters in the formalisation of $n$- or $\omega$-categories
and cannot be instantiated easily to fixed finite dimensions. Further,
Isabelle's type system does not support dependent types, which may
sometimes be desirable in mathematical specifications. 

To overcome the first limitation, Isabelle offers locales as a more
set-based specification mechanism.  This allows more than one type
parameter and hence the standard approach to categories with objects
and arrows. Yet a locale-based approach to formalising
  mathematics~\cite{Ballarin14}, which has been developed over many
  years, had recently to be revised~\cite{Ballarin20} for modelling
  more advanced mathematical concepts such as Grothendieck
  schemes~\cite{BordgPL21}, at the expense of losing many benefits of
  types and type checking as well as a weaker coupling with Isabelle's
  main libraries. It might therefore be desirable to use proof
assistants with more expressive type systems, such as Coq or Lean, for
formalising more advanced features of higher categories. However, more
work is needed to evaluate the strengths and weaknesses of different
proof assistants in this regard. The formalisation of higher
  categories would certainly provide excellent test cases.

\subsection{Formalising single-set categories}
\label{SS:catoid-isa}

As previously for globular $\omega$-categories \cite{CalkStruth24},
the basic Isabelle type class for our formalisation of cubical
$\omega$-categories is that of a \emph{catoid}, a structure mentioned
en passant in Remark~\ref{R:RemarkCatoids}. We start with recalling
the basic features of the formalisation of catoids and single-set
categories with Isabelle from the Archive of Formal
proofs~\cite{Struth23}.

\begin{isabellebody}
\isanewline
  \isacommand{class}\isamarkupfalse%
\ multimagma\ {\isacharequal}{\kern0pt}\ \isanewline
\ \ \isakeyword{fixes}\ mcomp\ {\isacharcolon}{\kern0pt}{\isacharcolon}{\kern0pt}\ {\isachardoublequoteopen}{\isacharprime}{\kern0pt}a\ {\isasymRightarrow}\ {\isacharprime}{\kern0pt}a\ {\isasymRightarrow}\ {\isacharprime}{\kern0pt}a\ set{\isachardoublequoteclose}\ {\isacharparenleft}{\kern0pt}\isakeyword{infixl}\ {\isachardoublequoteopen}{\isasymodot}{\isachardoublequoteclose}\ {\isadigit{7}}{\isadigit{0}}{\isacharparenright}{\kern0pt}\ \isanewline
\isanewline
\isacommand{class}\isamarkupfalse%
\ multisemigroup\ {\isacharequal}{\kern0pt}\ multimagma\ {\isacharplus}{\kern0pt}\isanewline
\ \ \isakeyword{assumes}\ assoc{\isacharcolon}{\kern0pt}\ {\isachardoublequoteopen}{\isacharparenleft}{\kern0pt}{\isasymUnion}v\ {\isasymin}\ y\ {\isasymodot}\ z{\isachardot}{\kern0pt}\ x\ {\isasymodot}\ v{\isacharparenright}{\kern0pt}\ {\isacharequal}{\kern0pt}\ {\isacharparenleft}{\kern0pt}{\isasymUnion}v\ {\isasymin}\ x\ {\isasymodot}\ y{\isachardot}{\kern0pt}\ v\ {\isasymodot}\ z{\isacharparenright}{\kern0pt}{\isachardoublequoteclose}\isanewline
\isanewline
\isacommand{class}\isamarkupfalse%
\ st{\isacharunderscore}{\kern0pt}op\ {\isacharequal}{\kern0pt}\ \isanewline
\ \ \isakeyword{fixes}\ src\ {\isacharcolon}{\kern0pt}{\isacharcolon}{\kern0pt}\ {\isachardoublequoteopen}{\isacharprime}{\kern0pt}a\ {\isasymRightarrow}\ {\isacharprime}{\kern0pt}a{\isachardoublequoteclose}\ {\isacharparenleft}{\kern0pt}{\isachardoublequoteopen}{\isasymsigma}{\isachardoublequoteclose}{\isacharparenright}{\kern0pt}\isanewline
\ \ \isakeyword{and}\ tgt\ {\isacharcolon}{\kern0pt}{\isacharcolon}{\kern0pt}\ {\isachardoublequoteopen}{\isacharprime}{\kern0pt}a\ {\isasymRightarrow}\ {\isacharprime}{\kern0pt}a{\isachardoublequoteclose}\ {\isacharparenleft}{\kern0pt}{\isachardoublequoteopen}{\isasymtau}{\isachardoublequoteclose}{\isacharparenright}{\kern0pt}\isanewline
\isanewline
\isacommand{class}\isamarkupfalse%
\ st{\isacharunderscore}{\kern0pt}multimagma\ {\isacharequal}{\kern0pt}\ multimagma\ {\isacharplus}{\kern0pt}\ st{\isacharunderscore}{\kern0pt}op\ {\isacharplus}{\kern0pt}\isanewline
\ \ \isakeyword{assumes}\ Dst{\isacharcolon}{\kern0pt}\ {\isachardoublequoteopen}x\ {\isasymodot}\ y\ {\isasymnoteq}\ {\isacharbraceleft}{\kern0pt}{\isacharbraceright}{\kern0pt}\ {\isasymLongrightarrow}\ {\isasymtau}\ x\ {\isacharequal}{\kern0pt}\ {\isasymsigma}\ y{\isachardoublequoteclose}\isanewline
\ \ \isakeyword{and}\ s{\isacharunderscore}{\kern0pt}absorb\ {\isacharbrackleft}{\kern0pt}simp{\isacharbrackright}{\kern0pt}{\isacharcolon}{\kern0pt}\ {\isachardoublequoteopen}{\isasymsigma}\ x\ {\isasymodot}\ x\ {\isacharequal}{\kern0pt}\ {\isacharbraceleft}{\kern0pt}x{\isacharbraceright}{\kern0pt}{\isachardoublequoteclose}\ \isanewline
\ \ \isakeyword{and}\ t{\isacharunderscore}{\kern0pt}absorb\
{\isacharbrackleft}{\kern0pt}simp{\isacharbrackright}{\kern0pt}{\isacharcolon}{\kern0pt}\
{\isachardoublequoteopen}x\ {\isasymodot}\ {\isasymtau}\ x\
{\isacharequal}{\kern0pt}\
{\isacharbraceleft}{\kern0pt}x{\isacharbraceright}{\kern0pt}{\isachardoublequoteclose}%
\isanewline
\isanewline
\isacommand{class}\isamarkupfalse%
\ catoid\ {\isacharequal}{\kern0pt}\ st{\isacharunderscore}{\kern0pt}multimagma\ {\isacharplus}{\kern0pt}\ multisemigroup\isanewline
\end{isabellebody}
The type classes \isa{multimagma} and \isa{st-op} introduce a
multioperation and source and target maps, together with notation
$\odot$, $\sigma$ and $\tau$.  The classes \isa{multisemigroup},
\isa{st-op} and \isa{catoid} add structure,
extending the classes previously defined. The catoid class, for
instance, extends the st-multimagma and multisemigroup classes, while
the multisemigroup class extends the class defining its
multiplication.

Catoids are extended to single-set categories by
imposing the locality and functionality axioms from
Section~\ref{SSS:SingleSetOneCategories}.

\begin{isabellebody}
  \isanewline
  \isacommand{class}\isamarkupfalse%
\ single{\isacharunderscore}{\kern0pt}set{\isacharunderscore}{\kern0pt}category\ {\isacharequal}{\kern0pt}\ functional{\isacharunderscore}{\kern0pt}catoid\ {\isacharplus}{\kern0pt}\ local{\isacharunderscore}{\kern0pt}catoid\isanewline
\end{isabellebody}

Each of the above type class features one single type parameter
\isa{{\isachardoublequoteopen}{\isacharprime}{\kern0pt}a} (spelled
$\alpha$) and is polymorphic in this parameter. It can therefore be
instantiated to more concrete types. In the class
\isa{multisemigroup}, for instance, $\alpha$ could be instantiated to
the type of strings and $\odot$ to the shuffle operation on
strings.

Note also that multisemigroups or catoids have been specified without
carrier sets. While such sets can be added easily, it often suffices
to regard the type $\alpha$ roughly as a set.

The standard function type in proof assistants such as Coq, Isabelle
or Lean models total functions. Partiality is usually modelled using a
monad or by adjoining a zero. Here instead we take the multioperation
$\odot$ as a basis, mapping to the empty set when two elements cannot
be composed. We define the resulting domain of definition as
\begin{isabellebody}
  \isanewline
  \isacommand{abbreviation}\isamarkupfalse%
\ {\isachardoublequoteopen}{\isasymDelta}\ x\ y\ {\isasymequiv}\
{\isacharparenleft}{\kern0pt}x\ {\isasymodot}\ y\ {\isasymnoteq}\
{\isacharbraceleft}{\kern0pt}{\isacharbraceright}{\kern0pt}{\isacharparenright}{\kern0pt}{\isachardoublequoteclose}%
\isanewline
\end{isabellebody}
\noindent and the partial operation $\otimes$ (denoted $\circ$ in
previous sections), using the definite description operator \isa{THE}
as
\begin{isabellebody}
 \isanewline
\isacommand{definition}\isamarkupfalse%
\ pcomp\ {\isacharcolon}{\kern0pt}{\isacharcolon}{\kern0pt}\ {\isachardoublequoteopen}{\isacharprime}{\kern0pt}a\ {\isasymRightarrow}\ {\isacharprime}{\kern0pt}a\ {\isasymRightarrow}\ {\isacharprime}{\kern0pt}a{\isachardoublequoteclose}\ {\isacharparenleft}{\kern0pt}\isakeyword{infixl}\ {\isachardoublequoteopen}{\isasymotimes}{\isachardoublequoteclose}\ {\isadigit{7}}{\isadigit{0}}{\isacharparenright}{\kern0pt}\ \isakeyword{where}\isanewline
\ \ {\isachardoublequoteopen}x\ {\isasymotimes}\ y\ {\isacharequal}{\kern0pt}\ {\isacharparenleft}{\kern0pt}THE\ z{\isachardot}{\kern0pt}\ z\ {\isasymin}\ x\ {\isasymodot}\ y{\isacharparenright}{\kern0pt}{\isachardoublequoteclose}\isanewline
\end{isabellebody}
\noindent If $x\odot y$ is empty, Isabelle maps $x\otimes y$ to a
value about which nothing particular can be proved. Using
$\otimes$ in place of $\odot$ with single-set categories allows us to
be precise about definedness conditions in higher categories, which
may be subtle, while avoiding clumsy specifications with many set
braces and a proliferation of cases in proofs due to
undefinedness. In proof assistants with dependent types, the
  partiality of composition in categories can alternatively be
  formalised at type level, defining arrow composition on
  homsets. This approach might be harder do integrate with tools like
  Sledgehammer.

\subsection{Formalising single-set cubical $\omega$-categories}
\label{SS:cubical-isabelle}

To formalise $\omega$-categories, we have first created indexed
variants of the classes leading to \isa{single-set-category} above,
that is, classes based on $\odot_i$, $\otimes_i$ for the compositions
and $\partial\, i\, \alpha$ for face maps. Unlike in previous
sections, our indices start with $0$ and we write $\partial$ instead
of $\delta$ for single-set face maps. With Isabelle, we can formally
link these indexed classes with the non-indexed ones (using so-called
sublocale statements between classes) so that all theorems about
single-set categories are in scope in the indexed variants. In
particular, we have introduced the definedness conditions \isa{DD i}
and linked them formally with $\Delta$.

Using the class \isa{icategory} for indexed single-set categories, we
have first defined an auxiliary class for $\omega$-categories without symmetries.

\begin{isabellebody}
  \isanewline 
  \isacommand{class}\isamarkupfalse%
\ semi{\isacharunderscore}{\kern0pt}cubical{\isacharunderscore}{\kern0pt}omega{\isacharunderscore}{\kern0pt}category\ {\isacharequal}{\kern0pt}\ icategory\ {\isacharplus}{\kern0pt}\isanewline
\ \ \isakeyword{assumes}\ face{\isacharunderscore}{\kern0pt}comm{\isacharcolon}{\kern0pt}\ {\isachardoublequoteopen}i\ {\isasymnoteq}\ j\ {\isasymLongrightarrow}\ {\isasympartial}\ i\ {\isasymalpha}\ {\isasymcirc}\ {\isasympartial}\ j\ {\isasymbeta}\ {\isacharequal}{\kern0pt}\ {\isasympartial}\ j\ {\isasymbeta}\ {\isasymcirc}\ {\isasympartial}\ i\ {\isasymalpha}{\isachardoublequoteclose}\isanewline
\ \ \isakeyword{and}\ face{\isacharunderscore}{\kern0pt}func{\isacharcolon}{\kern0pt}\ {\isachardoublequoteopen}i\ {\isasymnoteq}\ j\ {\isasymLongrightarrow}\ DD\ j\ x\ y\ {\isasymLongrightarrow}\ {\isasympartial}\ i\ {\isasymalpha}\ {\isacharparenleft}{\kern0pt}x\ {\isasymotimes}\isactrlbsub j\isactrlesub \ y{\isacharparenright}{\kern0pt}\ {\isacharequal}{\kern0pt}\ {\isasympartial}\ i\ {\isasymalpha}\ x\ {\isasymotimes}\isactrlbsub j\isactrlesub \ {\isasympartial}\ i\ {\isasymalpha}\ y{\isachardoublequoteclose}\isanewline
\ \ \isakeyword{and}\ interchange{\isacharcolon}{\kern0pt}\ {\isachardoublequoteopen}i\ {\isasymnoteq}\ j\ {\isasymLongrightarrow}\ DD\ i\ w\ x\ {\isasymLongrightarrow}\ DD\ i\ y\ z\ {\isasymLongrightarrow}\ DD\ j\ w\ y\ {\isasymLongrightarrow}\ DD\ j\ x\ z\ \isanewline
\ \ \ \ \ \ \ \ \ \ \ \ \ \ \ \ \ \ \ \ \ \ \ \ \ \ \ {\isasymLongrightarrow}\ {\isacharparenleft}{\kern0pt}w\ {\isasymotimes}\isactrlbsub i\isactrlesub \ x{\isacharparenright}{\kern0pt}\ {\isasymotimes}\isactrlbsub j\isactrlesub \ {\isacharparenleft}{\kern0pt}y\ {\isasymotimes}\isactrlbsub i\isactrlesub \ z{\isacharparenright}{\kern0pt}\ {\isacharequal}{\kern0pt}\ {\isacharparenleft}{\kern0pt}w\ {\isasymotimes}\isactrlbsub j\isactrlesub \ y{\isacharparenright}{\kern0pt}\ {\isasymotimes}\isactrlbsub i\isactrlesub \ {\isacharparenleft}{\kern0pt}x\ {\isasymotimes}\isactrlbsub j\isactrlesub \ z{\isacharparenright}{\kern0pt}{\isachardoublequoteclose}\isanewline
\ \ \isakeyword{and}\ fin{\isacharunderscore}{\kern0pt}fix{\isacharcolon}{\kern0pt}\ {\isachardoublequoteopen}\ {\isasymexists}k{\isachardot}\ {\isasymforall}i{\isachardot}\ k\ {\isasymle}\ i\ {\isasymlongrightarrow}\ fFx\ i\ x{\isachardoublequoteclose}\isanewline
\end{isabellebody}

\noindent In the last axiom, \isa{{\isachardoublequoteopen}fFx\ i\ x\ {\isasymequiv}\
  {\isacharparenleft}{\kern0pt}{\isasympartial}\ i\ ff\ x\
  {\isacharequal}{\kern0pt}\
  x{\isacharparenright}{\kern0pt}{\isachardoublequoteclose}}, which we
use in place of the predicate $x \in \Scal^i$ from
Section~\ref{SS:SingleSetCubicalCategories}, has been defined in the
context of a class on which \isa{icategory} is based.

We have further extended this class to one for 
$\omega$-categories with the remaining axioms for symmetries and
reverse symmetries from
Definition~\ref{SSS:DefCubSingleSetWithoutConn}, after introducing a
separate class for these two maps.

\begin{isabellebody}
 \isanewline
\isacommand{class}\isamarkupfalse%
\ symmetry{\isacharunderscore}{\kern0pt}ops\ {\isacharequal}{\kern0pt}\isanewline
\ \ \isakeyword{fixes}\ symmetry\ {\isacharcolon}{\kern0pt}{\isacharcolon}{\kern0pt}\ {\isachardoublequoteopen}nat\ {\isasymRightarrow}\ {\isacharprime}{\kern0pt}a\ {\isasymRightarrow}\ {\isacharprime}{\kern0pt}a{\isachardoublequoteclose}\ {\isacharparenleft}{\kern0pt}{\isachardoublequoteopen}{\isasymsigma}{\isachardoublequoteclose}{\isacharparenright}{\kern0pt}\isanewline
\ \ \isakeyword{and}\ inv{\isacharunderscore}{\kern0pt}symmetry\
{\isacharcolon}{\kern0pt}{\isacharcolon}{\kern0pt}\
{\isachardoublequoteopen}nat\ {\isasymRightarrow}\
{\isacharprime}{\kern0pt}a\ {\isasymRightarrow}\
{\isacharprime}{\kern0pt}a{\isachardoublequoteclose}\
{\isacharparenleft}{\kern0pt}{\isachardoublequoteopen}{\isasymtheta}{\isachardoublequoteclose}{\isacharparenright}{\kern0pt}\
\isanewline
   \isanewline
  \isacommand{class}\isamarkupfalse%
\ cubical{\isacharunderscore}{\kern0pt}omega{\isacharunderscore}{\kern0pt}category\ {\isacharequal}{\kern0pt}\ semi{\isacharunderscore}{\kern0pt}cubical{\isacharunderscore}{\kern0pt}omega{\isacharunderscore}{\kern0pt}category\ {\isacharplus}{\kern0pt}\ symmetry{\isacharunderscore}{\kern0pt}ops\ {\isacharplus}{\kern0pt}\isanewline
\ \ \isakeyword{assumes}\ sym{\isacharunderscore}{\kern0pt}type{\isacharcolon}{\kern0pt}\ {\isachardoublequoteopen}{\isasymsigma}{\isasymsigma}\ i\ {\isacharparenleft}{\kern0pt}face{\isacharunderscore}{\kern0pt}fix\ i{\isacharparenright}{\kern0pt}\ {\isasymsubseteq}\ face{\isacharunderscore}{\kern0pt}fix\ {\isacharparenleft}{\kern0pt}i\ {\isacharplus}{\kern0pt}\ {\isadigit{1}}{\isacharparenright}{\kern0pt}{\isachardoublequoteclose}\isanewline
\ \ \isakeyword{and}\ inv{\isacharunderscore}{\kern0pt}sym{\isacharunderscore}{\kern0pt}type{\isacharcolon}{\kern0pt}\ {\isachardoublequoteopen}{\isasymtheta}{\isasymtheta}\ i\ {\isacharparenleft}{\kern0pt}face{\isacharunderscore}{\kern0pt}fix\ {\isacharparenleft}{\kern0pt}i\ {\isacharplus}{\kern0pt}\ {\isadigit{1}}{\isacharparenright}{\kern0pt}{\isacharparenright}{\kern0pt}\ {\isasymsubseteq}\ face{\isacharunderscore}{\kern0pt}fix\ i{\isachardoublequoteclose}\isanewline
\ \ \isakeyword{and}\ sym{\isacharunderscore}{\kern0pt}inv{\isacharunderscore}{\kern0pt}sym{\isacharcolon}{\kern0pt}\ {\isachardoublequoteopen}fFx\ {\isacharparenleft}{\kern0pt}i\ {\isacharplus}{\kern0pt}\ {\isadigit{1}}{\isacharparenright}{\kern0pt}\ x\ {\isasymLongrightarrow}\ {\isasymsigma}\ i\ {\isacharparenleft}{\kern0pt}{\isasymtheta}\ i\ x{\isacharparenright}{\kern0pt}\ {\isacharequal}{\kern0pt}\ x{\isachardoublequoteclose}\isanewline
\ \ \isakeyword{and}\ inv{\isacharunderscore}{\kern0pt}sym{\isacharunderscore}{\kern0pt}sym{\isacharcolon}{\kern0pt}\ {\isachardoublequoteopen}fFx\ i\ x\ \ {\isasymLongrightarrow}\ {\isasymtheta}\ i\ {\isacharparenleft}{\kern0pt}{\isasymsigma}\ i\ x{\isacharparenright}{\kern0pt}\ {\isacharequal}{\kern0pt}\ x{\isachardoublequoteclose}\isanewline
\ \ \isakeyword{and}\ sym{\isacharunderscore}{\kern0pt}face{\isadigit{1}}{\isacharcolon}{\kern0pt}\ {\isachardoublequoteopen}fFx\ i\ x\ {\isasymLongrightarrow}\ {\isasympartial}\ i\ {\isasymalpha}\ {\isacharparenleft}{\kern0pt}{\isasymsigma}\ i\ x{\isacharparenright}{\kern0pt}\ {\isacharequal}{\kern0pt}\ {\isasymsigma}\ i\ {\isacharparenleft}{\kern0pt}{\isasympartial}\ {\isacharparenleft}{\kern0pt}i\ {\isacharplus}{\kern0pt}\ {\isadigit{1}}{\isacharparenright}{\kern0pt}\ {\isasymalpha}\ x{\isacharparenright}{\kern0pt}{\isachardoublequoteclose}\isanewline
\ \ \isakeyword{and}\ sym{\isacharunderscore}{\kern0pt}face{\isadigit{2}}{\isacharcolon}{\kern0pt}\ {\isachardoublequoteopen}i\ {\isasymnoteq}\ j\ {\isasymLongrightarrow}\ i\ {\isasymnoteq}\ j\ {\isacharplus}{\kern0pt}\ {\isadigit{1}}\ {\isasymLongrightarrow}\ fFx\ j\ x\ {\isasymLongrightarrow}\ {\isasympartial}\ i\ {\isasymalpha}\ {\isacharparenleft}{\kern0pt}{\isasymsigma}\ j\ x{\isacharparenright}{\kern0pt}\ {\isacharequal}{\kern0pt}\ {\isasymsigma}\ j\ {\isacharparenleft}{\kern0pt}{\isasympartial}\ i\ {\isasymalpha}\ x{\isacharparenright}{\kern0pt}{\isachardoublequoteclose}\isanewline
\ \ \isakeyword{and}\ sym{\isacharunderscore}{\kern0pt}func{\isacharcolon}{\kern0pt}\ {\isachardoublequoteopen}i\ {\isasymnoteq}\ j\ {\isasymLongrightarrow}\ fFx\ i\ x\ {\isasymLongrightarrow}\ fFx\ i\ y\ {\isasymLongrightarrow}\ DD\ j\ x\ y\ {\isasymLongrightarrow}\ \isanewline
\ \ \ \ \ \ \ \ \ \ \ \ \ \ \ \ \ \ \ \ \ {\isasymsigma}\ i\ {\isacharparenleft}{\kern0pt}x\ {\isasymotimes}\isactrlbsub j\isactrlesub \ y{\isacharparenright}{\kern0pt}\ {\isacharequal}{\kern0pt}\ {\isacharparenleft}{\kern0pt}if\ j\ {\isacharequal}{\kern0pt}\ i\ {\isacharplus}{\kern0pt}\ {\isadigit{1}}\ then\ {\isasymsigma}\ i\ x\ {\isasymotimes}\isactrlbsub i\isactrlesub \ {\isasymsigma}\ i\ y\ else\ {\isasymsigma}\ i\ x\ {\isasymotimes}\isactrlbsub j\isactrlesub \ {\isasymsigma}\ i\ y{\isacharparenright}{\kern0pt}{\isachardoublequoteclose}\isanewline
\ \ \isakeyword{and}\ sym{\isacharunderscore}{\kern0pt}fix{\isacharcolon}{\kern0pt}\ {\isachardoublequoteopen}fFx\ i\ x\ {\isasymLongrightarrow}\ fFx\ {\isacharparenleft}{\kern0pt}i\ {\isacharplus}{\kern0pt}\ {\isadigit{1}}{\isacharparenright}{\kern0pt}\ x\ {\isasymLongrightarrow}\ {\isasymsigma}\ i\ x\ {\isacharequal}{\kern0pt}\ x{\isachardoublequoteclose}\isanewline
\ \ \isakeyword{and}\ sym{\isacharunderscore}{\kern0pt}sym{\isacharunderscore}{\kern0pt}braid{\isacharcolon}{\kern0pt}\ {\isachardoublequoteopen}diffSup\ i\ j\ {\isadigit{2}}\ {\isasymLongrightarrow}\ fFx\ i\ x\ {\isasymLongrightarrow}\ fFx\ j\ x\ \ {\isasymLongrightarrow}\ {\isasymsigma}\ i\ {\isacharparenleft}{\kern0pt}{\isasymsigma}\ j\ x{\isacharparenright}{\kern0pt}\ {\isacharequal}{\kern0pt}\ {\isasymsigma}\ j\ {\isacharparenleft}{\kern0pt}{\isasymsigma}\ i\ x{\isacharparenright}{\kern0pt}{\isachardoublequoteclose}\isanewline
\end{isabellebody}

\noindent In Axioms \isa{sym-type} and \isa{inv-sym-type}, the
functions $\sigma\sigma$ and $\vartheta\vartheta$ are the image maps
corresponding to symmetries and reverse symmetries. Further
\isa{face-fix i} denotes the set $\Scal^i$ of fixed points of the
lower face map in direction $i$. It is defined as
\isa{{\isachardoublequoteopen}face{\isacharunderscore}{\kern0pt}fix\
  i\ {\isasymequiv}\ Fix\
  {\isacharparenleft}{\kern0pt}{\isasympartial}\ i\
  ff{\isacharparenright}{\kern0pt}{\isachardoublequoteclose}} as in
Section~\ref{SS:SingleSetCubicalCategories}. Other axioms in the class
use the predicate \isa{{\isachardoublequoteopen}diffSup\ i\ j\ k\
  {\isasymequiv}\ {\isacharparenleft}{\kern0pt}i\
  {\isacharminus}{\kern0pt}\ j\ {\isasymge}\ k\ {\isasymor}\ j\
  {\isacharminus}{\kern0pt}\ i\ {\isasymge}\
  k{\isacharparenright}{\kern0pt}{\isachardoublequoteclose}}.

Though our formalisation is a so-called shallow
embedding of categories in Isabelle, as it uses
Isabelle's built-in types for functions, sets and numbers to
axiomatise $\omega$-categories, it has nevertheless some deep
features, as we do not define a (sub)type for each \isa{
  {\isachardoublequoteopen}fFx\ i\ x\ {\isasymequiv}\
  {\isacharparenleft}{\kern0pt}{\isasympartial}\ i\ ff\ x\
  {\isacharequal}{\kern0pt}\
  x{\isacharparenright}{\kern0pt}{\isachardoublequoteclose}} and we
capture the partiality of cell composition in terms of the predicate
\isa{DD}, but not at type level. Such typing or composition conditions
must therefore be declared explicitly in axioms and lemmas, and they
need to be checked explicitly in proofs.

Finally, we have provided a class for connections and defined a class for
$\omega$-categories with connections.

\begin{isabellebody}
  \isanewline
  \isacommand{class}\isamarkupfalse%
\ cubical{\isacharunderscore}{\kern0pt}omega{\isacharunderscore}{\kern0pt}category{\isacharunderscore}{\kern0pt}connections\ {\isacharequal}{\kern0pt}\ cubical{\isacharunderscore}{\kern0pt}omega{\isacharunderscore}{\kern0pt}category\ {\isacharplus}{\kern0pt}\ connection{\isacharunderscore}{\kern0pt}ops\ {\isacharplus}{\kern0pt}\isanewline
\ \ \isakeyword{assumes}\ conn{\isacharunderscore}{\kern0pt}face{\isadigit{1}}{\isacharcolon}{\kern0pt}\ {\isachardoublequoteopen}fFx\ j\ x\ {\isasymLongrightarrow}\ {\isasympartial}\ j\ {\isasymalpha}\ {\isacharparenleft}{\kern0pt}{\isasymGamma}\ j\ {\isasymalpha}\ x{\isacharparenright}{\kern0pt}\ {\isacharequal}{\kern0pt}\ x{\isachardoublequoteclose}\isanewline
\ \ \isakeyword{and}\ conn{\isacharunderscore}{\kern0pt}face{\isadigit{2}}{\isacharcolon}{\kern0pt}\ {\isachardoublequoteopen}fFx\ j\ x\ {\isasymLongrightarrow}\ {\isasympartial}\ {\isacharparenleft}{\kern0pt}j\ {\isacharplus}{\kern0pt}\ {\isadigit{1}}{\isacharparenright}{\kern0pt}\ {\isasymalpha}\ {\isacharparenleft}{\kern0pt}{\isasymGamma}\ j\ {\isasymalpha}\ x{\isacharparenright}{\kern0pt}\ {\isacharequal}{\kern0pt}\ {\isasymsigma}\ j\ x{\isachardoublequoteclose}\isanewline
\ \ \isakeyword{and}\ conn{\isacharunderscore}{\kern0pt}face{\isadigit{3}}{\isacharcolon}{\kern0pt}\ {\isachardoublequoteopen}i\ {\isasymnoteq}\ j\ {\isasymLongrightarrow}\ i\ {\isasymnoteq}\ j\ {\isacharplus}{\kern0pt}\ {\isadigit{1}}\ {\isasymLongrightarrow}\ fFx\ j\ x\ {\isasymLongrightarrow}\ {\isasympartial}\ i\ {\isasymalpha}\ {\isacharparenleft}{\kern0pt}{\isasymGamma}\ j\ {\isasymbeta}\ x{\isacharparenright}{\kern0pt}\ {\isacharequal}{\kern0pt}\ {\isasymGamma}\ j\ {\isasymbeta}\ {\isacharparenleft}{\kern0pt}{\isasympartial}\ i\ {\isasymalpha}\ x{\isacharparenright}{\kern0pt}{\isachardoublequoteclose}\isanewline
\ \ \isakeyword{and}\
conn{\isacharunderscore}{\kern0pt}corner{\isadigit{1}}{\isacharcolon}{\kern0pt}\
{\isachardoublequoteopen}fFx\ i\ x\ {\isasymLongrightarrow}\ fFx\ i\
y\ {\isasymLongrightarrow}\ DD\ {\isacharparenleft}{\kern0pt}i\
{\isacharplus}{\kern0pt}\
{\isadigit{1}}{\isacharparenright}{\kern0pt}\ x\ y\
{\isasymLongrightarrow}\isanewline
\ \ \ \ \ \ \ \ \ \ \ \ \ \ \ \ \ \ \ \ {\isasymGamma}\ i\ tt\ {\isacharparenleft}{\kern0pt}x\ {\isasymotimes}\isactrlbsub {\isacharparenleft}{\kern0pt}i\ {\isacharplus}{\kern0pt}\ {\isadigit{1}}{\isacharparenright}{\kern0pt}\isactrlesub \ y{\isacharparenright}{\kern0pt}\ {\isacharequal}{\kern0pt}\ {\isacharparenleft}{\kern0pt}{\isasymGamma}\ i\ tt\ x\ {\isasymotimes}\isactrlbsub {\isacharparenleft}{\kern0pt}i\ {\isacharplus}{\kern0pt}\ {\isadigit{1}}{\isacharparenright}{\kern0pt}\isactrlesub \ {\isasymsigma}\ i\ x{\isacharparenright}{\kern0pt}\ {\isasymotimes}\isactrlbsub i\isactrlesub \ {\isacharparenleft}{\kern0pt}x\ {\isasymotimes}\isactrlbsub {\isacharparenleft}{\kern0pt}i\ {\isacharplus}{\kern0pt}\ {\isadigit{1}}{\isacharparenright}{\kern0pt}\isactrlesub \ {\isasymGamma}\ i\ tt\ y{\isacharparenright}{\kern0pt}{\isachardoublequoteclose}\ \isanewline
\ \ \isakeyword{and}\ conn{\isacharunderscore}{\kern0pt}corner{\isadigit{2}}{\isacharcolon}{\kern0pt}\ {\isachardoublequoteopen}fFx\ i\ x\ {\isasymLongrightarrow}\ fFx\ i\ y\ {\isasymLongrightarrow}\ DD\ {\isacharparenleft}{\kern0pt}i\ {\isacharplus}{\kern0pt}\ {\isadigit{1}}{\isacharparenright}{\kern0pt}\ x\ y\ {\isasymLongrightarrow}\isanewline
\ \ \ \ \ \ \ \ \ \ \ \ \ \ \ \ \ \ \ \ {\isasymGamma}\ i\ ff\ {\isacharparenleft}{\kern0pt}x\ {\isasymotimes}\isactrlbsub {\isacharparenleft}{\kern0pt}i\ {\isacharplus}{\kern0pt}\ {\isadigit{1}}{\isacharparenright}{\kern0pt}\isactrlesub \ y{\isacharparenright}{\kern0pt}\ {\isacharequal}{\kern0pt}\ {\isacharparenleft}{\kern0pt}{\isasymGamma}\ i\ ff\ x\ {\isasymotimes}\isactrlbsub {\isacharparenleft}{\kern0pt}i\ {\isacharplus}{\kern0pt}\ {\isadigit{1}}{\isacharparenright}{\kern0pt}\isactrlesub \ y{\isacharparenright}{\kern0pt}\ {\isasymotimes}\isactrlbsub i\isactrlesub \ {\isacharparenleft}{\kern0pt}{\isasymsigma}\ i\ y\ {\isasymotimes}\isactrlbsub {\isacharparenleft}{\kern0pt}i\ {\isacharplus}{\kern0pt}\ {\isadigit{1}}{\isacharparenright}{\kern0pt}\isactrlesub \ {\isasymGamma}\ i\ ff\ y{\isacharparenright}{\kern0pt}{\isachardoublequoteclose}\isanewline
\ \ \isakeyword{and}\ conn{\isacharunderscore}{\kern0pt}corner{\isadigit{3}}{\isacharcolon}{\kern0pt}\ {\isachardoublequoteopen}j\ {\isasymnoteq}\ i\ {\isasymand}\ j\ {\isasymnoteq}\ i\ {\isacharplus}{\kern0pt}\ {\isadigit{1}}\ {\isasymLongrightarrow}\ fFx\ i\ x\ {\isasymLongrightarrow}\ fFx\ i\ y\ {\isasymLongrightarrow}\ DD\ j\ x\ y\ {\isasymLongrightarrow}\ {\isasymGamma}\ i\ {\isasymalpha}\ {\isacharparenleft}{\kern0pt}x\ {\isasymotimes}\isactrlbsub j\isactrlesub \ y{\isacharparenright}{\kern0pt}\ {\isacharequal}{\kern0pt}\ {\isasymGamma}\ i\ {\isasymalpha}\ x\ {\isasymotimes}\isactrlbsub j\isactrlesub \ {\isasymGamma}\ i\ {\isasymalpha}\ y{\isachardoublequoteclose}\isanewline
\ \ \isakeyword{and}\ conn{\isacharunderscore}{\kern0pt}fix{\isacharcolon}{\kern0pt}\ {\isachardoublequoteopen}fFx\ i\ x\ {\isasymLongrightarrow}\ fFx\ {\isacharparenleft}{\kern0pt}i\ {\isacharplus}{\kern0pt}\ {\isadigit{1}}{\isacharparenright}{\kern0pt}\ x\ {\isasymLongrightarrow}\ {\isasymGamma}\ i\ {\isasymalpha}\ x\ {\isacharequal}{\kern0pt}\ x{\isachardoublequoteclose}\isanewline
\ \ \isakeyword{and}\ conn{\isacharunderscore}{\kern0pt}zigzag{\isadigit{1}}{\isacharcolon}{\kern0pt}\ {\isachardoublequoteopen}fFx\ i\ x\ {\isasymLongrightarrow}\ {\isasymGamma}\ i\ tt\ x\ {\isasymotimes}\isactrlbsub {\isacharparenleft}{\kern0pt}i\ {\isacharplus}{\kern0pt}\ {\isadigit{1}}{\isacharparenright}{\kern0pt}\isactrlesub \ {\isasymGamma}\ i\ ff\ x\ {\isacharequal}{\kern0pt}\ x{\isachardoublequoteclose}\isanewline
\ \ \isakeyword{and}\ conn{\isacharunderscore}{\kern0pt}zigzag{\isadigit{2}}{\isacharcolon}{\kern0pt}\ {\isachardoublequoteopen}fFx\ i\ x\ {\isasymLongrightarrow}\ {\isasymGamma}\ i\ tt\ x\ {\isasymotimes}\isactrlbsub i\isactrlesub \ {\isasymGamma}\ i\ ff\ x\ {\isacharequal}{\kern0pt}\ {\isasymsigma}\ i\ x{\isachardoublequoteclose}\isanewline
\ \ \isakeyword{and}\ conn{\isacharunderscore}{\kern0pt}conn{\isacharunderscore}{\kern0pt}braid{\isacharcolon}{\kern0pt}\ {\isachardoublequoteopen}diffSup\ i\ j\ {\isadigit{2}}\ {\isasymLongrightarrow}\ fFx\ j\ x\ {\isasymLongrightarrow}\ fFx\ i\ x\ {\isasymLongrightarrow}\ {\isasymGamma}\ i\ {\isasymalpha}\ {\isacharparenleft}{\kern0pt}{\isasymGamma}\ j\ {\isasymbeta}\ x{\isacharparenright}{\kern0pt}\ {\isacharequal}{\kern0pt}\ {\isasymGamma}\ j\ {\isasymbeta}\ {\isacharparenleft}{\kern0pt}{\isasymGamma}\ i\ {\isasymalpha}\ x{\isacharparenright}{\kern0pt}{\isachardoublequoteclose}\isanewline
\ \ \isakeyword{and}\ conn{\isacharunderscore}{\kern0pt}shift{\isacharcolon}{\kern0pt}\ {\isachardoublequoteopen}fFx\ i\ x\ {\isasymLongrightarrow}\ fFx\ {\isacharparenleft}{\kern0pt}i\ {\isacharplus}{\kern0pt}\ {\isadigit{1}}{\isacharparenright}{\kern0pt}\ x\ {\isasymLongrightarrow}\ {\isasymsigma}\ {\isacharparenleft}{\kern0pt}i\ {\isacharplus}{\kern0pt}\ {\isadigit{1}}{\isacharparenright}{\kern0pt}\ {\isacharparenleft}{\kern0pt}{\isasymsigma}\ i\ {\isacharparenleft}{\kern0pt}{\isasymGamma}\ {\isacharparenleft}{\kern0pt}i\ {\isacharplus}{\kern0pt}\ {\isadigit{1}}{\isacharparenright}{\kern0pt}\ {\isasymalpha}\ x{\isacharparenright}{\kern0pt}{\isacharparenright}{\kern0pt}\ {\isacharequal}{\kern0pt}\ {\isasymGamma}\ i\ {\isasymalpha}\ {\isacharparenleft}{\kern0pt}{\isasymsigma}\ {\isacharparenleft}{\kern0pt}i\ {\isacharplus}{\kern0pt}\ {\isadigit{1}}{\isacharparenright}{\kern0pt}\ x{\isacharparenright}{\kern0pt}{\isachardoublequoteclose}\isanewline
\end{isabellebody}

\subsection{Example proofs}

We present two Isabelle proofs as examples. The first one shows a
proof of Lemma~\ref{L:LemmaInvSymProp0}\eqref{I:LemmaSymCompo} by automated proof search.
\begin{isabellebody}
  \isanewline
\isacommand{lemma}\isamarkupfalse%
\ sym{\isacharunderscore}{\kern0pt}func{\isadigit{1}}{\isacharcolon}{\kern0pt}\ \isanewline
\ \ \isakeyword{assumes}\ {\isachardoublequoteopen}fFx\ i\ x{\isachardoublequoteclose}\isanewline
\ \ \isakeyword{and}\ {\isachardoublequoteopen}fFx\ i\ y{\isachardoublequoteclose}\isanewline
\ \ \isakeyword{and}\ {\isachardoublequoteopen}DD\ i\ x\ y{\isachardoublequoteclose}\isanewline
\ \ \isakeyword{shows}\ {\isachardoublequoteopen}{\isasymsigma}\ i\ {\isacharparenleft}{\kern0pt}x\ {\isasymotimes}\isactrlbsub i\isactrlesub \ y{\isacharparenright}{\kern0pt}\ {\isacharequal}{\kern0pt}\ {\isasymsigma}\ i\ x\ {\isasymotimes}\isactrlbsub {\isacharparenleft}{\kern0pt}i\ {\isacharplus}{\kern0pt}\ {\isadigit{1}}{\isacharparenright}{\kern0pt}\isactrlesub \ {\isasymsigma}\ i\ y{\isachardoublequoteclose}\isanewline
\ \ \isacommand{by}\isamarkupfalse%
\ {\isacharparenleft}{\kern0pt}metis\ assms\ icid{\isachardot}{\kern0pt}ts{\isacharunderscore}{\kern0pt}compat\ local{\isachardot}{\kern0pt}iDst\ local{\isachardot}{\kern0pt}icat{\isachardot}{\kern0pt}sscatml{\isachardot}{\kern0pt}l{\isadigit{0}}{\isacharunderscore}{\kern0pt}absorb\ sym{\isacharunderscore}{\kern0pt}type{\isacharunderscore}{\kern0pt}var{\isadigit{1}}{\isacharparenright}{\kern0pt}%
\isanewline
\end{isabellebody}
\noindent Isabelle's Sledgehammer tactic has returned the proof
shown. Sledgehammer invoques external proof-search tools for
first-order logic, which are internally reconstructed by Isabelle's
\isa{metis} tool, which itself has been verified using Isabelle to
increase trustworthiness. The proof statement lists the lemmas
used. All of them are part of our Isabelle component for 
cubical $\omega$-categories and the components on which it builds.

The second proof shows how a proof by hand can be typed into Isabelle
line by line, and each line then be verified automatically using
Isabelle's proof tactics -- here the third case
in Lemma~\ref{L:LemmaInvSymProp}\eqref{I:LemmaFaceInvSym}.
\begin{isabellebody}
  \isanewline \isacommand{lemma}\isamarkupfalse%
  \
  inv{\isacharunderscore}{\kern0pt}sym{\isacharunderscore}{\kern0pt}face{\isacharcolon}{\kern0pt}\
  \isanewline \ \ \isakeyword{assumes}\ {\isachardoublequoteopen}i\
  {\isasymnoteq}\ j{\isachardoublequoteclose}\isanewline \ \
  \isakeyword{and}\ {\isachardoublequoteopen}i\ {\isasymnoteq}\ j\
  {\isacharplus}{\kern0pt}\
  {\isadigit{1}}{\isachardoublequoteclose}\isanewline \ \
  \isakeyword{and}\ {\isachardoublequoteopen}fFx\
  {\isacharparenleft}{\kern0pt}j\ {\isacharplus}{\kern0pt}\
  {\isadigit{1}}{\isacharparenright}{\kern0pt}\
  x{\isachardoublequoteclose}\isanewline \ \ \isakeyword{shows}\
  {\isachardoublequoteopen}{\isasympartial}\ i\ {\isasymalpha}\
  {\isacharparenleft}{\kern0pt}{\isasymtheta}\ j\
  x{\isacharparenright}{\kern0pt}\ {\isacharequal}{\kern0pt}\
  {\isasymtheta}\ j\ {\isacharparenleft}{\kern0pt}{\isasympartial}\ i\
  {\isasymalpha}\
  x{\isacharparenright}{\kern0pt}{\isachardoublequoteclose}\isanewline
  \isacommand{proof}\isamarkupfalse%
  {\isacharminus}{\kern0pt}\isanewline \ \
  \isacommand{have}\isamarkupfalse%
  \ {\isachardoublequoteopen}{\isasymsigma}\ j\
  {\isacharparenleft}{\kern0pt}{\isasympartial}\ i\ {\isasymalpha}\
  {\isacharparenleft}{\kern0pt}{\isasymtheta}\ j\
  x{\isacharparenright}{\kern0pt}{\isacharparenright}{\kern0pt}\
  {\isacharequal}{\kern0pt}\ {\isasymsigma}\ j\
  {\isacharparenleft}{\kern0pt}{\isasympartial}\ i\ {\isasymalpha}\
  {\isacharparenleft}{\kern0pt}{\isasympartial}\ j\ ff\
  {\isacharparenleft}{\kern0pt}{\isasymtheta}\ j\
  x{\isacharparenright}{\kern0pt}{\isacharparenright}{\kern0pt}{\isacharparenright}{\kern0pt}{\isachardoublequoteclose}\isanewline
  \ \ \ \ \isacommand{using}\isamarkupfalse%
  \
  assms{\isacharparenleft}{\kern0pt}{\isadigit{3}}{\isacharparenright}{\kern0pt}\
  inv{\isacharunderscore}{\kern0pt}sym{\isacharunderscore}{\kern0pt}type{\isacharunderscore}{\kern0pt}var\
  \isacommand{by}\isamarkupfalse%
  \ simp\isanewline \ \ \isacommand{also}\isamarkupfalse%
  \ \isacommand{have}\isamarkupfalse%
  \ {\isachardoublequoteopen}{\isasymdots}\ {\isacharequal}{\kern0pt}\
  {\isasympartial}\ i\ {\isasymalpha}\
  {\isacharparenleft}{\kern0pt}{\isasymsigma}\ j\
  {\isacharparenleft}{\kern0pt}{\isasympartial}\ j\ {\isasymalpha}\
  {\isacharparenleft}{\kern0pt}{\isasymtheta}\ j\
  x{\isacharparenright}{\kern0pt}{\isacharparenright}{\kern0pt}{\isacharparenright}{\kern0pt}{\isachardoublequoteclose}\isanewline
  \ \ \ \ \isacommand{by}\isamarkupfalse%
  \ {\isacharparenleft}{\kern0pt}metis\
  assms{\isacharparenleft}{\kern0pt}{\isadigit{1}}{\isacharparenright}{\kern0pt}\
  assms{\isacharparenleft}{\kern0pt}{\isadigit{2}}{\isacharparenright}{\kern0pt}\
  assms{\isacharparenleft}{\kern0pt}{\isadigit{3}}{\isacharparenright}{\kern0pt}\
  inv{\isacharunderscore}{\kern0pt}sym{\isacharunderscore}{\kern0pt}type{\isacharunderscore}{\kern0pt}var\
  local{\isachardot}{\kern0pt}fFx{\isacharunderscore}{\kern0pt}prop\
  sym{\isacharunderscore}{\kern0pt}face{\isacharunderscore}{\kern0pt}var{\isadigit{1}}{\isacharparenright}{\kern0pt}\isanewline
  \ \ \isacommand{also}\isamarkupfalse%
  \ \isacommand{have}\isamarkupfalse%
  \ {\isachardoublequoteopen}{\isasymdots}\ {\isacharequal}{\kern0pt}\
  {\isasympartial}\ i\ {\isasymalpha}\
  {\isacharparenleft}{\kern0pt}{\isasymsigma}\ j\
  {\isacharparenleft}{\kern0pt}{\isasymtheta}\ j\
  x{\isacharparenright}{\kern0pt}{\isacharparenright}{\kern0pt}{\isachardoublequoteclose}\isanewline
  \ \ \ \ \isacommand{using}\isamarkupfalse%
  \
  assms{\isacharparenleft}{\kern0pt}{\isadigit{1}}{\isacharparenright}{\kern0pt}\
  assms{\isacharparenleft}{\kern0pt}{\isadigit{2}}{\isacharparenright}{\kern0pt}\
  assms{\isacharparenleft}{\kern0pt}{\isadigit{3}}{\isacharparenright}{\kern0pt}\
  calculation\
  inv{\isacharunderscore}{\kern0pt}sym{\isacharunderscore}{\kern0pt}type{\isacharunderscore}{\kern0pt}var\
  local{\isachardot}{\kern0pt}sym{\isacharunderscore}{\kern0pt}face\
  \isacommand{by}\isamarkupfalse%
  \ presburger\isanewline \ \ \isacommand{also}\isamarkupfalse%
  \ \isacommand{have}\isamarkupfalse%
  \ {\isachardoublequoteopen}{\isasymdots}\ {\isacharequal}{\kern0pt}\
  {\isasympartial}\ i\ {\isasymalpha}\
  {\isacharparenleft}{\kern0pt}{\isasympartial}\
  {\isacharparenleft}{\kern0pt}j\ {\isacharplus}{\kern0pt}\
  {\isadigit{1}}{\isacharparenright}{\kern0pt}\ {\isasymalpha}\
  x{\isacharparenright}{\kern0pt}{\isachardoublequoteclose}\isanewline
  \ \ \ \ \isacommand{by}\isamarkupfalse%
  \ {\isacharparenleft}{\kern0pt}metis\
  assms{\isacharparenleft}{\kern0pt}{\isadigit{3}}{\isacharparenright}{\kern0pt}\
  local{\isachardot}{\kern0pt}face{\isacharunderscore}{\kern0pt}compat{\isacharunderscore}{\kern0pt}var\
  sym{\isacharunderscore}{\kern0pt}inv{\isacharunderscore}{\kern0pt}sym{\isacharunderscore}{\kern0pt}var{\isadigit{1}}{\isacharparenright}{\kern0pt}\isanewline
  \ \ \isacommand{finally}\isamarkupfalse%
  \ \isacommand{have}\isamarkupfalse%
  \ {\isachardoublequoteopen}{\isasymsigma}\ j\
  {\isacharparenleft}{\kern0pt}{\isasympartial}\ i\ {\isasymalpha}\
  {\isacharparenleft}{\kern0pt}{\isasymtheta}\ j\
  x{\isacharparenright}{\kern0pt}{\isacharparenright}{\kern0pt}\
  {\isacharequal}{\kern0pt}\ {\isasympartial}\ i\ {\isasymalpha}\
  {\isacharparenleft}{\kern0pt}{\isasympartial}\
  {\isacharparenleft}{\kern0pt}j\ {\isacharplus}{\kern0pt}\
  {\isadigit{1}}{\isacharparenright}{\kern0pt}\ {\isasymalpha}\
  x{\isacharparenright}{\kern0pt}{\isachardoublequoteclose}\isacommand{{\isachardot}{\kern0pt}}\isamarkupfalse%
  \isanewline \ \ \isacommand{thus}\isamarkupfalse%
  \ {\isacharquery}{\kern0pt}thesis\isanewline \ \ \ \
  \isacommand{by}\isamarkupfalse%
  \ {\isacharparenleft}{\kern0pt}smt\
  {\isacharparenleft}{\kern0pt}z{\isadigit{3}}{\isacharparenright}{\kern0pt}\
  assms{\isacharparenleft}{\kern0pt}{\isadigit{3}}{\isacharparenright}{\kern0pt}\
  icid{\isachardot}{\kern0pt}st{\isacharunderscore}{\kern0pt}eq{\isadigit{1}}\
  inv{\isacharunderscore}{\kern0pt}sym{\isacharunderscore}{\kern0pt}type{\isacharunderscore}{\kern0pt}var\
  local{\isachardot}{\kern0pt}face{\isacharunderscore}{\kern0pt}comm{\isacharunderscore}{\kern0pt}var\
  local{\isachardot}{\kern0pt}inv{\isacharunderscore}{\kern0pt}sym{\isacharunderscore}{\kern0pt}sym{\isacharparenright}{\kern0pt}\isanewline
  \isacommand{qed}\isamarkupfalse%
  \isanewline
\end{isabellebody}

Beyond fully automatic proofs found using Sledgehammer and interactive
proofs using Isabelle's proof scripting language Isar, as here,
Isabelle also offers so-called apply-style proofs, in which
simplification steps, application of rules or substitutions of
particular formulas can be combined step-wise with Sledgehammer
proofs. Examples can be found in our Isabelle component.

\subsection{Taming $\omega$-categories with Isabelle}

We have already outlined in the introduction how Isabelle has
  helped developing the single-set axioms for
  $\omega$-categories. Here we provide more details. Recall that we
  have justified these axioms via the equivalence proofs in
  Section~\ref{S:EquivalenceWithCubicalCategories}. Their selection
  was driven by the construction of the functors $\FC{(-)}$ and
  $\FS{(-)}$ and their extensions, which relate classical and
  single-set concepts. We aimed for a small set of structurally
  meaningful axioms to make the equivalence proofs smooth and simple.

We started with translating the axioms for
  $\partial^\alpha_{n,i}:\Ccal_n\to \Ccal_{n-1}$ into those for
  $\delta^\alpha_i:\Scal\to \Scal$. This was straightforward, for
  instance, for \eqref{I:AxiomCommutativity},
  \eqref{I:DeltaCompoCompat} or \eqref{I:ExchangeLaw} in~\eqref{SSS:DefCubSingleSetWithoutConn}, but others required
  encoding the index-shift in $\Ccal_n\to \Ccal_{n-1}$ of face maps in
  the single-set axioms, where graded sets $\Ccal_n$ are not
  immediately available. Instead we used the fixed point sets
  $\Scal^i$ or the face maps $\delta^\alpha_i$ as guards in the
  single-set axioms; Isabelle helped us to bring them into convenient
  form.  While this obviated degeneracies, we had to introduce
  symmetries to relate fixed points at the same dimension but in
  different directions, and to model the rotations of degenerated
  cubes through the interactions of symmetries with face maps and
  compositions. Starting from lattices like that in
  Subsection~\ref{SSS:Lattice}, the translations between symmetries
  and degeneracies in $\FC{(-)}$ and $\FS{(-)}$ and geometric
  intuitions, we used Isabelle, in particular Sledgehammer in
  combination with other proof tactics, in an iterative process to
  adapt or simplify candidate axioms, to analyse their dependencies,
  and to add axioms in light of the equivalence proof. Beyond
  symmetries, the equivalence proof for $\omega$-categories led us to
  experiment with dimensionality axioms using Isabelle. This resulted
  in Axiom~\eqref{I:AxiomFiniteDimCells}, and enabled the colimit and
  filtration constructions for $\FS{(-)}$.

During this process, we compressed the single-set
  axiomatisation for $\omega$-categories by a factor ${>}2$ to a size
  similar to the classical one.  A significant part of the process was
  automatic. Most redundant candidate axioms now feature in
  Lemmas~\ref{L:LemmaInvSymProp0} and \ref{L:LemmaInvSymProp}. Our
  work flow for $\omega$-categories with connections has been similar
  and resulted in a similar compression. Redundant laws are shown in
  Lemma~\ref{L:LemmaInvConnProp0}. Interestingly, we found
  Axiom~\eqref{I:AxiomConnShift} quite late through the equivalence
  proof.

Our insights in the strengths and weaknesses of Isabelle's
  proof automation might be valuable for mathematicians working with
  higher categories, where proofs tend to be highly combinatorial,
  axiomatisations often fill pages and there can be a big
  formalisation gap with respect to geometric or (string)
  diagrammatical intuition. In our work, we were sometimes surprised
  when Sledgehammer managed to derive seemingly independent
  conjectures, such as the Yang-Baxter identity in
  Lemma~\ref{L:LemmaInvSymProp0} or the face identities in
  Lemma~\ref{L:LemmaInvSymProp}(1). But we also spent hours feeding
  paper-and-pencil proofs into Isabelle and hard-coding rule
  applications, including the proof in the following
  subsection. Overall, interactive proofs with higher categories at
  the granularity of paper and pencil proofs seem nowadays feasible --
  and highly beneficial for activities like the one described in this
  article. Yet a main source of disappointment was that, unlike in
  previous work, we could not use Isabelle's Nitpick tool for
  verifying the irredundancy of our axiomatisation: it seems that the
  underlying SAT solver cannot cope with the arithmetic constraints in
  our axioms, though that should certainly be possible for SMT
  solvers.

\subsection{A non-trivial proof}
\label{SS:isaatwork}

At the end of this section, we show our formalisation at work,
presenting a proof of Proposition~\ref{P:OmegaZeroInv}. This
  example shows that Isabelle's proof automation smoothly supports
  interactive proofs in higher category theory. For this we have
formalised $(\omega,0)$-categories with Isabelle. 
A formalisation of $(\omega,p)$-categories based on type
classes seems impossible as it would require more than one type
parameter.

Defining a type class for $(\omega,0)$-categories needs some
preliminary definitions. First we have defined compositions of
sequences of symmetries and reverse symmetries.
\begin{isabellebody}
  \isanewline
  \isacommand{primrec}\isamarkupfalse%
\ symcomp\ {\isacharcolon}{\kern0pt}{\isacharcolon}{\kern0pt}\ {\isachardoublequoteopen}nat\ {\isasymRightarrow}\ nat\ {\isasymRightarrow}\ {\isacharprime}{\kern0pt}a\ {\isasymRightarrow}\ {\isacharprime}{\kern0pt}a{\isachardoublequoteclose}\ {\isacharparenleft}{\kern0pt}{\isachardoublequoteopen}{\isasymSigma}{\isachardoublequoteclose}{\isacharparenright}{\kern0pt}\ \isakeyword{where}\isanewline
\ \ \ \ {\isachardoublequoteopen}{\isasymSigma}\ i\ {\isadigit{0}}\ x\ {\isacharequal}{\kern0pt}\ x{\isachardoublequoteclose}\isanewline
\ \ {\isacharbar}{\kern0pt}\ {\isachardoublequoteopen}{\isasymSigma}\ i\ {\isacharparenleft}{\kern0pt}Suc\ j{\isacharparenright}{\kern0pt}\ x\ {\isacharequal}{\kern0pt}\ {\isasymsigma}\ {\isacharparenleft}{\kern0pt}i\ {\isacharplus}{\kern0pt}\ j{\isacharparenright}{\kern0pt}\ {\isacharparenleft}{\kern0pt}{\isasymSigma}\ i\ j\ x{\isacharparenright}{\kern0pt}{\isachardoublequoteclose}%
\isanewline
\isanewline
\isacommand{primrec}\isamarkupfalse%
\ inv{\isacharunderscore}{\kern0pt}symcomp\ {\isacharcolon}{\kern0pt}{\isacharcolon}{\kern0pt}\ {\isachardoublequoteopen}nat\ {\isasymRightarrow}\ nat\ {\isasymRightarrow}\ {\isacharprime}{\kern0pt}a\ {\isasymRightarrow}\ {\isacharprime}{\kern0pt}a{\isachardoublequoteclose}\ {\isacharparenleft}{\kern0pt}{\isachardoublequoteopen}{\isasymTheta}{\isachardoublequoteclose}{\isacharparenright}{\kern0pt}\ \isakeyword{where}\isanewline
\ \ \ \ {\isachardoublequoteopen}{\isasymTheta}\ i\ {\isadigit{0}}\ x\ {\isacharequal}{\kern0pt}\ x{\isachardoublequoteclose}\isanewline
\ \ {\isacharbar}{\kern0pt}\ {\isachardoublequoteopen}{\isasymTheta}\ i\ {\isacharparenleft}{\kern0pt}Suc\ j{\isacharparenright}{\kern0pt}\ x\ {\isacharequal}{\kern0pt}\ {\isasymTheta}\ i\ j\ {\isacharparenleft}{\kern0pt}{\isasymtheta}\ {\isacharparenleft}{\kern0pt}i\ {\isacharplus}{\kern0pt}\ j{\isacharparenright}{\kern0pt}\ x{\isacharparenright}{\kern0pt}{\isachardoublequoteclose}\isanewline
\end{isabellebody}

\noindent Then we have defined $r_i$-invertibility and shell
$r_i$-invertibility, following Definition~\ref{SSS:DefInverse}.
\begin{isabellebody}
\isanewline
\isacommand{abbreviation}\isamarkupfalse%
\ {\isacharparenleft}{\kern0pt}\isakeyword{in}\
cubical{\isacharunderscore}{\kern0pt}omega{\isacharunderscore}{\kern0pt}category{\isacharunderscore}{\kern0pt}connections{\isacharparenright}{\kern0pt}\isanewline
\ \ \ \ {\isachardoublequoteopen}ri{\isacharunderscore}{\kern0pt}inv\ i\ x\ y\ {\isasymequiv}\ {\isacharparenleft}{\kern0pt}DD\ i\ x\ y\ {\isasymand}\ DD\ i\ y\ x\ {\isasymand}\ x\ {\isasymotimes}\isactrlbsub i\isactrlesub \ y\ {\isacharequal}{\kern0pt}\ {\isasympartial}\ i\ ff\ x\ {\isasymand}\ y\ {\isasymotimes}\isactrlbsub i\isactrlesub \ x\ {\isacharequal}{\kern0pt}\ {\isasympartial}\ i\ tt\ x{\isacharparenright}{\kern0pt}{\isachardoublequoteclose}\isanewline
\isanewline
\isacommand{abbreviation}\isamarkupfalse%
\ {\isacharparenleft}{\kern0pt}\isakeyword{in}\
cubical{\isacharunderscore}{\kern0pt}omega{\isacharunderscore}{\kern0pt}category{\isacharunderscore}{\kern0pt}connections{\isacharparenright}{\kern0pt}\isanewline
\ \ \ \ {\isachardoublequoteopen}ri{\isacharunderscore}{\kern0pt}inv{\isacharunderscore}{\kern0pt}shell\ k\ i\ x\ {\isasymequiv}\ {\isacharparenleft}{\kern0pt}{\isasymforall}j\ {\isasymalpha}{\isachardot}{\kern0pt}\ j\ {\isacharplus}{\kern0pt}\ {\isadigit{1}}\ {\isasymle}\ k\ {\isasymand}\ j\ {\isasymnoteq}\ i\ {\isasymlongrightarrow}\ {\isacharparenleft}{\kern0pt}{\isasymexists}y{\isachardot}{\kern0pt}\ ri{\isacharunderscore}{\kern0pt}inv\ i\ {\isacharparenleft}{\kern0pt}{\isasympartial}\ j\ {\isasymalpha}\ x{\isacharparenright}{\kern0pt}\ y{\isacharparenright}{\kern0pt}{\isacharparenright}{\kern0pt}{\isachardoublequoteclose}%
\isanewline
\end{isabellebody}

\noindent This allowed us to specify a class for
$(\omega,0)$-categories following Definition~\ref{SSS:DefInverse}.
\begin{isabellebody}
\isanewline
\isacommand{class}\isamarkupfalse%
\ cubical{\isacharunderscore}{\kern0pt}omega{\isacharunderscore}{\kern0pt}zero{\isacharunderscore}{\kern0pt}category{\isacharunderscore}{\kern0pt}connections\ {\isacharequal}{\kern0pt}\ cubical{\isacharunderscore}{\kern0pt}omega{\isacharunderscore}{\kern0pt}category{\isacharunderscore}{\kern0pt}connections\ {\isacharplus}{\kern0pt}\isanewline
\ \ \isakeyword{assumes}\ ri{\isacharunderscore}{\kern0pt}inv{\isacharcolon}{\kern0pt}\ {\isachardoublequoteopen}k\ {\isasymge}\ {\isadigit{1}}\ {\isasymLongrightarrow}\ i\ {\isasymle}\ k\ {\isacharminus}{\kern0pt}\ {\isadigit{1}}\ {\isasymLongrightarrow}\ dim{\isacharunderscore}{\kern0pt}bound\ k\ x\ {\isasymLongrightarrow}\ ri{\isacharunderscore}{\kern0pt}inv{\isacharunderscore}{\kern0pt}shell\ k\ i\ x\ {\isasymLongrightarrow}\ {\isasymexists}y{\isachardot}{\kern0pt}\ ri{\isacharunderscore}{\kern0pt}inv\ i\ x\ y{\isachardoublequoteclose}\isanewline
\end{isabellebody}

\noindent
In the axiom \isa{ri-inv}, the predicate
\isa{{\isachardoublequoteopen}dim{\isacharunderscore}{\kern0pt}bound\
  k\ x\ {\isasymequiv}\
  {\isacharparenleft}{\kern0pt}{\isasymforall}i{\isachardot}{\kern0pt}\
  k\ {\isasymle}\ i\ {\isasymlongrightarrow}\ fFx\ i\
  x{\isacharparenright}{\kern0pt}{\isachardoublequoteclose}}, which we
use in place of the set $\Scal^{>k}$ from
Section~\ref{SS:SingleSetCubicalCategories}, has been defined in the
context of the class \isa{icategory}.

We have shown uniqueness of $r_i$-inverses and used this property,
together with Isabelle's definite description operator \isa{THE}, to
define an inversion map.
\begin{isabellebody}
  \isanewline
\isacommand{lemma}\isamarkupfalse%
\ ri{\isacharunderscore}{\kern0pt}unique{\isacharcolon}{\kern0pt}\
{\isachardoublequoteopen}{\isacharparenleft}{\kern0pt}{\isasymexists}y{\isachardot}{\kern0pt}\
ri{\isacharunderscore}{\kern0pt}inv\ i\ x\
y{\isacharparenright}{\kern0pt}\ {\isacharequal}{\kern0pt}\
{\isacharparenleft}{\kern0pt}{\isasymexists}{\isacharbang}{\kern0pt}y{\isachardot}{\kern0pt}\
ri{\isacharunderscore}{\kern0pt}inv\ i\ x\
y{\isacharparenright}{\kern0pt}{\isachardoublequoteclose}\isanewline
\ \   $\langle$\,proof\,$\rangle$\isanewline
\isanewline
\isacommand{definition}\isamarkupfalse%
\ {\isachardoublequoteopen}ri\ i\ x\ {\isacharequal}{\kern0pt}\ {\isacharparenleft}{\kern0pt}THE\ y{\isachardot}{\kern0pt}\ ri{\isacharunderscore}{\kern0pt}inv\ i\ x\ y{\isacharparenright}{\kern0pt}{\isachardoublequoteclose}\isanewline
\end{isabellebody}

Our proof of Proposition~\ref{P:OmegaZeroInv} with Isabelle follows
that in Subsection~\ref{SS:SinSetWithInv} quite directly. We do not
show the verifications of individual proof steps by Isabelle's proof
tools. Details can be found in our Isabelle
component~\cite{MassacrierStruth24}. The main part of the proof is
captured by a technical lemma that proceeds by induction on the
dimension $k$ of the cell $x$.

\begin{isabellebody}
\isanewline
\isacommand{lemma}\isamarkupfalse%
\ every{\isacharunderscore}{\kern0pt}dim{\isacharunderscore}{\kern0pt}k{\isacharunderscore}{\kern0pt}ri{\isacharunderscore}{\kern0pt}inv{\isacharcolon}{\kern0pt}\isanewline
\ \ \isacommand{assumes}\ {\isachardoublequoteopen}dim{\isacharunderscore}{\kern0pt}bound\ k\ x{\isachardoublequoteclose}\isanewline
\ \ \isacommand{shows}\ {\isachardoublequoteopen}{\isasymforall}i{\isachardot}{\kern0pt}\ {\isasymexists}y{\isachardot}{\kern0pt}\ ri{\isacharunderscore}{\kern0pt}inv\ i\ x\ y{\kern0pt}{\isachardoublequoteclose}
\ \isacommand{using}\ {\isachardoublequoteopen}dim{\isacharunderscore}{\kern0pt}bound\ k\ x{\isachardoublequoteclose}\isanewline
\isacommand{proof}\isamarkupfalse%
\ {\isacharparenleft}{\kern0pt}induction\ k\ arbitrary{\isacharcolon}{\kern0pt}\ x{\isacharparenright}{\kern0pt}\isanewline
\ \ \isacommand{case}\isamarkupfalse%
\ {\isadigit{0}}\isanewline
\ \ \isacommand{thus}\isamarkupfalse%
\ {\isacharquery}{\kern0pt}case\ \isanewline
\ \ \ \ $\langle$\,proof\,$\rangle$\isanewline
\isacommand{next}\isamarkupfalse%
\isanewline
\ \ \isacommand{case}\isamarkupfalse%
\ {\isacharparenleft}{\kern0pt}Suc\ k{\isacharparenright}{\kern0pt}\isanewline
\ \ \isacommand{{\isacharbraceleft}{\kern0pt}}\isamarkupfalse%
\isacommand{fix}\isamarkupfalse%
\ i\isanewline
\ \ \ \ \isacommand{have}\isamarkupfalse%
\ {\isachardoublequoteopen}{\isasymexists}y{\isachardot}{\kern0pt}\ ri{\isacharunderscore}{\kern0pt}inv\ i\ x\ y{\isachardoublequoteclose}
\isanewline
\end{isabellebody}

\noindent
Here we start a proof by cases for $i \geq k+1$ as in
Section~\ref{SS:SinSetWithInv}. As in the proof by hand, the first case
is trivial, and automatic with Isabelle.
\begin{isabellebody}
\isanewline
\ \ \ \ \isacommand{proof}\isamarkupfalse%
\ {\isacharparenleft}{\kern0pt}cases\ {\isachardoublequoteopen}Suc\ k\ {\isasymge}\ i{\isachardoublequoteclose}{\isacharparenright}{\kern0pt}\isanewline
\ \ \ \ \ \ \isacommand{case}\isamarkupfalse%
\ True\isanewline
\ \ \ \ \ \ \isacommand{thus}\isamarkupfalse%
\ {\isacharquery}{\kern0pt}thesis\isanewline
\ \ \ \ \ \ \ \ $\langle$\,proof\,$\rangle$\isanewline
\ \ \ \ \isacommand{next}\isamarkupfalse%
\isanewline
\ \ \ \ \ \ \isacommand{case}\isamarkupfalse%
\ False\isanewline
\ \ \ \ \ \ \isacommand{{\isacharbraceleft}{\kern0pt}}\isamarkupfalse%
\isacommand{fix}\isamarkupfalse%
\ j\ {\isasymalpha}\isanewline
\ \ \ \ \ \ \ \ \isacommand{assume}\isamarkupfalse%
\ h{\isacharcolon}{\kern0pt}\ {\isachardoublequoteopen}j\ {\isasymle}\ k\ {\isasymand}\ j\ {\isasymnoteq}\ i{\isachardoublequoteclose}
\isanewline
\end{isabellebody}

\noindent
While the proof of $\delta_j^\alpha x \in \Scal^{j,k+1,k+2,\dots}$
with Isabelle is  automatic, we need to check
$s_{k-1} \dots s_j \delta_j^\alpha x \in \Scal^{>k}$.
\begin{isabellebody}
\isanewline
\ \ \ \ \ \ \ \ \isacommand{hence}\isamarkupfalse%
\ a{\isacharcolon}{\kern0pt}\ {\isachardoublequoteopen}has{\isacharunderscore}{\kern0pt}dim{\isacharunderscore}{\kern0pt}bound\ k\ {\isacharparenleft}{\kern0pt}{\isasymSigma}\ j\ {\isacharparenleft}{\kern0pt}k\ {\isacharminus}{\kern0pt}\ j{\isacharparenright}{\kern0pt}\ {\isacharparenleft}{\kern0pt}{\isasympartial}\ j\ {\isasymalpha}\ x{\isacharparenright}{\kern0pt}{\isacharparenright}{\kern0pt}{\isachardoublequoteclose}\isanewline
\ \ \ \ \ \ \ \ \ \ $\langle$\,proof\,$\rangle$\isanewline
\ \ \ \ \ \ \ \ \isacommand{have}\isamarkupfalse%
\ {\isachardoublequoteopen}{\isasymexists}y{\isachardot}{\kern0pt}\ ri{\isacharunderscore}{\kern0pt}inv\ i\ {\isacharparenleft}{\kern0pt}{\isasympartial}\ j\ {\isasymalpha}\ x{\isacharparenright}{\kern0pt}\ y{\isachardoublequoteclose}
\isanewline
\end{isabellebody}

\noindent
To construct an $r_i$-inverse of $\delta_j^\alpha x$, we perform a proof by cases on $j$.
\begin{isabellebody}
\isanewline
\ \ \ \ \ \ \ \ \isacommand{proof}\isamarkupfalse%
\ {\isacharparenleft}{\kern0pt}cases\ {\isachardoublequoteopen}j\ {\isacharless}{\kern0pt}\ i{\isachardoublequoteclose}{\isacharparenright}{\kern0pt}\isanewline
\ \ \ \ \ \ \ \ \ \ \isacommand{case}\isamarkupfalse%
\ True
\isanewline
\end{isabellebody}

\noindent
For $j < i$, we introduce $y$ as the $r_{i-1}$-inverse of
$s_{k-1} \dots s_j \delta_j^\alpha x$ using the
induction hypothesis.
\begin{isabellebody}
\isanewline
\ \ \ \ \ \ \ \ \ \ \isacommand{obtain}\isamarkupfalse%
\ y\ \isakeyword{where}\ b{\isacharcolon}{\kern0pt}\ {\isachardoublequoteopen}ri{\isacharunderscore}{\kern0pt}inv\ {\isacharparenleft}{\kern0pt}i\ {\isacharminus}{\kern0pt}\ {\isadigit{1}}{\isacharparenright}{\kern0pt}\ {\isacharparenleft}{\kern0pt}{\isasymSigma}\ j\ {\isacharparenleft}{\kern0pt}k\ {\isacharminus}{\kern0pt}\ j{\isacharparenright}{\kern0pt}\ {\isacharparenleft}{\kern0pt}{\isasympartial}\ j\ {\isasymalpha}\ x{\isacharparenright}{\kern0pt}{\isacharparenright}{\kern0pt}\ y{\isachardoublequoteclose}\isanewline
\ \ \ \ \ \ \ \ \ \ \ \ $\langle$\,proof\,$\rangle$
\isanewline
\end{isabellebody}

\noindent
We check that $\tilde{s}_j \dots \tilde{s}_{k-1} y$ and
$\delta_j^\alpha x$ are composable and then show that these
expressions are inverses.
\begin{isabellebody}
\isanewline
\ \ \ \ \ \ \ \ \ \ \isacommand{have}\isamarkupfalse%
\ c{\isacharcolon}{\kern0pt}\ {\isachardoublequoteopen}dim{\isacharunderscore}{\kern0pt}bound\ k\ y{\isachardoublequoteclose}\isanewline
\ \ \ \ \ \ \ \ \ \ \ \ $\langle$\,proof\,$\rangle$\isanewline
\ \ \ \ \ \ \ \ \ \ \isacommand{hence}\isamarkupfalse%
\ d{\isacharcolon}{\kern0pt}\ {\isachardoublequoteopen}DD\ i\ {\isacharparenleft}{\kern0pt}{\isasympartial}\ j\ {\isasymalpha}\ x{\isacharparenright}{\kern0pt}\ {\isacharparenleft}{\kern0pt}{\isasymTheta}\ j\ {\isacharparenleft}{\kern0pt}k\ {\isacharminus}{\kern0pt}\ j{\isacharparenright}{\kern0pt}\ y{\isacharparenright}{\kern0pt}{\isachardoublequoteclose}\isanewline
\ \ \ \ \ \ \ \ \ \ \ \ $\langle$\,proof\,$\rangle$\isanewline
\ \ \ \ \ \ \ \ \ \ \isacommand{hence}\isamarkupfalse%
\ e{\isacharcolon}{\kern0pt}\ {\isachardoublequoteopen}DD\ i\ {\isacharparenleft}{\kern0pt}{\isasymTheta}\ j\ {\isacharparenleft}{\kern0pt}k\ {\isacharminus}{\kern0pt}\ j{\isacharparenright}{\kern0pt}\ y{\isacharparenright}{\kern0pt}\ {\isacharparenleft}{\kern0pt}{\isasympartial}\ j\ {\isasymalpha}\ x{\isacharparenright}{\kern0pt}{\isachardoublequoteclose}\isanewline
\ \ \ \ \ \ \ \ \ \ \ \ $\langle$\,proof\,$\rangle$\isanewline
\ \ \ \ \ \ \ \ \ \ \isacommand{have}\isamarkupfalse%
\ f{\isacharcolon}{\kern0pt}\ {\isachardoublequoteopen}{\isacharparenleft}{\kern0pt}{\isasympartial}\ j\ {\isasymalpha}\ x{\isacharparenright}{\kern0pt}\ {\isasymotimes}\isactrlbsub i\isactrlesub \ {\isacharparenleft}{\kern0pt}{\isasymTheta}\ j\ {\isacharparenleft}{\kern0pt}k\ {\isacharminus}{\kern0pt}\ j{\isacharparenright}{\kern0pt}\ y{\isacharparenright}{\kern0pt}\ {\isacharequal}{\kern0pt}\ {\isasymTheta}\ j\ {\isacharparenleft}{\kern0pt}k\ {\isacharminus}{\kern0pt}\ j{\isacharparenright}{\kern0pt}\ {\isacharparenleft}{\kern0pt}{\isacharparenleft}{\kern0pt}{\isasymSigma}\ j\ {\isacharparenleft}{\kern0pt}k\ {\isacharminus}{\kern0pt}\ j{\isacharparenright}{\kern0pt}\ {\isacharparenleft}{\kern0pt}{\isasympartial}\ j\ {\isasymalpha}\ x{\isacharparenright}{\kern0pt}{\isacharparenright}{\kern0pt}\ {\isasymotimes}\isactrlbsub {\isacharparenleft}{\kern0pt}i\ {\isacharminus}{\kern0pt}\ {\isadigit{1}}{\isacharparenright}{\kern0pt}\isactrlesub \ y{\isacharparenright}{\kern0pt}{\isachardoublequoteclose}\isanewline
\ \ \ \ \ \ \ \ \ \ \ \ $\langle$\,proof\,$\rangle$\isanewline
\ \ \ \ \ \ \ \ \ \ \isacommand{have}\isamarkupfalse%
\ {\isachardoublequoteopen}{\isacharparenleft}{\kern0pt}{\isasymTheta}\ j\ {\isacharparenleft}{\kern0pt}k\ {\isacharminus}{\kern0pt}\ j{\isacharparenright}{\kern0pt}\ y{\isacharparenright}{\kern0pt}\ {\isasymotimes}\isactrlbsub i\isactrlesub \ {\isacharparenleft}{\kern0pt}{\isasympartial}\ j\ {\isasymalpha}\ x{\isacharparenright}{\kern0pt}\ {\isacharequal}{\kern0pt}\ {\isasymTheta}\ j\ {\isacharparenleft}{\kern0pt}k\ {\isacharminus}{\kern0pt}\ j{\isacharparenright}{\kern0pt}\ {\isacharparenleft}{\kern0pt}y\ {\isasymotimes}\isactrlbsub {\isacharparenleft}{\kern0pt}i\ {\isacharminus}{\kern0pt}\ {\isadigit{1}}{\isacharparenright}{\kern0pt}\isactrlesub \ {\isacharparenleft}{\kern0pt}{\isasymSigma}\ j\ {\isacharparenleft}{\kern0pt}k\ {\isacharminus}{\kern0pt}\ j{\isacharparenright}{\kern0pt}\ {\isacharparenleft}{\kern0pt}{\isasympartial}\ j\ {\isasymalpha}\ x{\isacharparenright}{\kern0pt}{\isacharparenright}{\kern0pt}{\isacharparenright}{\kern0pt}{\isachardoublequoteclose}\isanewline
\ \ \ \ \ \ \ \ \ \ \ \ $\langle$\,proof\,$\rangle$\isanewline
\ \ \ \ \ \ \ \ \ \ \isacommand{thus}\isamarkupfalse%
\ {\isacharquery}{\kern0pt}thesis\isanewline
\ \ \ \ \ \ \ \ \ \ \ \ $\langle$\,proof\,$\rangle$
\isanewline
\end{isabellebody}

\noindent
We proceed similarly in the case of $j > i$.
\begin{isabellebody}
\isanewline
\ \ \ \ \ \ \ \ \isacommand{next}\isamarkupfalse%
\isanewline
\ \ \ \ \ \ \ \ \ \ \isacommand{case}\isamarkupfalse%
\ False\isanewline
\ \ \ \ \ \ \ \ \ \ \isacommand{obtain}\isamarkupfalse%
\ y\ \isakeyword{where}\ b{\isacharcolon}{\kern0pt}\ {\isachardoublequoteopen}ri{\isacharunderscore}{\kern0pt}inv\ i\ {\isacharparenleft}{\kern0pt}{\isasymSigma}\ j\ {\isacharparenleft}{\kern0pt}k\ {\isacharminus}{\kern0pt}\ j{\isacharparenright}{\kern0pt}\ {\isacharparenleft}{\kern0pt}{\isasympartial}\ j\ {\isasymalpha}\ x{\isacharparenright}{\kern0pt}{\isacharparenright}{\kern0pt}\ y{\isachardoublequoteclose}\isanewline
\ \ \ \ \ \ \ \ \ \ \ \ $\langle$\,proof\,$\rangle$\isanewline
\ \ \ \ \ \ \ \ \ \ \isacommand{have}\isamarkupfalse%
\ c{\isacharcolon}{\kern0pt}\ {\isachardoublequoteopen}dim{\isacharunderscore}{\kern0pt}bound\ k\ y{\isachardoublequoteclose}\isanewline
\ \ \ \ \ \ \ \ \ \ \ \ $\langle$\,proof\,$\rangle$\isanewline
\ \ \ \ \ \ \ \ \ \ \isacommand{hence}\isamarkupfalse%
\ d{\isacharcolon}{\kern0pt}\ {\isachardoublequoteopen}DD\ i\ {\isacharparenleft}{\kern0pt}{\isasympartial}\ j\ {\isasymalpha}\ x{\isacharparenright}{\kern0pt}\ {\isacharparenleft}{\kern0pt}{\isasymTheta}\ j\ {\isacharparenleft}{\kern0pt}k\ {\isacharminus}{\kern0pt}\ j{\isacharparenright}{\kern0pt}\ y{\isacharparenright}{\kern0pt}{\isachardoublequoteclose}\isanewline
\ \ \ \ \ \ \ \ \ \ \ \ $\langle$\,proof\,$\rangle$\isanewline
\ \ \ \ \ \ \ \ \ \ \isacommand{hence}\isamarkupfalse%
\ e{\isacharcolon}{\kern0pt}\ {\isachardoublequoteopen}DD\ i\ {\isacharparenleft}{\kern0pt}{\isasymTheta}\ j\ {\isacharparenleft}{\kern0pt}k\ {\isacharminus}{\kern0pt}\ j{\isacharparenright}{\kern0pt}\ y{\isacharparenright}{\kern0pt}\ {\isacharparenleft}{\kern0pt}{\isasympartial}\ j\ {\isasymalpha}\ x{\isacharparenright}{\kern0pt}{\isachardoublequoteclose}\isanewline
\ \ \ \ \ \ \ \ \ \ \ \ $\langle$\,proof\,$\rangle$\isanewline
\ \ \ \ \ \ \ \ \ \ \isacommand{have}\isamarkupfalse%
\ f{\isacharcolon}{\kern0pt}\ {\isachardoublequoteopen}{\isacharparenleft}{\kern0pt}{\isasympartial}\ j\ {\isasymalpha}\ x{\isacharparenright}{\kern0pt}\ {\isasymotimes}\isactrlbsub i\isactrlesub \ {\isacharparenleft}{\kern0pt}{\isasymTheta}\ j\ {\isacharparenleft}{\kern0pt}k\ {\isacharminus}{\kern0pt}\ j{\isacharparenright}{\kern0pt}\ y{\isacharparenright}{\kern0pt}\ {\isacharequal}{\kern0pt}\ {\isasymTheta}\ j\ {\isacharparenleft}{\kern0pt}k\ {\isacharminus}{\kern0pt}\ j{\isacharparenright}{\kern0pt}\ {\isacharparenleft}{\kern0pt}{\isacharparenleft}{\kern0pt}{\isasymSigma}\ j\ {\isacharparenleft}{\kern0pt}k\ {\isacharminus}{\kern0pt}\ j{\isacharparenright}{\kern0pt}\ {\isacharparenleft}{\kern0pt}{\isasympartial}\ j\ {\isasymalpha}\ x{\isacharparenright}{\kern0pt}{\isacharparenright}{\kern0pt}\ {\isasymotimes}\isactrlbsub i\isactrlesub \ y{\isacharparenright}{\kern0pt}{\isachardoublequoteclose}\isanewline
\ \ \ \ \ \ \ \ \ \ \ \ $\langle$\,proof\,$\rangle$\isanewline
\ \ \ \ \ \ \ \ \ \ \isacommand{have}\isamarkupfalse%
\ {\isachardoublequoteopen}{\isacharparenleft}{\kern0pt}{\isasymTheta}\ j\ {\isacharparenleft}{\kern0pt}k\ {\isacharminus}{\kern0pt}\ j{\isacharparenright}{\kern0pt}\ y{\isacharparenright}{\kern0pt}\ {\isasymotimes}\isactrlbsub i\isactrlesub \ {\isacharparenleft}{\kern0pt}{\isasympartial}\ j\ {\isasymalpha}\ x{\isacharparenright}{\kern0pt}\ {\isacharequal}{\kern0pt}\ {\isasymTheta}\ j\ {\isacharparenleft}{\kern0pt}k\ {\isacharminus}{\kern0pt}\ j{\isacharparenright}{\kern0pt}\ {\isacharparenleft}{\kern0pt}y\ {\isasymotimes}\isactrlbsub i\isactrlesub \ {\isacharparenleft}{\kern0pt}{\isasymSigma}\ j\ {\isacharparenleft}{\kern0pt}k\ {\isacharminus}{\kern0pt}\ j{\isacharparenright}{\kern0pt}\ {\isacharparenleft}{\kern0pt}{\isasympartial}\ j\ {\isasymalpha}\ x{\isacharparenright}{\kern0pt}{\isacharparenright}{\kern0pt}{\isacharparenright}{\kern0pt}{\isachardoublequoteclose}\isanewline
\ \ \ \ \ \ \ \ \ \ \ \ $\langle$\,proof\,$\rangle$\isanewline
\ \ \ \ \ \ \ \ \ \ \isacommand{thus}\isamarkupfalse%
\ {\isacharquery}{\kern0pt}thesis\isanewline
\ \ \ \ \ \ \ \ \ \ \ \ $\langle$\,proof\,$\rangle$\isanewline
\ \ \ \ \ \ \ \ \isacommand{qed}\isamarkupfalse%
\isacommand{{\isacharbraceright}{\kern0pt}}\isamarkupfalse%
\isanewline
\end{isabellebody}

\noindent
This shows that $x$ is $r_i$-invertible. We can
now conclude that $x$ is $r_i$-invertible in each direction $i$.

\begin{isabellebody}
\isanewline
\ \ \ \ \ \ \isacommand{thus}\isamarkupfalse%
\ {\isacharquery}{\kern0pt}thesis\isanewline
\ \ \ \ \ \ \ \ $\langle$\,proof\,$\rangle$\isanewline
\ \ \ \ \isacommand{qed}\isamarkupfalse%
\isacommand{{\isacharbraceright}{\kern0pt}}\isamarkupfalse%
\isanewline
\ \ \isacommand{thus}\isamarkupfalse%
\ {\isacharquery}{\kern0pt}case\isanewline
\ \ \ \ $\langle$\,proof\,$\rangle$\isanewline
\isacommand{qed}\isamarkupfalse%
\isanewline
\end{isabellebody}

\noindent
Every cell in a single-set cubical $(\omega,0)$-category has finite
dimension.  Lemma \isa{every-dim-k-ri-inv} therefore shows that every
cell is $r_i$-invertible in every direction $i$, which is
Proposition~\ref{P:OmegaZeroInv}.
\begin{isabellebody}
\isanewline
\isacommand{lemma}
\ every{\isacharunderscore}{\kern0pt}ri{\isacharunderscore}{\kern0pt}inv{\isacharcolon}{\kern0pt}\ {\isachardoublequoteopen}{\isasymexists}y{\isachardot}{\kern0pt}\ ri{\isacharunderscore}{\kern0pt}inv\ i\ x\ y{\isachardoublequoteclose}\isanewline
\ \ \isacommand{using}\isamarkupfalse%
\ every{\isacharunderscore}{\kern0pt}dim{\isacharunderscore}{\kern0pt}k{\isacharunderscore}{\kern0pt}ri{\isacharunderscore}{\kern0pt}inv\ local{\isachardot}{\kern0pt}fin{\isacharunderscore}{\kern0pt}fix\ \isacommand{by}\isamarkupfalse%
\ blast
\isanewline
\end{isabellebody}

Relative to the proof by hand, we had to prove several fixed points
and definedness conditions for compositions due the deep features of
our embedding. These are usually left implicit in proofs by hand or
with proper shallow embeddings, where they are discharged by type
inference. Except for such proof steps, the granularity of this formal
proof is similar to that of a proof by hand owing to Isabelle's proof
automation. Nevertheless, in some longer formulas we had to tell
Isabelle in detail how assumptions had to be matched with proof
goals. These details are visible in our proof document.

As already mentioned, formalising $(\omega,p)$-categories
with Isabelle for a finite $p$ is at least complicated with Isabelle's type
classes. Particular instances of $p$ can be given, but an arbitrary
$p$ would require more than one type parameter, and therefore locales.

\section{Conclusion}
\label{S:Conclusion}

We have introduced single-set axiomatisations of cubical categories with additional structure such as
connections and inverses.  We have justified their adequacy through
equivalence proofs relative to their classical counterparts. The
Isabelle/HOL proof assistant, with its powerful support for proof
automation and counterexample search, has been instrumental in this
development. We might not have undertaken this research
  without it. Cubical set and categories have a broad range of
applications in mathematics and computer science: from homotopy theory
and algebraic topology to homotopy type theory, concurrency theory and
rewriting theory. Our formalisation might therefore support innovative
applications of proof assistants in these fields.  In this regard, our
results allow us to outline several lines of research that we hope to
explore in the future.

First, our formalisation work yields an initial step towards formal
tools and methods that tame the combinatorial
complexity of proofs in higher categories.
 These are meant to support users in reasoning formally with
higher categories, based on geometric
intuitions if available, and with a high degree of automation.  For
this, a single-set approach seems relevant because of its algebraic
simplicity. As our axiomatisation is essentially
  single-sorted first-order, we expect it to work well with SMT solvers and similar
tools. Yet further experiments,
and in particular comparisons with formalisations using locales
  in Isabelle or proof assistants such as
Coq~\cite{BertotCasteran2004} and Lean~\cite{Lean}, are needed to
identify the most suitable approach. Dependent types, as supported by
Coq and Lean, might offer advantages in specifying and reasoning with
higher categories orthogonal to the proof automation supplied by
Isabelle.

Second, this work is part of a programme on proof support for
higher-dimensional rewriting. Higher globular and cubical categories
and higher globular and cubical polygraphs, which correspond
essentially to higher path categories organised in globes or cubes,
are particularly appropriate for this~\cite{Polybook2024}.  Cubical
categories, for instance, allow natural explicit formulations of
confluence results such as the Church-Rosser theorem, Newman's lemma
and their higher-dimensional extensions. Higher algebras such as
globular Kleene algebras and quantales have already been developed for
reasoning about such properties~\cite{CalkGoubaultMalbosStruth22,
  CalkMalbosPousStruth23} and the initial steps of the globular
approach have already been formalised with
Isabelle~\cite{CalkStruth24}. The construction of similar Kleene
algebras and quantales related to cubical categories and their
formalisation along similar lines seem equally desirable.

Third, one can construct polygraphic resolutions related to rewriting
properties in globular or cubical
categories~\cite{Lucas17,LucasPhD2017}.  In a companion article
\cite{MalbosMassacrierStruth24} we are using our single-set approach
to formalise normalisation strategies for higher abstract rewriting
systems, which provide a constructive approach to their polygraphic
resolutions in the cubical setting. Based on this, we aim to formalise
polygraphic constructions of higher-dimensional rewriting on
categories~\cite{GuiraudMalbos18,GuiraudMalbos12advances} and higher
algebras~\cite{GuiraudHoffbeckMalbos19,MalbosRen23} using proof
assistants and the mathematical components created for them. 

A fourth line of work might concern the categorical constructions from
Section~\ref{S:EquivalenceWithCubicalCategories}. Theorem~\ref{T:EquivSinSetClass},
in particular shows that the category $\SinCat{\omega}$ of single-set
cubical $\omega$-categories is equivalent to its classical counterpart
$\CubCat{\omega}$. This equivalence of categories lives in the
$2$-category of categories, where the natural isomorphisms
\ref{SSS:NaturalIsomorphism1} and \ref{SSS:NaturalIsomorphism2} are
$2$-cells. One may wonder how to formulate such an equivalence in a
single-set globular $2$-category.  This would constitute a
formalisation of the proof of justification of the single-set
axiomatisation within a single-set approach. The question arises in
the same way for the equivalences of
Theorems~\ref{T:EquivSinSetClassConn}
and~\ref{T:EquivSinSetClassConnInv}.

\subsubsection*{Acknowledgements}
The authors would like to thank Uli
Fahrenberg for a preliminary discussion about single-set cubical
2-categories. The third author would like to thank the CNRS for an
invited professorship at the \emph{Institut Camille Jordan} at Lyon, during
which part of this work has been conducted.

\begin{small}
\renewcommand{\refname}{\scshape\LARGE{References}}
\bibliographystyle{plain}
\bibliography{biblioFormalCubes}
\end{small}

\appendix

\section{Appendices}
\label{Section:Appendices}

\subsection{End of the proof of Lemma~\ref{L:FC-well-defined}}
\label{A:FC-well-defined}
To show that $\FC{\Scal}$ is a cubical $\omega$-category, we derive the remaining axioms:
\begin{enumerate}[{\bf (i)}]
\setcounter{enumi}{3}
\item if $x,y$ are $\star_{k,j}$-composable, by Axioms~\ref{SSS:DefCubSingleSetWithoutConn}\eqref{I:DeltaCompoCompat}, \eqref{I:AxiomSymCompatCompo} and other ones,
	\begin{itemize}
	\item if $i<j$ then
	\[
	\begin{array}{rcl}
	\partial_{k,i}^\alpha (x \star_{k,j} y) & = & s_{k-1} \dots s_i \delta_i^\alpha (x \circ_j y) \\
	 & = & s_{k-1} \dots s_{j-1} (s_{j-2} \dots s_i \delta_i^\alpha x \circ_j s_{j-2} \dots s_i \delta_i^\alpha y) \\
	 & = & s_{k-1} \dots s_j (s_{j-1} \dots s_i \delta_i^\alpha x \circ_{j-1} s_{j-1} \dots s_i \delta_i^\alpha y) \\
	 & = & s_{k-1} \dots s_i \delta_i^\alpha x \circ_{j-1} s_{k-1} \dots s_i \delta_i^\alpha y \\
	 & = & \partial_{k,i}^\alpha x \star_{k,j-1} \partial_{k,i}^\alpha y,
\end{array}
	\]
	\item if $i>j$ then
	\[
	\partial_{k,i}^\alpha (x \star_{k,j} y)  =  s_{k-1} \dots s_i \delta_i^\alpha (x \circ_j y) 
	 =  s_{k-1} \dots s_i \delta_i^\alpha x \circ_j s_{k-1} \dots s_i \delta_i^\alpha y 
	 =  \partial_{k,i}^\alpha x \star_{k,j} \partial_{k,i}^\alpha
         y,
       \]
	\item if $i=j$ then
	\[
	\partial_{k,i}^- (x \star_{k,i} y)  =  s_{k-1} \dots s_i \delta_i^- (x \circ_i y) 
	  =  s_{k-1} \dots s_i \delta_i^- x 
	  =  \partial_{k,i}^- x,
	\]
	\[
	\partial_{k,i}^+ (x \star_{k,i} y)  =  s_{k-1} \dots s_i \delta_i^+ (x \circ_i y) 
	 =  s_{k-1} \dots s_i \delta_i^+ y 
	  =  \partial_{k,i}^+ y,
	\]
	\end{itemize}
\item if $a,b$ are $\star_{k,i}$-composable, $c,d$ are $\star_{k,i}$-composable, $a,c$ are $\star_{k,j}$-composable, $b,d$ are $\star_{k,j}$-composable, and if $i \neq j$, then $\Delta_i(a,b)$, $\Delta_i(c,d)$, $\Delta_j(a,c)$ and $\Delta_j(b,d)$, so by Axiom~\ref{SSS:DefCubSingleSetWithoutConn}\eqref{I:ExchangeLaw} and other ones
\[
\begin{array}{rcl}
(a \star_{k,i} b) \star_{k,j} (c \star_{k,i} d) & = & (a \circ_i b) \circ_j (c \circ_i d) \\
 & = & (a \circ_j c) \circ_i (b \circ_j d) \\
 & = & (a \star_{k,j} c) \star_{k,i} (b \star_{k,j} d),
\end{array}
\]
\item
	\begin{itemize}
	\item if $i<j$ then by Axioms~\ref{SSS:DefCubSingleSetWithoutConn}\eqref{I:AxiomSymType}, \eqref{I:AxiomSymInv}, \eqref{I:AxiomFaceSym}, \eqref{I:AxiomSymFix} and other ones
	\[
	\begin{array}{rcl}
	\partial_{k,i}^\alpha \epsilon_{k,j} & = & s_{k-1} \dots s_i \delta_i^\alpha \tilde{s}_j \dots \tilde{s}_{k-1} \\
	 & = & s_{k-1} \dots s_i \tilde{s}_j \dots \tilde{s}_{k-1} \delta_i^\alpha \\
	 & = & s_{k-1} \dots s_j s_{j-2} \dots s_i \tilde{s}_j \dots \tilde{s}_{k-1} \delta_i^\alpha \\
	 & = & s_{j-2} \dots s_i s_{k-1} \dots s_j \tilde{s}_j \dots \tilde{s}_{k-1} \delta_i^\alpha \\
	 & = & s_{j-2} \dots s_i \delta_i^\alpha \\
	 & = & \epsilon_{k-1,j-1} \partial_{k-1,i}^\alpha	,
	\end{array}
	\]
      \item if $i>j$ then similarly
        $ \partial_{k,i}^\alpha \epsilon_{k,j} = s_{k-1} \dots s_i
        \delta_i^\alpha \tilde{s}_j \dots \tilde{s}_{k-1} = s_{j-1}
        \dots s_{i-1} \delta_{i-1}^\alpha = \epsilon_{k-1,j}
        \partial_{k-1,i-1}^\alpha , $
      \item if $i=j$ then
        $\partial_{k,i}^\alpha \epsilon_{k,i} = s_{k-1} \dots s_i
        \delta_i^\alpha \tilde{s}_i \dots \tilde{s}_{k-1} = s_{k-1}
        \dots s_i \tilde{s}_i \dots \tilde{s}_{k-1} = \id, $
	\end{itemize}
\item if $a,b$ are $\star_{k,j}$-composable,
	\begin{itemize}
	\item if $i \leq j$ then
	\[
	\begin{array}{rcl}
	\epsilon_{k+1,i} (a \star_{k,j} b) & = & \tilde{s}_i \dots \tilde{s}_k (a \circ_j b) \\
	 & = & \tilde{s}_i \dots \tilde{s}_j (\tilde{s}_{j+1} \dots \tilde{s}_k a \circ_j \tilde{s}_{j+1} \dots \tilde{s}_k b) \\
	 & = & \tilde{s}_i \dots \tilde{s}_{j-1} (\tilde{s}_j \dots \tilde{s}_k a \circ_{j+1} \tilde{s}_j \dots \tilde{s}_k b) \\
	 & = & \tilde{s}_i \dots \tilde{s}_k a \circ_{j+1} \tilde{s}_i \dots \tilde{s}_k b \\
	 & = & \epsilon_{k+1,i} a \star_{k+1,j+1} \epsilon_{k+1,i} b,
	\end{array}
	\]
	\item if $i>j$ then
	\[
	\epsilon_{k+1,i} (a \star_{k,j} b) = \tilde{s}_i \dots
        \tilde{s}_k (a \circ_j b)
        = \tilde{s}_i \dots \tilde{s}_k a \circ_j \tilde{s}_i \dots \tilde{s}_k b 
	 = \epsilon_{k+1,i} a \star_{k+1,j} \epsilon_{k+1,i} b,
	\]
	\end{itemize}
\item if $i \leq j$ then by Axioms~\ref{SSS:DefCubSingleSetWithoutConn}\eqref{I:AxiomSymFix}, \eqref{I:AxiomSymBraid} and other ones
\[
\begin{array}{rcl}
\epsilon_{k+1,j+1} \epsilon_{k,i} & = & \tilde{s}_{j+1} \dots \tilde{s}_k \tilde{s}_i \dots \tilde{s}_{k-1} \\
 & = & \tilde{s}_i \dots \tilde{s}_{j-1} \tilde{s}_{j+1} \dots \tilde{s}_k \tilde{s}_j \dots \tilde{s}_{k-1} \\
 & = & \tilde{s}_i \dots \tilde{s}_k \tilde{s}_j \dots \tilde{s}_{k-1} \\
 & = & \epsilon_{k+1,i} \epsilon_{k,j}.
\end{array}
\]
\end{enumerate}

\subsection{End of the proof of Lemma~\ref{L:FS-well-defined}}
\label{A:FS-well-defined}
To show that $\FS{\Ccal}$ is a single-set cubical $\omega$-category, we derive the remaining axioms:
\begin{enumerate}[{\bf (i)}]
\setcounter{enumi}{3}
\item if $\Delta_j(x,y)$,
	\begin{itemize}
	\item if $i<j$ then
	\[
	\begin{array}{rcl}
	\delta_i^\alpha (x \circ_j y) & = & \epsilon_{n,i} \partial_{n,i}^\alpha (x \star_{n,j} y) \\
	 & = & \epsilon_{n,i} (\partial_{n,i}^\alpha x \star_{n-1,j-1} \partial_{n,i}^\alpha y) \\
	 & = & \epsilon_{n,i} \partial_{n,i}^\alpha x \star_{n,j} \epsilon_{n,i} \partial_{n,i}^\alpha y) \\
	 & = & \delta_i^\alpha x \circ_j \delta_i^\alpha y,
	\end{array}
	\]
	\item if $i>j$ then
	\[
	\begin{array}{rcl}
	\delta_i^\alpha (x \circ_j y) & = & \epsilon_{n,i} \partial_{n,i}^\alpha (x \star_{n,j} y) \\
	 & = & \epsilon_{n,i} (\partial_{n,i}^\alpha x \star_{n-1,j} \partial_{n,i}^\alpha y) \\
	 & = & \epsilon_{n,i} \partial_{n,i}^\alpha x \star_{n,j} \epsilon_{n,i} \partial_{n,i}^\alpha y) \\
	 & = & \delta_i^\alpha x \circ_j \delta_i^\alpha y,
	\end{array}
	\]
	\end{itemize}
\item exchange law: if $i \neq j$, $\Delta_i(w,x)$, $\Delta_i(y,z)$, $\Delta_j(w,y)$ and $\Delta_j(x,z)$ then
\[
(w \circ_i x) \circ_j (y \circ_i d)  = (w \star_{n,i} x) \star_{n,j} (y \star_{n,i} z) 
 = (w \star_{n,j} y) \star_{n,i} (x \star_{n,j} z) 
 = (w \circ_j y) \circ_i (x \circ_j z),
\]
\item
	\begin{itemize}
	\item if $x \in \Scal^i$ then
          $\delta_{i+1}^- s_i x = \epsilon_{n,i+1} \partial_{n,i+1}^-
          \epsilon_{n,i+1} \partial_{n,i}^- x = s_i x$, so
          $s_i x \in \Scal^{i+1}$,
	\item if $y \in \Scal^{i+1}$ then
          $\delta_i^- \tilde{s}_i y = \epsilon_{n,i} \partial_{n,i}^-
          \epsilon_{n,i} \partial_{n,i+1}^- y = \tilde{s}_i y$, so
          $\tilde{s}_i y \in \Scal^i$,
          \end{itemize}
\item
	\begin{itemize}
	\item if $x \in \Scal^i$ then
$\tilde{s}_i s_i x = \epsilon_{n,i} \partial_{n,i+1}^- \epsilon_{n,i+1} \partial_{n,i}^- x 
	 = \delta_i^- x 
	 = x$, 
	\item if $y \in \Scal^{i+1}$ then
	$s_i \tilde{s}_i y = \epsilon_{n,i+1} \partial_{n,i}^- \epsilon_{n,i} \partial_{n,i+1}^- y 
	 = \delta_{i+1}^- y 
	 = y$,
	\end{itemize}
\item if $x \in \Scal^j$,
	\begin{itemize}
	\item then
	\begin{align*}
	\delta_j^\alpha s_j x = \epsilon_{n,j} \partial_{n,j}^\alpha \epsilon_{n,j+1} \partial_{n,j}^- x 
	= \epsilon_{n,j+1} \epsilon_{n-1,j} \partial_{n-1,j}^\alpha \partial_{n,j+1}^- x 
	= \epsilon_{n,j} \epsilon_{n-1,j} \partial_{n-1,j}^\alpha \partial_{n,j}^- x 
	= s_j \delta_{j+1}^\alpha x,
	\end{align*}
	\item if $i<j$ then
	\begin{align*}
	\delta_i^\alpha s_j x = \epsilon_{n,i} \partial_{n,i}^\alpha \epsilon_{n,j+1} \partial_{n,j}^- x 
	= \epsilon_{n,i} \epsilon_{n-1,j} \partial_{n-1,i}^\alpha \partial_{n,j}^\beta x 
	= \epsilon_{n,j+1} \epsilon_{n-1,i} \partial_{n-1,j-1}^\alpha \partial_{n,i}^\beta x 
	= s_j \delta_i^\alpha x,
	\end{align*}
	\item $i>j+1$ then
	\begin{align*}
	\delta_i^\alpha s_j x = \epsilon_{n,i} \partial_{n,i}^\alpha \epsilon_{n,j+1} \partial_{n,j}^- x 
	= \epsilon_{n,i} \epsilon_{n-1,j+1} \partial_{n-1,i-1}^\alpha \partial_{n,j}^\beta x 
	= \epsilon_{n,j+1} \epsilon_{n-1,i-1} \partial_{n-1,j}^\alpha \partial_{n,i}^\beta x 
	= s_j \delta_i^\alpha x,
	\end{align*}
	\end{itemize}
\item if $x,y \in \Scal^i$ and $\Delta_j(x,y)$,
	\begin{itemize}
	\item if $j=i+1$ then
	\begin{align*}
	s_i (x \circ_{i+1} y) &= \epsilon_{n,i+1} \partial_{n,i}^- (x \star_{n,i+1} y) \\
	&= \epsilon_{n,i+1} (\partial_{n,i}^- x \star_{n,i} \partial_{n,i}^- y) \\
	&= \epsilon_{n,i+1} \partial_{n,i}^- x \star_{n,i} \epsilon_{n,i+1} \partial_{n,i}^- y \\
	&= s_i x \circ_i s_i y
	\end{align*}
	\item if $j<i$ then
	\begin{align*}
	s_i (x \circ_j y) &= \epsilon_{n,i+1} \partial_{n,i}^- (x \star_{n,j} y) \\
	&= \epsilon_{n,i+1} (\partial_{n,i}^- x \star_{n,j} \partial_{n,i}^- y) \\
	&= \epsilon_{n,i+1} \partial_{n,i}^- x \star_{n,j} \epsilon_{n,i+1} \partial_{n,i}^- y \\
	&= s_i x \circ_j s_i y
	\end{align*}
	\item if $j>i+1$ then
	\begin{align*}
	s_i (x \circ_j y) &= \epsilon_{n,i+1} \partial_{n,i}^- (x \star_{n,j} y) \\
	&= \epsilon_{n,i+1} (\partial_{n,i}^- x \star_{n,j-1} \partial_{n,i}^- y) \\
	&= \epsilon_{n,i+1} \partial_{n,i}^- x \star_{n,j} \epsilon_{n,i+1} \partial_{n,i}^- y \\
	&= s_i x \circ_j s_i y
	\end{align*}
	\end{itemize}
\item if $x \in \Scal^{i,i+1}$ then
\begin{align*}
s_i x = s_i \delta_i^- \delta_{i+1}^- x 
 = \epsilon_{n,i+1} \partial_{n,i}^- \epsilon_{n,i+1} \partial_{n,i+1}^- x 
 = \epsilon_{n,i+1} \partial_{n,i+1}^- \epsilon_{n,i} \partial_{n,i}^- x 
 = \delta_{i+1}^- \delta_i^- x 
 = x,
\end{align*}
\item if $x \in \Scal^{i,j}$,
	\begin{itemize}
	\item if $i<j-1$ then
	\begin{align*}
	s_i s_j x = \epsilon_{n,i+1} \partial_{n,i}^- \epsilon_{n,j+1} \partial_{n,j}^- x 
	= \epsilon_{n,i+1} \epsilon_{n,j} \partial_{n,i}^- \partial_{n,j}^- x 
	= \epsilon_{n,j+1} \epsilon_{n-1,i+1} \partial_{n-1,j-1}^- \partial_{n,i}^- x 
	= s_j s_i x,
	\end{align*}
	\item if $i>j+1$ then
	\begin{align*}
	s_i s_j x = \epsilon_{n,i+1} \partial_{n,i}^- \epsilon_{n,j+1} \partial_{n,j}^- x 
	= \epsilon_{n,i+1} \epsilon_{n,j+1} \partial_{n,i-1}^- \partial_{n,j}^- x 
	= \epsilon_{n,j+1} \epsilon_{n-1,i} \partial_{n-1,j}^- \partial_{n,i}^- x 
	= s_j s_i x,
	\end{align*}
	\end{itemize}
\item each $x \in \Scal$ has a representative $a \in \Ccal_n$ for some $n \in \Nbb$, so let $i \geq n+1$, then $a' = \epsilon_{i,i} \dots \epsilon_{n+1,n+1} a \in \Ccal_i$ represents $x$, so by definition
\begin{align*}
\delta_i^- x = cl_\sim (\epsilon_{i,i} \partial_{i,i}^- a') 
= cl_\sim (\epsilon_{i,i} \partial_{i,i}^- \epsilon_{i,i} \dots \epsilon_{n+1,n+1} a) 
= cl_\sim (\epsilon_{i,i} \dots \epsilon_{n+1,n+1} a) 
= x,
\end{align*}
so $x \in \Scal^{>n}$.
\end{enumerate}

\subsection{End of the proof of Lemma~\ref{L:FCG-well-defined}}
\label{A:FCG-well-defined}
To show that $\FCG{\Scal}$ is a cubical $\omega$-category with connections, we prove the remaining axioms:
\begin{enumerate}[{\bf (i)}]
\item
	\begin{itemize}
	\item $\partial_{k,i+1}^\alpha \Gamma_{k,i}^\alpha
	= s_{k-1} \dots s_{i+1} \delta_{i+1}^\alpha \gamma_i^\alpha \tilde{s}_i \dots \tilde{s}_{k-1}
	= \id$,
	\item
$\partial_{k,i}^\alpha \Gamma_{k,i}^{-\alpha}
	= s_{k-1} \dots s_i \delta_{i+1}^\alpha \tilde{s}_i \dots \tilde{s}_{k-1} 
	= \tilde{s}_i \dots \tilde{s}_{k-2} s_{k-2} \dots s_i \delta_i^\alpha 
	= \epsilon_{k-1,i} \partial_{k-1,i}^\alpha$,

\item 
	$\partial_{k,i+1}^\alpha \Gamma_{k,i}^{-\alpha}
	= s_{k-1} \dots s_{i+1} \delta_{i+1}^\alpha \tilde{s}_i \dots \tilde{s}_{k-1} 
	= \tilde{s}_i \dots \tilde{s}_{k-2} s_{k-2} \dots s_i \delta_i^\alpha 
	= \epsilon_{k-1,i} \partial_{k-1,i}^\alpha$,

	\item if $i > j+1$ then
	\begin{align*}
	\partial_{k,i}^\alpha \Gamma_{k,j}^\beta
	= \gamma_j^\beta s_{k-1} \dots s_i \tilde{s}_j \dots \delta_i^\alpha \tilde{s}_{i-1} \tilde{s}_i \dots \tilde{s}_{k-1} 
	= \gamma_j^\beta \tilde{s}_j \dots \tilde{s}_{i-2} \delta_{i-1}^\alpha 
	= \Gamma_{k-1,j}^\beta \partial_{k-1,i-1}^\alpha,
	\end{align*}
	\end{itemize}
\item if $a, b$ are $\star_{k,j}$-composable,
	\begin{itemize}
	\item $i = j$ then
	\begin{align*}
	\Gamma_{k+1,i}^+ (a \star_{k,i} b)
	&= \gamma_i^+ (\tilde{s}_i \dots \tilde{s}_k a \circ_{i+1} \tilde{s}_i \dots \tilde{s}_k b) \\
	&= (\gamma_i^+ \tilde{s}_i \dots \tilde{s}_k a \circ_i \tilde{s}_i \dots \tilde{s}_k a) \circ_{i+1} (s_i \tilde{s}_i \dots \tilde{s}_k a \circ_i \gamma_i^+ \tilde{s}_i \dots \tilde{s}_k b) \\
	&= (\Gamma_{k+1,i}^+ a \star_{k+1,i} \epsilon_{k+1,i} a) \star_{k+1,i+1} (\epsilon_{k+1,i+1} a \star_{k+1,i} \Gamma_{k+1,i}^+ b),
	\end{align*}
      \item if $i > j$ then
        $ \Gamma_{k+1,i}^\alpha (a \star_{k,j} b) = \gamma_i^\alpha
        (\tilde{s}_i \dots \tilde{s}_k a \circ_j \tilde{s}_i \dots
        \tilde{s}_k b) = \Gamma_{k+1,i}^\alpha a \star_{k+1,j}
        \Gamma_{k+1,i}^\alpha b, $
	\end{itemize}
\item
	\begin{itemize}
	\item $\Gamma_{k,i}^+ a \star_{k,i} \Gamma_{k,i}^- a = \gamma_i^+ \tilde{s}_i \dots \tilde{s}_{k-1} a \circ_i \gamma_i^- \tilde{s}_i \dots \tilde{s}_{k-1} a = \epsilon_{k,i+1} a$ by Axiom~\ref{SSS:DefCubSingleSetWithConn}\eqref{I:AxiomConnZigzag} and other ones,
	\item $\Gamma_{k,i}^+ a \star_{k,i+1} \Gamma_{k,i}^- a = \gamma_i^+ \tilde{s}_i \dots \tilde{s}_{k-1} a \circ_{i+1} \gamma_i^- \tilde{s}_i \dots \tilde{s}_{k-1} a = \epsilon_{k,i} a$ by Axiom~\ref{SSS:DefCubSingleSetWithConn}\eqref{I:AxiomConnZigzag} and other ones,
	\end{itemize}
\item
	\begin{itemize}
	\item by
          Axiom~\ref{SSS:DefCubSingleSetWithConn}\eqref{I:AxiomConnStab}
          and other ones 
	$\Gamma_{k+1,i}^\alpha \epsilon_{k,i}
	= \gamma_i^\alpha \tilde{s}_i \dots \tilde{s}_k \tilde{s}_i \dots \tilde{s}_{k-1} 
	= \epsilon_{k+1,i} \epsilon_{k,i}$,
	\item if $i > j$ then
	\begin{align*}
	\Gamma_{k+1,i}^\alpha \epsilon_{k,j}
	= \tilde{s}_j \dots \tilde{s}_{i-2} \gamma_i^\alpha \tilde{s}_i \dots \tilde{s}_k \tilde{s}_{i-1} \dots \tilde{s}_{k-1} 
	= \tilde{s}_j \dots \tilde{s}_i \gamma_{i-1}^\alpha \tilde{s}_{i+1} \dots \tilde{s}_k \tilde{s}_{i-1} \dots \tilde{s}_{k-1} 
	= \epsilon_{k+1,j} \Gamma_{k,i-1}^\alpha,
	\end{align*}
	\end{itemize}
\item
	\begin{itemize}
	\item if $i < j$ then using Axioms~\ref{SSS:DefCubSingleSetWithConn}\eqref{I:AxiomConnBraid}, \eqref{I:AxiomConnShift} and other ones
	\begin{align*}
	\Gamma_{k+1,i}^\alpha \Gamma_{k,j}^\beta
	&= \gamma_i^\alpha \tilde{s}_i \dots \tilde{s}_j \tilde{s}_{j+1} \gamma_j^\beta \tilde{s}_{j+2} \dots \tilde{s}_k \tilde{s}_j \dots \tilde{s}_{k-1} \\
	&= \gamma_i^\alpha \tilde{s}_i \dots \tilde{s}_{j-1} \gamma_{j+1}^\beta \tilde{s}_{j+1} \dots \tilde{s}_k \tilde{s}_j \dots \tilde{s}_{k-1} \\
	&= \Gamma_{k+1,j+1}^\beta \Gamma_{k,i}^\alpha,
	\end{align*}
	\item
	$
	\Gamma_{k+1,i}^\alpha \Gamma_{k,i}^\alpha
	= \tilde{s}_i \tilde{s}_{i+1} \gamma_i^\alpha \tilde{s}_{i+2} \dots \tilde{s}_k \gamma_i^\alpha \tilde{s}_i \dots \tilde{s}_{k-1} 
	= \gamma_{i+1}^\alpha \tilde{s}_{i+1} \dots \tilde{s}_k \gamma_i^\alpha \tilde{s}_i \dots \tilde{s}_{k-1} 
	= \Gamma_{k+1,i+1}^\alpha \Gamma_{k,i}^\alpha$. 
	\end{itemize}
\end{enumerate}

\subsection{End of the proof of Lemma~\ref{L:FSG-well-defined}}
\label{A:FSG-well-defined}
To show that $\FSG{\Ccal}$ is a single-set cubical $\omega$-category with connections, we prove the remaining axioms:
\begin{enumerate}[{\bf (i)}]
\item if $i \neq j,j+1$ and $x \in \Scal^j$, then $\partial_{n,j}^- x= \partial_{n,j}^+ x$ so
	\begin{itemize}
	\item
	$\delta_{j+1}^\alpha \gamma_j^\alpha x
	= \epsilon_{n,j+1} \partial_{n,j+1}^\alpha \Gamma_{n,j}^\alpha \partial_{n,j}^\alpha x
	= \epsilon_{n,j+1} \partial_{n,j}^\alpha x
	= s_j x$,
	\item if $i<j$ then $\delta_i^\alpha \gamma_j^\beta x
	= \epsilon_{n,i} \Gamma_{n-1,j-1}^\beta \partial_{n-1,i}^\alpha \partial_{n,j}^\beta x
	= \Gamma_{n,j}^\beta \epsilon_{n-1,i} \partial_{n-1,j-1}^\beta \partial_{n,i}^\alpha x
	= \gamma_j^\beta \delta_i^\alpha x$,
	\item if $i>j+1$ then $\delta_i^\alpha \gamma_j^\beta x
	= \epsilon_{n,i} \Gamma_{n-1,j}^\beta \partial_{n-1,i-1}^\alpha \partial_{n,j}^\beta x
	= \Gamma_{n,j}^\beta \epsilon_{n-1,i-1} \partial_{n-1,j}^\beta \partial_{n,i}^\alpha x
	= \gamma_j^\beta \delta_i^\alpha x$,
	\end{itemize}
\item if $j \neq i,i+1$, $x,y \in \Scal^i$,
	\begin{itemize}
	\item if $\Delta_{i+1}(x,y)$ then
	\begin{align*}
	\gamma_i^- (x \circ_{i+1} y)
	&= \Gamma_{n,i}^- (\partial_{n,i}^- x \star_{n,i} \partial_{n,i}^- y) \\
	&= (\Gamma_{n,i}^- \partial_{n,i}^- x \star_{n+1,i+1} \epsilon_{n,i} \partial_{n,i}^- y) \star_{n+1,i} (\epsilon_{n,i+1} \partial_{n,i}^- y \star_{n+1,i+1} \Gamma_{n,i}^- \partial_{n,i}^- y) \\
	&= (\gamma_i^- x \circ_{i+1} y) \circ_i (s_i y \circ_{i+1} \gamma_i^- y),
	\end{align*}
	\item if $j < i$ and $\Delta_j(x,y)$ then $\gamma_i^\alpha (x \circ_j y) = \Gamma_{n,i}^\alpha (\partial_{n,i}^\alpha x \star_{n-1,j} \partial_{n,i}^\alpha y) = \gamma_i^\alpha x \circ_j \gamma_i^\alpha y$,
	\item if $j > i+1$ and $\Delta_j(x,y)$ then $\gamma_i^\alpha (x \circ_j y) = \Gamma_{n,i}^\alpha (\partial_{n,i}^\alpha x \star_{n-1,j-1} \partial_{n,i}^\alpha y) = \gamma_i^\alpha x \circ_j \gamma_i^\alpha y$,
	\end{itemize}
\item if $x \in \Scal^{i,i+1}$ then $x = \epsilon_{n,i} \epsilon_{n-1,i} \partial_{n-1,i}^- \partial_{n,i+1}^- x$ so
\begin{align*}
\gamma_i^\alpha x
= \Gamma_{n,i}^\alpha \epsilon_{n-1,i} \partial_{n-1,i}^- \partial_{n,i+1}^- x
= \epsilon_{n,i} \epsilon_{n-1,i} \partial_{n-1,i}^- \partial_{n,i+1}^- x
= x,
\end{align*}
\item if $x \in \Scal^i$, then $\partial_{n,i}^+ x = \partial_{n,i}^- x$ so
	\begin{itemize}
	\item $\gamma_i^+ x \circ_{i+1} \gamma_i^- x
	= \Gamma_{n,i}^+ \partial_{n,i}^+ x \star_{n,i+1} \Gamma_{n,i}^- \partial_{n,i}^- x
	= \epsilon_{n,i} \partial_{n,i}^+ x
	= x$,
	\item $\gamma_i^+ x \circ_i \gamma_i^- x
	= \Gamma_{n,i}^+ \partial_{n,i}^+ x \star_{n,i} \Gamma_{n,i}^- \partial_{n,i}^- x
	= \epsilon_{n,i+1} \partial_{n,i}^+ x
	= s_i x$,
	\end{itemize}
\item if $x \in \Scal^{i,j}$ and $i > j+1$ then
	\[
	\gamma_i^\alpha \gamma_j^\beta x
	= \Gamma_{n,i}^\alpha \Gamma_{n-1,j}^\beta \partial_{n-1,i-1}^\alpha \partial_{n,j}^\beta x
	= \Gamma_{n,j}^\beta \Gamma_{n-1,i-1}^\alpha \partial_{n-1,j}^\beta \partial_{n,i}^\alpha x
	= \gamma_j^\beta \gamma_i^\alpha x.
	\]
\end{enumerate}

\clearpage

\quad

\vfill

\begin{footnotesize}

\bigskip
\auteur{Philippe Malbos}{malbos@math.univ-lyon1.fr}
{Universit\'e Claude Bernard Lyon 1\\
ICJ UMR5208, CNRS\\
F-69622 Villeurbanne cedex, France}

\bigskip
\auteur{Tanguy Massacrier}{massacrier@math.univ-lyon1.fr}
{Universit\'e Claude Bernard Lyon 1\\
ICJ UMR5208, CNRS\\
69622 Villeurbanne cedex, France}

\bigskip
\auteur{Georg Struth}{g.struth@sheffield.ac.uk}
{University of Sheffield\\
  Department of Computer Science\\
  Regent Court, 211 Portobello\\
  Sheffield S1 4DP, United Kingdom
}
\end{footnotesize}

\vspace{1.5cm}

\begin{small}---\;\;\today\;\;-\;\;\hhmm\;\;---\end{small}
\end{document}